\documentclass[10pt,reqno]{amsart}

\usepackage[T1]{fontenc} 
\usepackage[utf8]{inputenc} 
\usepackage{longtable}
\usepackage{amssymb}
\usepackage{amsthm}
\usepackage{amsmath}
\usepackage{wasysym}
\usepackage{dcpic, pictex}
\usepackage{bbm}
\usepackage{enumerate}
\usepackage{amsfonts}
\usepackage{amsthm}
\usepackage{amsmath}
\usepackage{times}
\usepackage{amscd}
\usepackage[mathscr]{eucal}
\usepackage{indentfirst}
\usepackage{pict2e}
\usepackage{epic}
\usepackage{epstopdf} 
\usepackage{verbatim}
\usepackage{cancel}
\usepackage[foot]{amsaddr}
\usepackage{lipsum}
\usepackage{mathrsfs}
\usepackage{bm}

\usepackage{
amsfonts, 
latexsym, 
enumerate,
}
\usepackage[english]{babel}
\allowdisplaybreaks
\makeindex
\newtheorem{theorem}{Theorem}[section]
\newtheorem{lemma}[theorem]{Lemma}
\newtheorem{proposition}[theorem]{Proposition}
\newtheorem{definition}[theorem]{Definition}

\newtheorem{remark}[theorem]{Remark}
\newtheorem{corollary}[theorem]{Corollary}
\newtheorem{assumption}[theorem]{Assumption}
\usepackage[a4paper,top=3.5cm,bottom=3.5cm,left=3cm,right=3cm]{geometry}
\numberwithin{equation}{section}
\allowdisplaybreaks

\usepackage{ bbold }
\usepackage{xcolor}

\newcommand{\nc}{\normalcolor}
\newcommand{\dif}{\mathrm{d}}
\newcommand{\E}{\mathbf{E}}
\newcommand{\R}{\mathbf{R}}
\newcommand{\C}{\mathbf{C}}

\newcommand{\Av}{\mathrm{Av}}
\newcommand{\N}{\mathbf{N}}

\newcommand{\Var}{\mathrm{Var}}

\newcommand{\ii}{\mathrm{i}}

\newcommand{\bw}{{\bm u}}
\newcommand{\bq}{{\bm q}}

\usepackage{hyperref}

\usepackage{graphicx}

\title{Gaussian fluctuations in the Equipartition Principle for Wigner matrices}
\subjclass[2020]{60B20, 82B10, 58J51} 
\keywords{Quantum unique ergodicity, Matrix Dyson equation, Local law}
\date{\today}

\begin{document}
	\maketitle
		\vspace{0.25cm}
	
	\renewcommand{\thefootnote}{\fnsymbol{footnote}}	
	\noindent
	\mbox{}%
	\hfill%
		\begin{minipage}{0.21\textwidth}
		\centering
		{Giorgio Cipolloni}\footnotemark[1]\\
		\footnotesize{\textit{gc4233@princeton.edu}}
	\end{minipage}
	\hfill%
	\begin{minipage}{0.21\textwidth}
		\centering
		{L\'aszl\'o Erd\H{o}s}\footnotemark[2]\\
		\footnotesize{\textit{lerdos@ist.ac.at}}
	\end{minipage}
	\hfill%
	\begin{minipage}{0.21\textwidth}
		\centering
		{Joscha Henheik}\footnotemark[2]\\
		\footnotesize{\textit{joscha.henheik@ist.ac.at}}
	\end{minipage}
	\hfill%
	\begin{minipage}{0.21\textwidth}
		\centering
		{Oleksii Kolupaiev}\footnotemark[2]\\
		\footnotesize{\textit{oleksii.kolupaiev@ist.ac.at}}
	\end{minipage}
	\hfill%
	\mbox{}%
	\footnotetext[1]{Princeton Center for Theoretical Science, Princeton University, Princeton, NJ 08544, USA.}
	\footnotetext[2]{Institute of Science and Technology Austria, Am Campus 1, 3400 Klosterneuburg, Austria.}
			\footnotetext[2]{Supported by ERC Advanced Grant ``RMTBeyond'' No.~101020331.}
	
	\renewcommand*{\thefootnote}{\arabic{footnote}}
	\vspace{0.25cm}
\begin{abstract}
The total energy of an eigenstate in a composite quantum system tends
to be distributed equally among its constituents. 
 We identify  the quantum fluctuation
around this equipartition principle in the simplest disordered quantum system
consisting of linear combinations of Wigner matrices. As our main ingredient, we prove the Eigenstate
Thermalisation Hypothesis and Gaussian fluctuation for general quadratic forms of the bulk eigenvectors
of  Wigner matrices with an arbitrary  deformation.  
\end{abstract}

\section{Introduction}

The general principle of the equipartition of  energy for a classical  ergodic system
asserts that in equilibrium the total energy  is equally distributed among all elementary degrees of freedom.
A similar principle   for the kinetic energy has recently been verified 
for a general quantum system coupled to a heat bath~\cite{Luczka}, see
also the previous works on the free Brownian particle and
a dissipative harmonic oscillator in~\cite{BSL2018, SBL2018, BialasSpiechowiczLuczka} and extensive literature therein. 
Motivated by E.~Wigner's original vision to model any sufficiently complex quantum
system by random matrices,  a particularly strong \emph{microcanonical}
version of 
the equipartition principle for Wigner matrices was first formulated
and proven in~\cite{Equipartition}. 
In its simplest form, 
consider
a fully mean field random Hamilton operator $H$ acting on the 
high  dimensional quantum state space $\C^N$ that
consists of the sum of two independent $N\times N$ Wigner 
matrices, 
$$
H = W_1+W_2\,,
$$
as two constituents of the system. Recall that Wigner matrices $W= (w_{ij})$
	are real or complex Hermitian random matrices with independent (up to the symmetry constraint $w_{ij}=\bar w_{ji}$),
	identically distributed entries. 
 Let $\bm{u}$ be a normalised eigenvector of $H$ with  eigenvalue (energy) 
$\lambda=\langle \bm{u}, H\bm{u} \rangle $, 
then equipartition asserts that  
$\langle \bm{u}, W_l\bm{u} \rangle\approx \frac{1}{2}\lambda$
for $l=1,2$. In~\cite[Theorem 3.4]{Equipartition} even a precise  error bound was proven, i.e. that
\begin{equation}\label{old}
	\Big|  \langle \bm{u}, W_l\bm{u} \rangle - \frac{1}{2}\lambda \Big| \le \frac{N^\epsilon}{\sqrt{N}}\,, \qquad l=1,2\,,
\end{equation}
holds  with very high probability for any fixed $\epsilon>0$; this estimate is optimal up the $N^\epsilon$ factor.
The main result of the current paper is to identify the fluctuation in~\eqref{old}, more precisely we
will show that
$$
\sqrt{N} \Big[  \langle \bw, W_l\bw \rangle - \frac{1}{2}\lambda \Big]
$$
converges to a centred normal distribution as $N\to\infty$.
 We also compute its variance that turns out to be independent of the energy $\lambda$
but depends on the symmetry class (real or complex).
The result can easily be extended to the case when $H$ is a more general linear combination of several independent Wigner matrices.

The estimate~\eqref{old} is reminiscent to the recently proven \emph{Eigenstate Thermalisation Hypothesis (ETH)},
also known as the \emph{Quantum Unique Ergodicity (QUE)},\footnote{ETH for Wigner matrices was first conjectured
	by Deutsch~\cite{deutsch}. Quantum ergodicity has
	a long history in the context of the  quantisations of chaotic classical dynamical systems starting from the fundamental theorem
	by \u{S}nirel’man~\cite{Shnirelman}. 
	For more background and
	related literature, see the Introduction of~\cite{ETHpaper}.} for Wigner matrices in~\cite[Theorem 2.2]{ETHpaper}
(see also~\cite[Theorem 2.6]{A2} for an improvement) which asserts that
\begin{equation}\label{ethold}
	\Big| \langle \bm{u}, A\bm{u}\rangle - \langle A\rangle \Big| \le \frac{N^\epsilon}{\sqrt{N}}\,, \qquad \langle A\rangle:= \frac{1}{N} \mathrm{Tr}  A\,,
\end{equation} 
holds for any bounded \emph{deterministic} matrix $A$. In fact, 
even the Gaussian fluctuation of  
\begin{equation}\label{quefluc}
	\sqrt{N} \Big[ \langle \bm{u}, A\bm{u}\rangle - \langle A\rangle \Big] 
\end{equation} 
was proven in~\cite[Theorem 2.2]{normalfluc} and \cite[Theorem 2.8]{A2}, see also~\cite{BenigniLopatto2103.12013} and the recent generalisation~\cite{CipolloniBenigni2022}
to off-diagonal elements as well as joint Gaussianity of several quadratic forms.
Earlier results on ETH~\cite{ESYDeloc, KnowlesYinIsotropic, BloemendalEKYY, BenigniLopattoDelocalization} and its 
fluctuations~\cite{BourgadeYau1312.1301, KnowlesYin13, TVEigenvector, MarcinekYau2005.08425}
for Wigner matrices  typically concerned rank one or finite rank observables $A$.

Despite their apparent similarities,
 the quadratic form in~\eqref{old} is essentially different from that in~\eqref{ethold} since $W_l$ is strongly correlated with $\bw$
while $A$ in~\eqref{ethold} is deterministic. This explains the difference in the two leading terms; note that
$\frac{1}{2}\lambda$ in~\eqref{old} is energy dependent and it is far from the value $\langle W_l\rangle \approx 0$, which might  be erroneously guessed from~\eqref{ethold}. Still, the basic approach
leading to ETH~\eqref{ethold}  is very useful to study   $\langle \bm{u}, W_l\bm{u} \rangle$ as well. 
The basic idea is  to  condition on one of the Wigner matrices, say $W_2$, and consider $H=W_1+W_2$ in the
probability space of $W_1$ as a Wigner matrix with an additive \emph{deterministic} deformation $W_2$. 
Assume that we can prove the generalisation of ETH~\eqref{ethold}  for such \emph{deformed Wigner} matrices.
In the space of $W_1$ this would result in a concentration
of $\langle \bw, W_2\bw\rangle$ around some quantity $f(W_2)$ depending only on $W_2$; however, the answer will nontrivially 
depend on the deformation, i.e. it 
will not  be simply  $\langle W_2\rangle$. Once the correct form of  $f(W_2)$  is established, we can 
find its concentration in the probability space of $W_2$, yielding the final answer.

To achieve these results we prove more general  ETH and fluctuation results
for  eigenvector overlaps of  deformed Wigner matrices of the general form $H=W+D$, where $W$ is an $N\times N$ Wigner
matrix and $D$ is arbitrary, bounded deterministic matrix. The goal is to establish the concentration and the fluctuation
of the quadratic form $\langle \bw, A\bw\rangle$ for a normalised eigenvector $\bw$ of $H$ with a bounded deterministic matrix $A$.
We remark that for the special case of a rank one matrix $A=|\bq\rangle \langle \bq|$,
ETH is equivalent to the \emph{complete isotropic delocalisation}
of the eigenvector $\bw$, i.e. that $|\langle \bq, \bw\rangle|\le N^\epsilon/\sqrt{N}$ for any deterministic vector $\bq$ with $\|\bq\|=1$.
For a diagonal deformation $D$ this has been achieved in
\cite{LeeSchnelli, LandonYau}
and for the general deformation $D$  in~\cite{slowcorr}. 
The normal fluctuation of $\langle \bw, A\bw\rangle$  
for a finite rank $A$ and
diagonal $D$  was obtained in~\cite{Benigni}.

It is well known that for very general mean-field type random matrices $H$
their resolvent $G(z)=(H-z)^{-1}$  concentrates around a deterministic matrix $M=M(z)$; such
results are called \emph{local laws}, and we will recall them precisely  in~\eqref{eq:singlegllaw}.
Here $M$ is a solution of the 
\emph{matrix Dyson equation (MDE)}, which, in case of $H=W+D$, reads as
\begin{equation}\label{mde}
	-\frac{1}{M(z)}= z-D +\langle M (z)\rangle\,. 
\end{equation}
Given $M$, it turns out that
\begin{equation}\label{ethuj}
	\langle \bw_i, A\bw_j \rangle \approx \delta_{i, j}
	\frac{ \langle\Im M(\lambda_i) A\rangle}{\langle \Im M(\lambda_i)\rangle}\,,
\end{equation}
where $\lambda_i$ is the eigenvalue corresponding to the normalised eigenvector
$\bw_i$, i.e.~$H\bw_i=\lambda_i \bw_i$. 
Since the eigenvalues are \emph{rigid}, i.e.
they fluctuate only very little, the right hand side of~\eqref{ethuj} is essentially deterministic and in general  it  depends on the energy.
Similarly to~\eqref{quefluc}, we will also establish the Gaussian fluctuation around the approximation~\eqref{ethuj}.
For zero deformation, $D=0$, the matrix $M$ is constant and~\eqref{ethuj} recovers~\eqref{ethold}--\eqref{quefluc}
as a special case.
For simplicity, in the current paper we establish all these results only in the \emph{bulk} of the  spectrum, but similar 
results may be obtained at the \emph{edge} and at the possible \emph{cusp regime} 
of the spectrum as well; the details are deferred to later works.

We now comment on the new aspects of our methods. The proof of~\eqref{ethuj} relies on a basic observation  about
the local law for $H=W+D$. Its  \emph{average} form  asserts that
\begin{equation}\label{locl}
	\Big| \big\langle (G(z)-M(z)) A\big\rangle\Big| \le \frac{N^\epsilon}{N \eta}\,, \qquad \eta : =|\Im z| \gg \frac{1}{N}\,,
\end{equation}
holds with very high probability and the error is essentially optimal for any bounded deterministic matrix $A$.
However, there is a codimension one subspace of the matrices $A$ for which the estimate improves to 
$N^\epsilon/(N\sqrt{\eta})$, gaining a $\sqrt{\eta}$ factor in the relevant small $\eta\ll 1$ regime.
For Wigner matrices $H=W$ without deformation, the \emph{traceless} matrices $A$ played this special role.
The key idea behind the proof of ETH for Wigner matrices  in~\cite{ETHpaper} was to decompose
any deterministic matrix as $A =: \langle A\rangle + \mathring{A}$ into its tracial and traceless parts
and prove multi-resolvent generalisations of the local law~\eqref{locl} with an error term distinguishing
whether the deterministic matrices  are traceless or not. For example, 
for a typical $A$ we have 
$$
\langle G(z) A G(z)^* A \rangle\sim \frac{1}{\eta}
$$
but for the traceless part of $A$ we have
$$
\langle G(z) \mathring{A} G(z)^* \mathring{A} \rangle\sim 1\,, \qquad \eta\gg \frac{1}{N}\,,
$$
with appropriately matching error terms.
In general, each traceless $A$ improves the typical estimate by a factor $\sqrt{\eta}$.  ETH then follows
from the spectral theorem,
\begin{equation}\label{specthm}
	\frac{1}{N}\sum_{i,j=1}^N 
	\frac{\eta}{(\lambda_i-e)^2+\eta^2} \frac{\eta}{(\lambda_j-e)^2+\eta^2} |\langle \bw_i, \mathring{A} \bw_j\rangle|^2
	= \langle \Im G(z)  \mathring{A} \Im G(z) \mathring{A}  \rangle \le N^\epsilon\,;
\end{equation}
choosing $z=e+\ii \eta$ appropriately with $\eta \sim N^{-1+\epsilon}$ we obtain that 
$ |\langle \bw_i, \mathring{A} \bw_j\rangle|^2\le N^{-1+3\epsilon}$ which even includes
an off-diagonal version of~\eqref{ethold}.

To extend this argument  to deformed Wigner matrices requires to 
identify the appropriate \emph{singular} (``tracial'')  and \emph{regular}  (``traceless'') parts of an arbitrary matrix.
It turns out that the improved local laws around an energy $e=\Re z$ 
hold if $A$ is orthogonal\footnote{The space of matrices is equipped with the
	usual Hilbert-Schmidt scalar product.} to $\Im M(e)$, see~\eqref{eq:defreg1} for the new definition of $\mathring{A}$,
which denotes the regular part of $A$.
In this theory the matrix $\Im M$ emerges as the critical eigenvector of a \emph{linear stability operator} 
$\mathcal{B}=I-M \langle\cdot\rangle M^*$ related to the MDE~\eqref{mde}.
The major complication compared with the pure Wigner case in~\cite{ETHpaper} is that now the regular part of a matrix becomes
energy dependent. In particular, in a multi-resolvent chain $\langle G(z_1)A_1 G(z_2) A_2 \ldots \rangle$ it is \emph{a priori}
unclear at which spectral parameters the matrices $A_i$ should be regularised; it turns out that the correct
regularisation depends on both $z_i$ and $z_{i+1}$, see~\eqref{eq:circ def} later. 
A similar procedure was performed for the Hermitisation of non-Hermitian i.i.d. matrices with a general deformation
in~\cite{iid}, see~\cite[Appendix A]{iid} for a more conceptual presentation.
Having identified the correct regularisation, we derive a system of \emph{master inequalities} for the
error terms in multi-resolvent local laws for regular observables; a  similar strategy (with minor modifications)
have been used in~\cite{multiG, A2} for Wigner matrices and in~\cite{iid} for i.i.d. matrices.
To keep the presentation short, here we will not aim at the
most general local laws with optimal errors unlike in~\cite{multiG, A2}. Although these would be achievable with our methods,
here  we prove only what is needed  for our main results on the equipartition.

The proof of the fluctuation around the ETH uses \emph{Dyson Brownian motion (DBM)} techniques, namely
the \emph{Stochastic Eigenstate Equation} for quadratic forms of eigenvectors. 
This theory has been gradually developed for Wigner matrices
in~\cite{BourgadeYau1312.1301, BourgadeYauYin1807.01559, MarcinekYau2005.08425},
we closely follow the presentation in~\cite{normalfluc, A2}. The extension of this technique to
deformed Wigner matrices is fairly straightforward, so our presentation will be brief.
The necessary inputs for this DBM analysis follow from 
the multi-resolvent local laws that we prove  for deformed Wigner matrices.

In a closing remark  we mention that 
the original proof of~\eqref{old} in~\cite{Equipartition} was 
considerably simpler than that of~\eqref{ethold}. 
This may appear quite surprising due to the complicated correlation between $\bw$ and $\mathcal{W}$, but a
special algebraic cancellation greatly  helped  in~\cite{Equipartition}.  
Namely, with the notation $\mathcal{W}:=W_1-W_2$ and $G(z):=(H-z)^{-1}$, $z=e+\ii \eta\in \C_+$,
a relatively straightforward cumulant expansion 
showed that $ \langle \Im G(z) \mathcal{W} \Im G(z) \mathcal{W} \rangle$
is essentially bounded\footnote{This means up to an  $N^\epsilon$ factor with arbitrary small $\epsilon$.} even for spectral parameters $z$ very close to the real axis, $\eta\ge N^{-1+\epsilon/2}$.
Within this cumulant expansion an algebraic cancellation emerged due to the special form of $\mathcal{W}$. Then, 
exactly as in~\eqref{specthm} we 
obtain $|\langle \bw, \mathcal{W}\bw\rangle|^2\lesssim N^{1+\epsilon}\eta^2 = N^{-1+2\epsilon}$.
In particular, it shows that  $\langle \bw, \mathcal{W} \bw\rangle = \langle \bw, W_1 \bw\rangle - \langle \bw, W_2 \bw\rangle$
is essentially of order $N^\epsilon/\sqrt{N}$ for every eigenvector of $H$. Since
$\langle \bw, W_1 \bw\rangle + \langle \bw, W_2 \bw\rangle = \langle \bw, H \bw\rangle =\lambda $, we immediately 
obtain the equipartition~\eqref{old}. Similar idea proved the more general case, see~\eqref{eq:equipartitionbound} later.
Note, however, that this trick does not help in establishing the fluctuations of
$\langle \bw, W_l \bw\rangle$. In fact,  the full ETH analysis for deformed Wigner matrices needs to  be performed to 
establish the necessary \emph{a priori} bounds for the Dyson Brownian motion arguments.

\subsection*{Notations and conventions} 

For positive quantities $f,g$ we write $f\lesssim g$ and $f\sim g$ if $f \le C g$ or $c g\le f\le Cg$, respectively, for some constants $c,C>0$ which depend only on the constants appearing in the moment condition, see~\eqref{eq:momass} later.
For any natural number $n$ we set $[n]: =\{ 1, 2,\ldots ,n\}$.

We denote vectors by bold-faced lower case Roman letters ${\bm x}, {\bm y}\in\C ^N$, for some $N\in\N$. Vector and matrix norms, $\lVert {\bm x}\rVert$ and $\lVert A\rVert$, indicate the usual Euclidean norm and the corresponding induced matrix norm. For any $N\times N$ matrix $A$ we use the notation $\langle A\rangle:= N^{-1}\mathrm{Tr}  A$ to denote the normalised trace of $A$. Moreover, for vectors ${\bm x}, {\bm y}\in\C^N$ and matrices  $A\in\C^{N\times N}$ we define the scalar product
\[
\langle {\bm x},{\bm y}\rangle:= \sum_{i=1}^N \overline{x}_i y_i\,. 
\]

Finally, we will use the concept of ``with very high probability'' \emph{(w.v.h.p.)} meaning that for any fixed $D>0$ the probability of an $N$-dependent event is bigger than $1-N^{-D}$ for $N\ge N_0(D)$. We introduce the notion of \emph{stochastic domination} (see e.g.~\cite{semicirclegeneral}): given two families of non-negative random variables
\[
X=\left(X^{(N)}(u) : N\in\N, u\in U^{(N)} \right) \quad \mathrm{and}\quad Y=\left(Y^{(N)}(u) : N\in\N, u\in U^{(N)} \right)
\] 
indexed by $N$ (and possibly some parameter $u$  in some parameter space $U^{(N)}$), 
we say that $X$ is stochastically dominated by $Y$, if for all $\xi, D>0$ we have 
\begin{equation}
\label{stochdom}
\sup_{u\in U^{(N)}} \mathbf{P}\left[X^{(N)}(u)>N^\xi  Y^{(N)}(u)\right]\le N^{-D}
\end{equation}
for large enough $N\ge N_0(\xi,D)$. In this case we use the notation $X\prec Y$ or $X= \mathcal{O}_\prec(|Y|)$.
We also use the convention that $\xi>0$ denotes an arbitrary small exponent which is independent of $N$.

\subsection*{Acknowledgement} G.C. and L.E.  gratefully acknowledge many discussions with Dominik Schr\"oder 
at the preliminary stage of this project, especially his  essential contribution to identify the correct
generalisation of traceless observables to the deformed Wigner ensembles.

\section{Main results}
We consider $N\times N$ real symmetric or complex Hermitian Wigner matrices $W = W^*$
having single-entry distributions $w_{ab}=N^{-1/2}\chi_{\mathrm{od}}$, for $a>b$, and $w_{aa}=N^{-1/2}\chi_d$, where $\chi_{\mathrm{od}}$ and $\chi_{\mathrm{od}}$ are two independent random variables satisfying the following assumptions:
\begin{assumption}
\label{ass:momass}
We assume that $\chi_{\mathrm{d}}$ is a real centred random variable, that $\chi_{\mathrm{od}}$ is a real or 
complex random variable such that $\E \chi_{\mathrm{od}}=0$ and $\E |\chi_{\mathrm{od}}|^2=1$; additionally 
in the complex case we also assume that $\E\chi_{\mathrm{od}}^2=0$.  Customarily, we use the parameter $\beta=1,2$ 
to indicate the real or complex case, respectively.
Furthermore, we assume that all the moments of $\chi_{\mathrm{od}}$ and $\chi_{\mathrm{d}}$ exist, i.e.~for any $p\in \N$ there exists a constant $C_p >0$ such that
\begin{equation}
\label{eq:momass}
\E|\chi_{\mathrm{od}}|^p+\E |\chi_{\mathrm{d}}|^p\le C_p.
\end{equation}
\end{assumption}
For definiteness, in the sequel we perform the entire analysis for the complex case; the real case being completely analogous and hence omitted. 

The equipartition principle concerns linear combinations of Wigner matrices.
Fix $k \in \N$ and consider
\begin{equation}
	\label{eq:defH}
	H:=p_1W_1+\dots+p_kW_k,
\end{equation}
for some fixed $N$-independent vector $\bm{p} = (p_1, ... , p_k) \in \R^k$ of weights and for $k$ independent $N\times N$ Wigner matrices $W_l$, belonging to the \emph{same} symmetry class (i.e.~the off-diagonal random variables $\chi_{\mathrm{od}}$ are either real or complex for each of the $W_l$, $l \in [k]$). Then, denoting by $\{\lambda_i\}_{i \in [N]}$ the eigenvalues of $H$, arranged in increasing order, with associated normalised eigenvectors $\{{\bm u}_i\}_{i \in [N]}$, 
the total energy $\langle \bm{u}_i, H \bm{u}_i \rangle = \lambda_i$ of the composite system \eqref{eq:defH} is \emph{proportionally}  distributed among the $k$ constituents, i.e. 
\begin{equation}
	\label{eq:equipartion}
	\langle {\bm u}_i, p_lW_l \, {\bm u}_i\rangle\approx \frac{p_l^2 }{\Vert \bm{p}\Vert^2} \lambda_i
\end{equation}
for every $l \in [k]$, where $\Vert \bm{p}\Vert := \big(\sum_{l=1}^k |p_l|^2\big)^{1/2}$ denotes the usual $\ell^2$-norm. This phenomenon, known as \emph{equipartition}, was first proven in \cite[Theorem~3.4]{Equipartition},  with an optimal error estimate:
\begin{equation} \label{eq:equipartitionbound}
	\left|\langle {\bm u}_i, \, p_lW_l {\bm u}_j\rangle-\delta_{i,j}\frac{p_l^2 }{\Vert \bm{p}\Vert^2} \lambda_i\right|\prec \frac{1}{\sqrt{N}}\,. 
\end{equation}
Our main result is the corresponding Central Limit Theorem to \eqref{eq:equipartitionbound} for $i = j$,
 i.e.~the proof of Gaussian fluctuations in Equipartition for Wigner matrices -- for energies in the \emph{bulk} of the spectrum of $H$.
 
\begin{theorem}[Gaussian Fluctuations in Equipartition]  \label{thm:main} ~\\  Fix $k\in \N$. Let $W_1,\dots, W_k$ be 
independent Wigner matrices satisfying Assumption \ref{ass:momass}, all of which being in the same real ($\beta=1$)
 or complex ($\beta=2$) symmetry class, and $\bm{p} = (p_1,\dots,p_k)\in \R^k$ be $N$-independent.
  Define $H$ as in \eqref{eq:defH} and denote by $\{\lambda_i\}_{i \in [N]}$ the eigenvalues of $H$, arranged
   in increasing order, with associated normalised eigenvectors $\{{\bm u}_i\}_{i \in [N]}$. Then, for
    fixed $\kappa > 0$, every  $l \in [k]$ and for every bulk  index $i \in [\kappa N, (1-\kappa)N]$  it holds that 
	\begin{equation} \label{eq:mainthm}
	\sqrt{\frac{\beta N }{2} \, \frac{\Vert \bm{p}\Vert^2}{p_l^2 \, \big(\Vert \bm{p}\Vert^2 - p_l^2\big)  }}  \; \left[\langle {\bm u}_i, \, p_l W_l {\bm u}_i\rangle-\frac{p_l^2 }{\Vert \bm{p}\Vert^2} \lambda_i\right]\Longrightarrow \mathcal{N}(0,1)
	\end{equation}
in the sense of moments,\footnote{Given a sequence of $N$-dependent random variables, we say that $X_N$ converges to $X_\infty$
in the sense of moments if for any $k\in \N$ it holds $\E |X_N|^k=\E|X_\infty|^k+\mathcal{O}(N^{-c(k)})$, for some small possibly $k$--dependent constant $c(k)>0$.\label{moments}} where $\mathcal{N}(0,1)$ denotes a real standard Gaussian. 
\end{theorem}

By polarisation we will also  obtain the following:
\begin{corollary} \label{cor:main}
Under the assumptions from Theorem \ref{thm:main}, the random vector $\bm{X} = (X_1, ... , X_k) \in \R^k$ with
\begin{equation}\label{Xl}
X_l:= \sqrt{\frac{\beta N}{2}} \; 
\left[\langle {\bm u}_i, \, p_l W_l {\bm u}_i\rangle-\frac{p_l^2 }{\Vert \bm{p}\Vert^2} \lambda_i\right]\,, \quad l \in [k]\,,
\end{equation}
is approximately (in the sense of moments) jointly Gaussian with covariance structure
\begin{equation*}
\mathrm{Cov}(X_l, X_m) = \frac{p_l^2 \big(\delta_{l,m} \Vert \bm{p}\Vert^2 - p_m^2\big)}{\Vert \bm{p}\Vert^2}\,. 
\end{equation*}
\end{corollary}

\begin{remark}
\label{rem:off}

We stated Theorem~\ref{thm:main} only for diagonal overlaps for simplicity. However, one can see that following the proof in \cite[Section 3]{CipolloniBenigni2022} it is possible to obtain an analogous Central Limit Theorem (CLT) for off--diagonal overlaps as well:
\begin{equation} \label{eq:mainthm}
	\sqrt{\frac{\Vert \bm{p}\Vert^2\beta N}{p_l^2 \, \big(\Vert \bm{p}\Vert^2 - p_l^2\big)  }}  \; \big|\langle {\bm u}_i, \, p_l W_l {\bm u}_j\rangle \big|\Longrightarrow \big|\mathcal{N}(0,1)\big|.
	\end{equation}
This also gives an analogous version of \eqref{cor:main} for off-diagonal overlaps. Furthermore, again following \cite[Theorem 2.2]{CipolloniBenigni2022}, it is also possible to derive a multivariate CLT jointly for diagonal and off--diagonal overlaps. See also Remark~\ref{rem:off2} below for further explanation.
\end{remark}

Theorem \ref{thm:main} 
and Corollary~\ref{cor:main} will follow as a corollary to the Eigenstate Thermalisation Hypothesis (ETH) and its Gaussian fluctuations for \emph{deformed} Wigner matrices, which we present as Theorem \ref{theo:ETH} and Theorem \ref{theo:newCLT} in the following subsection. 
\begin{remark}
	By a quick inspection of our proof of Theorem \ref{thm:main}, given in Section \ref{sec:proofmain}, it is possible to 
	generalise the Equipartition principle \eqref{eq:equipartitionbound} as well as its Gaussian fluctuations \eqref{eq:mainthm} to 
	linear combinations of 
	\emph{deformed Wigner matrices}, i.e. each~$W_l$ in \eqref{eq:defH} being replaced by 
	$W_l + D_l$, where $D_l = D_l^*$ is an 
	essentially   arbitrary bounded deterministic matrix
	(see Assumption \ref{ass:Mbdd} later). However, for brevity of the current paper, 
	we refrain from presenting this extension  explicitly. 
\end{remark}

\subsection{ETH and its fluctuations for deformed Wigner matrices} In this section, we consider 
\emph{deformed Wigner matrices}, $H = W + D$, with increasingly ordered eigenvalues $\lambda_1\le\lambda_2\le \dots\le \lambda_N$ 
and corresponding orthonormal eigenvectors ${\bm u}_1,\dots, {\bm u}_N$. Here, $D =D^* \in \C^{N \times N}$ 
is a self-adjoint matrix with uniformly bounded norm, i.e.~$\lVert D\rVert\le C_D$ for some $N$-independent
 constant $C_D>0$. While the  \emph{Eigenstate Thermalisation Hypothesis (ETH)} will be shown to hold for 
 general deformations $D$, we shall require slightly stronger assumptions for proving the Gaussian
  fluctuations (see Assumption \ref{ass:Mbdd} below).

In order to state our results on the ETH and its fluctuations (Theorems \ref{theo:ETH} and \ref{theo:newCLT}, respectively),
 we need to introduce the concept of \emph{regular observables}, first in 
a simple form in Definition~\ref{eq:defreg1} (later  along the proofs we will need a more general version in Definition~\ref{def:regobs}).
 For this purpose we introduce $M(z)$ being the unique solution of the Matrix Dyson Equation (MDE):\footnote{The MDE for very general mean field random matrices has been introduced in~\cite{firstcorr} and further analysed in~\cite{1804.07752}.
The properties we use here have been summarised in Appendix B of the \href{https://arxiv.org/abs/2301.03549}{arXiv: 2301.03549} 
version of~\cite{iid}.}
\begin{equation}
\label{eq:MDE}
-\frac{1}{M(z)}=z-D+\langle M(z)\rangle, \qquad\quad \Im M(z)\Im z>0.
\end{equation}
The \emph{self consistent density of states (scDos)} is then defined as
\begin{equation}
\label{eq:scdos}
\rho(e):=\frac{1}{\pi}\lim_{\eta\downarrow 0}\langle \Im M(e+\ii\eta)\rangle\,. 
\end{equation}
We point out that not only $\langle \Im M(e+\ii\eta)\rangle$ has an extension to the real axis, but the whole matrix $M(e) := \lim_{\eta\downarrow 0} M(e + \ii \eta)$ is well defined (see Lemma B.1~(b) of the \href{https://arxiv.org/abs/2301.03549}{arXiv: 2301.03549} version of~\cite{iid}).
The scDos $\rho$ is a compactly supported Hölder-$1/3$ continuous function on $\R$ which is real-analytic on the 
set $\{ \rho > 0\}$\footnote{The scDos has been thoroughly analysed in increasing generality in~\cite{1506.05095,firstcorr,1804.07752}.
It is supported on finitely many finite intervals.
Roughly speaking there are  three regimes: the \emph{bulk}, where $\rho$ is well separated away from 0, the
\emph{edge} where $\rho$ vanishes as a square root at the edges of each supporting interval that are well separated, and the \emph{cusp} where two supporting intervals (almost) meet
and $\rho$ behaves (almost) as a cubic root. Correspondingly, $\rho$ is locally real analytic, H\"older-$1/2$, or H\"older-$1/3$ continuous,
respectively. Near the singularities, it has  an approximately universal shape. No other singularity type can occur
and for typical deformation $D$ there is no cusp regime.\label{cusp}}. 
Moreover, for any small $\kappa > 0$ (independent of $N$) we define the \emph{$\kappa$-bulk} of the scDos as 
\begin{equation} \label{eq:bulk}
\mathbf{B}_\kappa = \left\{ x \in \R \; : \; \rho(x) \ge \kappa^{1/3} \right\}\,,
\end{equation}
which is a finite union of disjoint compact intervals, see Lemma B.2 in
the \href{https://arxiv.org/abs/2301.03549}{arXiv: 2301.03549} version of~\cite{iid}. 
For $\Re z \in \mathbf{B}_\kappa$ it holds that $\Vert M(z) \Vert \lesssim 1$, as easily follows by taking the imaginary part of \eqref{eq:MDE}.

\begin{definition}[Regular observables -- One-point regularisation] \label{def:regobs1}
Fix $\kappa > 0$ and an energy $e\in \mathbf{B}_\kappa$ in the bulk. Given a 
 matrix $A\in \C^{N\times N}$, we define its \emph{one-point regularisation w.r.t.~the energy $e$}, denoted by $\mathring{A}^e$, as
\begin{equation}
\label{eq:defreg1}
\mathring{A}=\mathring{A}^e:=A-\frac{\langle \Im M(e) A\rangle}{\langle \Im M(e)\rangle}\,. 
\end{equation}
Moreover, we call $A$ \emph{regular w.r.t.~the energy $e$}, if and only if $A = \mathring{A}^e$. 
\end{definition}
Notice that in the analysis of Wigner matrices  without deformation, $D=0$, in~\cite{ETHpaper, thermalization, multiG, A2},
 $M$ was a constant matrix
and the regular observables were simply given by \emph{traceless} matrices, i.e. $\mathring{A}= A-\langle A\rangle$.
For deformed Wigner matrices  the concept of regular observables depends
 on the  energy. 

Next, we define the \emph{quantiles} $\gamma_i$ of the density $\rho$ implicitly by
\begin{equation}
\label{eq:quant}
\int_{-\infty}^{\gamma_i} \rho(x)\,\dif x=\frac{i}{N}, \qquad i\in[N].
\end{equation}
We can now formulate the ETH in the bulk for deformed Wigner matrices
which generalises the same result for Wigner matrices, $D=0$, from~\cite{ETHpaper}.

\begin{theorem}[Eigenstate Thermalisation Hypothesis] \label{theo:ETH} 
Let $\kappa > 0$ be an $N$-independent constant and fix a bounded deterministic $D =D^* \in \C^{N \times N}$. 
Let $H = W + D$ be a deformed Wigner matrix, where $W$ satisfies Assumption~\ref{ass:momass}, and denote the orthonormal eigenvectors of $H$ by $\{ \bm{u}_i\}_{i \in [N]}$. Then, for any deterministic $A\in \C^{N\times N}$ with $\Vert A \Vert \lesssim 1$, it holds that 
\begin{equation}
\label{eq:ETH}
\max_{i, j}\left|\langle \bm{u}_i,\mathring{A}^{\gamma_i}\bm{u}_j\rangle \right| = \max_{i, j}\left|\langle \bm{u}_i,A\bm{u}_j\rangle-\delta_{ij}\frac{\langle A\Im M(\gamma_i)\rangle}{\langle \Im M(\gamma_i)\rangle} \right| \prec \frac{1}{\sqrt{N}}\,,
\end{equation}
where the maximum is taken over all $i,j \in [N]$ such that the quantiles $\gamma_i, \gamma_j \in \mathbf{B}_\kappa$ defined in \eqref{eq:quant} are in the $\kappa$-bulk of the scDos $\rho$.
\end{theorem}
This "Law of Large Numbers"-type result \eqref{eq:ETH} is complemented by the corresponding Central Limit Theorem \eqref{eq:standgauss}, which requires slightly strengthened assumptions on the deformation $D$. 
\begin{assumption} \label{ass:Mbdd}
We assume that $D  \in \C^{N \times N}$ is a bounded self-adjoint deterministic matrix such that 
\begin{itemize}
	\item[(i)] the unique solution $M(z)$ to \eqref{eq:MDE} is uniformly bounded in norm, i.e.~$\sup_{z \in \C} \Vert M(z) \Vert \le C_M$ for some $N$-independent constant $C_M > 0$;
	\item[(ii)] the scDos $\rho$ is Hölder-$1/2$ regular, i.e.~it does not have any cusps (see Footnote~\ref{cusp}).
\end{itemize}
\end{assumption}
The requirements on $D$ in Assumption \ref{ass:Mbdd} are natural and they hold for typical applications,
see Remark~\ref{rmk:Mbdd} later for more details.
We can now formulate our result on the Gaussian fluctuations in the ETH
which generalises the analogous result for Wigner matrices, $D=0$ from~\cite{normalfluc}.
\begin{theorem}[Fluctuations in ETH]
\label{theo:newCLT}
Fix $\kappa,\sigma > 0$ $N$-independent constants  and let $H = W + D$ be a deformed Wigner matrix, where $W$ satisfies Assumption~\ref{ass:momass} and $D$ satisfies Assumption \ref{ass:Mbdd}. Denote the orthonormal eigenvectors of $H$ by $\{ \bm{u}_i\}_{i \in [N]}$ and fix an index $i \in [N]$, such that the quantile $\gamma_i \in \mathbf{B}_\kappa$ defined in \eqref{eq:quant} is in the bulk. 
Then, for any deterministic
Hermitian matrix $A\in\C^{N\times N}$ with $\Vert A \Vert \lesssim 1$
 (which we assume to be real in the case of a real Wigner matrix) satisfying $\langle (A-\langle{A}\rangle)^2\rangle \ge \sigma$,  it holds that
\begin{equation}
\label{eq:standgauss}
\sqrt{\frac{\beta N }{2 \,  \Var_{\gamma_i}(A)}}\left[\langle \bm{u}_i,A\bm{u}_i\rangle-\frac{\langle A\Im M(\gamma_i)\rangle}{\langle \Im M(\gamma_i)\rangle} \right]\Longrightarrow \mathcal{N}(0,1)
\end{equation}
in the sense of moments (see Footnote~\ref{moments}), where\footnote{See the first paragraph of Section~\ref{sec:proofETHfluct} for an explanation of why the variance takes this specific form.}
\begin{equation} \label{eq:VarAgammai}
	\Var_{\gamma_i}(A) := \frac{1}{\langle \Im M (\gamma_i)\rangle^2} \left( \big\langle \big(\mathring{A}^{\gamma_i} \Im M(\gamma_i)\big)^2\big\rangle - \frac{1}{2} \Re\left[ \frac{\big\langle \big(M(\gamma_i)\big)^2 \mathring{A}^{\gamma_i} \big\rangle^2}{1 - \big\langle \big(M(\gamma_i)\big)^2 \big\rangle} \right]  \right) \,. 
\end{equation}
This variance  is strictly positive with an effective lower bound
\begin{equation} \label{eq:Varlowerbound}
		\Var_{\gamma_i}(A)\ge c \, \langle (A-\langle{A}\rangle)^2\rangle
	\end{equation}
	for some constant $c = c(\kappa, \Vert D \Vert) > 0.$ \nc
\end{theorem}

\begin{remark}
\label{rem:off2}


We stated Theorem~\ref{theo:newCLT} only for diagonal overlaps to keep the statement simple, but a corresponding CLT for off--diagonal overlaps as well as a multivariate CLT for any finite family of diagonal and off--diagonal overlaps can also be proven.

We decided not to give a detailed proof of these facts in the current paper in order to present the main new ideas in the analysis of deformed Wigner matrices in the simplest possible setting consisting of only diagonal overlaps. But we remark that following an analysis similar to \cite[Section 3]{CipolloniBenigni2022}, combined with the details presented in Section~\ref{sec:proofETHfluct}, would give an analogous result to \cite[Theorem 2.2]{CipolloniBenigni2022} also in the deformed Wigner matrices setup. However, this would require introducing several new notations that would obfuscate the main novelties in the analysis of deformed Wigner matrices compared to the Wigner case, which instead are clearer in the simpler setup of Section~\ref{sec:proofETHfluct}.

\end{remark}

In the following two remarks we comment on the
 condition $\langle (A-\langle{A}\rangle)^2\rangle \ge \sigma$
 and on Assumption \ref{ass:Mbdd}.
 
\begin{remark}\label{rmk:Variance}  
The restriction to matrices satisfying $\langle (A-\langle{A}\rangle)^2\rangle \ge \sigma$, 
i.e.~$A - \langle A \rangle$ being of  high rank, is technical. It is
 due to the fact that our multi-resolvent local laws
 for resolvent chains $\langle G(z_1) A_1 G(z_2) A_2 \ldots \rangle$
  in Proposition~\ref{pro:mresllaw} are non-optimal in terms of the norm 
 for the matrices $A_i$; they involve the Euclidean norm $\Vert A_i\Vert$ and 
 not the smaller Hilbert-Schmidt norm
 $\sqrt{\langle |A_i|^2\rangle}$ which would be optimal. For the Wigner ensemble, this subtlety is the main difference between 
the main result in \cite{normalfluc} for high rank observable matrices $A$ and its extension to any low rank  $A$ in \cite{A2}.
Following the technique in \cite{A2}  it would be possible to achieve the estimate with the 
 optimal norm of $A$ also for deformed Wigner matrices.
However, we refrain from doing so, since in our main application, Theorem \ref{thm:main}, 
$A$ itself will be a Wigner matrix which has high rank.
\end{remark}

\begin{remark}\label{rmk:Mbdd} We have several comments on Assumption \ref{ass:Mbdd}. 
	\begin{itemize}
		\item[(i)] 
		The boundedness of $\Vert M (z)\Vert$ is automatically fulfilled in the bulk $\mathbf{B}_\kappa$
		(see remark below~\eqref{eq:bulk})
		or when  $\Re z$ away from the support of the scDos $\rho$
		(see \cite[Proposition~3.5]{1804.07752})
		without any further condition.
		 However, the uniform (in $z$) estimate formulated in Assumption \ref{ass:Mbdd} does not hold for arbitrary $D$.
		  A sufficient condition for the boundedness of $\Vert M\Vert $ in terms of the spectrum of $D$ is given 
		  in \cite[Lemma 9.1 (i)]{1804.07752}. 
		  This especially applies if the eigenvalues $\{d_i\}_{i \in [N]}$ of $D$ (in increasing order) form a 
		  piecewise Hölder-$1/2$ regular sequence,\footnote{In this context, Hölder-$1/2$ regularity 
		  means that $|d_i -d_j|\le C_0 (|i-j|/N)^{1/2}$ for some universal constant $C_0 > 0$.} see~\cite[Lemma 9.3]{1804.07752}.
		   In particular, by eigenvalue rigidity \cite{firstcorr, slowcorr}, it is easy to see that any ``Wigner-like" 
		   matrix $D$ has Hölder-$1/2$ regular sequence of eigenvalues with very high probability.
		    This is important for the applicability of 
		    Theorem~\ref{theo:newCLT} below in the proof of our main result, Theorem \ref{thm:main}, 
		    given in Section \ref{sec:proofmain}. 
		\item[(ii)] The assumption that $\rho$ does not have any cusps is a typical condition and of  technical nature 
		(needed in the local law~\eqref{eq:singlegllaw} and in Lemma \ref{lem:2Gllfaraway}). 
		In case that the sequence of matrices $D=D_N$ has a 
		limiting density of states with single interval support, then also $\rho$, the scDos of $W+D$, 
		 has single interval support \cite{biane}, 
		in particular,  $\rho$  has no cusps \cite{1804.07752}. 
		Again, this is important for the applicability of Theorem \ref{theo:newCLT} in the proof of our main result, in which case $D$ is a Wigner matrix with a semicircle as the limiting density of states.
	\end{itemize}
\end{remark}

In the following Section \ref{sec:proofmain}, we will prove our main result, Theorem \ref{thm:main}, 
assuming Theorems~\ref{theo:ETH} and \ref{theo:newCLT} on deformed Wigner matrices
 as inputs. These will be proven in Sections \ref{sec:proofETH} and \ref{sec:proofETHfluct}, respectively. 
 Both proofs crucially rely on an averaged local law for two resolvents and two \emph{regular} observables, 
 Proposition~\ref{pro:mresllaw}, which we prove in Section~\ref{sec:llaw}. Several additional technical and auxiliary results are deferred to the Appendix.

\section{Fluctuations in Equipartition: Proof of Theorem \ref{thm:main}} \label{sec:proofmain}
It is sufficient to prove Theorem \ref{thm:main} only for $k=2$ with $p_1, p_2 \neq 0$, since
we can view the sum~\eqref{eq:defH}  as the sum of $p_1W_1$ and 
\begin{equation*}
\sum\limits_{l = 2}^k p_lW_l\stackrel{\mathrm{d}}{=}\left(\sum\limits_{ l=2}^kp_l^2\right)^{1/2}\widetilde{W}\,,
\end{equation*}
where $\widetilde{W}$ is a Wigner matrix independent of $W_1$
 and the equality is understood in distribution. 

As a main step, we shall prove the following lemma, where we condition on $W_2$.
\begin{lemma} \label{lem:main}Under the assumptions of Theorem \ref{thm:main} with $k=2$ it holds that
\begin{align}
    \E_{W_1} \langle \bw_i,p_2W_2\bw_i\rangle &= \frac{p_2^2}{\|\bm{p}\|^2}\gamma_i+\mathcal{O}_\prec\left(N^{-1/2-\epsilon}\right)\,, \label{eq:expectation}\\[1mm]
        \frac{\beta N}{2} \Var_{W_1}[\langle \bw_i,p_2W_2\bw_i\rangle ]  &= \frac{ p_1^2p_2^2}{\Vert \bm{p}\Vert^2}+\mathcal{O}_\prec\left(N^{-\epsilon}\right)\,, \label{eq:variance}
\end{align}
for any $\epsilon> 0$, where $\gamma_i$ is the $i^{\rm th}$ quantile of the semicircular density with 
radius $2 \Vert \bm{p}\Vert$, i.e.
$$
   \frac{1}{2\pi  \Vert \bm{p}\Vert^2} \int_{-\infty}^{\gamma_i} \sqrt{ [4  \Vert \bm{p}\Vert^2 -x]_+}{\rm d} x = \frac{i}{N}\,. 
$$
 Expectation and variance are taken in the probability space of $W_1$,  conditioned on $W_2$ being in an event
 of very high probability, 
while the stochastic domination in the error terms are understood in the probability space of $W_2$. 
\end{lemma}
\begin{proof}[Proof of Theorem \ref{thm:main}] First, we note that all requirements for applying Theorem \ref{theo:newCLT} to $H = p_1 W_1 + D$, with $D = p_2 W_2$ for some fixed realisation of $W_2$ in a very high probability event, are satisfied. This follows from Remark \ref{rmk:Mbdd} and $\langle (W_2 - \langle W_2 \rangle)^2\rangle \gtrsim 1$ with very high probability. Next, observe that replacing $\gamma_i$ in \eqref{eq:expectation} by the eigenvalue $\lambda_i$ appearing in \eqref{eq:mainthm} is trivial
 by the usual eigenvalue rigidity $|\gamma_i - \lambda_i| \prec 1/N$ for Wigner matrices in the bulk  \cite{EYbook}. 
 Thus, 
  Theorem \ref{theo:newCLT}  shows that, conditioned on a fixed realisation of $W_2$,
	\begin{equation} \label{eq:Gaussian2cond}
\sqrt{ \frac{\beta N}{2} }\left[\langle \bw_i,p_2W_2\bw_i\rangle - \frac{p_2^2}{\|\bm{p}\|^2}\lambda_i\right]
	\end{equation}
is approximately Gaussian  with an approximately constant  variance (independent of $W_2$) 
given in \eqref{eq:variance}.
Since this holds with very high probability w.r.t. $W_2$, this proves~\eqref{eq:mainthm} for $l=2$; the proof
for $l=1$ is the same.
\end{proof}

\begin{proof}[Proof of Corollary~\ref{cor:main}] We formulated Theorem~\ref{theo:newCLT} as
a CLT for overlaps $\langle \bm{u}_i,A\bm{u}_i\rangle$ for a single deterministic matrix $A$, but
by standard polarisation it also shows the joint approximate Gaussianity
of any $p$-vector 
\begin{equation}\label{vector}
 \big(  \langle \bm{u}_i,A_1\bm{u}_i\rangle, \langle \bm{u}_i,A_2\bm{u}_i\rangle, \ldots, \langle \bm{u}_i,A_p\bm{u}_i\rangle\big)
 \end{equation}
 for any fixed $k$ and deterministic observables $A_1, A_2, \ldots  A_p$  satisfying $\langle (A_j- \langle A_j\rangle)^2\rangle\ge c$,
 $j\in [p]$. Namely, using Theorem~\ref{theo:newCLT} to compute the moments of
  $\langle \bm{u}_i,A(\bm{t}) \bm{u}_i\rangle$ for the linear combination
 $A(\bm{t})= \sum_j t_j A_j$
with  any 
 real vector $\bm{t} = (t_1, t_2,\ldots, t_p)$, we can identify any joint moments of the coordinates of the vector in~\eqref{vector}
 and we find that they satisfy the (approximate) Wick theorem. 
 
 Now we can follow the above proof of Theorem \ref{thm:main}, but without the simplification $k=2$. 
 Conditioning on $W_2,\ldots, W_k$ and working in the probability space of $W_1$, by the polarisation
 argument above we find that
  not only each $X_l$   from~\eqref{Xl} is  asymptotically Gaussian with a variance 
  independent of $W_2, \ldots, W_k$, but they are jointly Gaussian
  for $l=2,3,\ldots, k$. This is sufficient for the joint Gaussianity of the entire vector $\bm{X}$ since $\sum_l X_l=0$.
This completes the proof of Corollary~\ref{cor:main}.
\end{proof}

The proof of Lemma~\ref{lem:main} is divided into the computation of the
 expectation \eqref{eq:expectation} and the variance \eqref{eq:variance}.

 \subsection{Computation of the expectation \eqref{eq:expectation}}\label{subsec:exp} 
 As in the proof of Theorem \ref{thm:main} above, we condition on $W_2$ and 
work in the probability space of $W_1$ 
i.e. we consider $p_2 W_2$ as a deterministic deformation of $p_1W_1$.
This allows us to use Theorem~\ref{theo:newCLT}
as\footnote{Note that Theorem \ref{theo:ETH} alone would prove \eqref{eq: expectation in equipart 1} only with $\epsilon=0$, but the convergence in the sense of 
moments from Theorem \ref{theo:newCLT} gains  a factor $N^{-\epsilon}$ with a positive $\epsilon$.} 
\begin{equation}
	\E_{W_1}\langle \bw_i,p_2W_2 \bw_i\rangle = \frac{\langle p_2W_2\Im M_2(\gamma_{i,2}))\rangle}{\langle \Im M_2(\gamma_{i,2})\rangle}+\mathcal{O}_\prec\left(N^{-1/2-\epsilon}\right)
	\label{eq: expectation in equipart 1}
\end{equation}
for some constant $\epsilon> 0$. Here $M_2(z)$, depending on $W_2$, is the unique solution of the MDE
\begin{equation}
	-\frac{1}{M_2(z)}=z-p_2W_2+p_1^2\langle M_2(z)\rangle\,,
	\label{eq:MDE equipart}
\end{equation}
corresponding to the matrix $p_1W_1+p_2W_2$, where $p_2W_2$ is considered a deformation, 
and $\gamma_{i,2}$ is the $i^{\rm th}$ quantile of the scDos $\rho_2$ corresponding to \eqref{eq:MDE equipart}. The 
subscript `$2$' for $M_2$, $\rho_2$ and $\gamma_{i,2}$ in \eqref{eq: expectation in equipart 1} and \eqref{eq:MDE equipart} 
indicates that these objects are dependent on $W_2$ and hence random.

The Stieltjes transform $m_2(z)$ of $\rho_2$ is given by the implicit equation
$$
  m_2(z):=\langle M_2(z)\rangle = \frac{1}{p_2}\cdot  \Big\langle
  \frac{1}{W_2-\frac{1}{p_2}(z+p_1^2m_2(z))} \Big\rangle
 $$
 with the usual side condition $\Im z\cdot \Im m_2(z)>0$.
 Applying  the standard local law for the resolvent of $W_2$  on the right hand side shows that
 \begin{equation}\label{loclaw}
 \Big| m_2(z) - \frac{1}{p_2} m_{\rm sc}(w_2)\Big|\prec \frac{1}{N |\Im w_2|}, \qquad w_2:=\frac{1}{p_2}(z+p_1^2m_2(z)).
 \end{equation}
 where $m_{\rm sc}$ is the Stieltjes transform of the standard semicircle law, i.e. it satisfies the
 quadratic   equation
\begin{equation} \label{eq:mdesc}
m_{\rm sc}(w)^2+w m_{\rm sc}(w)+1=0\,
\end{equation}
with the side condition $\Im w\cdot \Im m_{\rm sc}(w)>0$. 
 Note that in~\eqref{loclaw}  $w_2$ is random, it depends on $W_2$, but the local law for $\langle (W_2-w)^{-1}\rangle$
 holds uniformly  in the spectral parameter $|\Im w|\ge N^{-1}$, hence a standard grid argument and the Lipschitz continuity 
 of the resolvents shows that it holds for any (random) $w$ with $|\Im w|\ge N^{-1+\xi}$ with any fixed $\xi>0$. 
 
 Applying~\eqref{eq:mdesc}  at $w=w_2$ together with~\eqref{loclaw} implies that
\begin{equation}\label{apprloclaw}
 - \frac{1}{\Vert \bm{p} \Vert \, m_2(z)} =\frac{z}{\Vert \bm{p}\Vert} + \Vert \bm{p}\Vert\,  m_2(z) + \mathcal{O}_\prec \Big( \frac{1}{N |\Im w_2|}\Big).
\end{equation}
 We  view this relation as a small additive perturbation of the exact equation 
 \begin{equation}\label{msc}
 - \frac{1}{m_{\rm sc}\big(\frac{z}{\Vert \bm{p}\Vert}\big)} =\frac{z}{\Vert \bm{p}\Vert} + m_{\rm sc}\big(\tfrac{z}{\Vert \bm{p}\Vert}\big)
\end{equation}
to conclude
\begin{equation}\label{loc1}
\left| \Vert \bm{p} \Vert \, m_2(z) - m_{\rm sc}\big(\tfrac{z}{\Vert \bm{p}\Vert}\big)\right|\prec \frac{1}{N|\Im z|}\,, \qquad |\Im z|\ge N^{-1+\xi}\,,
\end{equation}
using that $|\Im w_2|\gtrsim |\Im z|$ from the definition of $w_2$ in~\eqref{loclaw} and that $\Im z\cdot \Im m_2(z)>0$.
The conclusion~\eqref{loc1} requires a standard continuity argument, starting from a $z$ with a large imaginary part  and continuously
reducing the imaginary part by keeping the real part fixed (the same argument is routinely used in the proof of the local law for Wigner
matrices, see, e.g., \cite{EYbook}).  

The estimate~\eqref{loc1} implies that the quantiles of $\rho_2$ satisfy the usual rigidity estimate, i.e. 
\begin{equation}\label{rig}
  |\gamma_{2,i} -\gamma_i|\prec \frac{1}{N}
\end{equation}
for bulk indices $i\in [\kappa N, (1-\kappa)N]$ with any $N$-independent $\kappa>0$.
Moreover, \eqref{loc1}  also  implies that for any $z$ in the bulk of the semicircle, i.e. $|\Im m_{sc}(z)|\ge c>0$
for some $c>0$, independent of $N$, we have $|\Im m_2(z)|\ge c/2$ as long as $|\Im z|\ge N^{-1+\xi}$.
Using the definition of $w_2$ in~\eqref{loclaw} again, this shows $|\Im w_2|\sim |\Im w|$ 
for the deterministic $w:=\frac{1}{p_2}(z+p_1^2 m_{\rm sc}(z))$ for any $z$ with $|\Im z|\ge N^{-1+\xi}$.
Feeding this information into~\eqref{apprloclaw} and viewing it again as a perturbation of~\eqref{msc} but with the 
improved deterministic bound $\mathcal{O}_\prec(1/(N|\Im w|))$, we obtain
\begin{equation}\label{loc2}
\left| \Vert \bm{p} \Vert \,  m_2(z) - m_{\rm sc}\big(\tfrac{z}{\Vert \bm{p}\Vert}\big)\right|\prec \frac{1}{N|\Im w|}\,, \qquad \mbox{with}\quad
 w=\frac{1}{p_2}(z+p_1^2 m_{\rm sc}(z))\,,
\end{equation}
uniformly in $|\Im z|\ge N^{-1+\xi}$. In particular, when $z$ is in the bulk of the semicircle, then we have that
$$\left| \Vert \bm{p} \Vert \, m_2(z) - m_{\rm sc}\big(\tfrac{z}{\Vert \bm{p}\Vert}\big)\right|\prec \frac{1}{N}$$
 and this relation holds even down to the real axis by
the Lipschitz continuity (in fact, real analyticity) of the Stieltjes transform $m_2(z)$ in the bulk.

In the following, we will use the shorthand notation $A\approx B$ for two (families of) random
variables $A$ and $B$ if and only if $\vert A - B \vert \prec N^{-1}$.
Evaluating~\eqref{eq:MDE equipart} at $z=\gamma_{i,2}$, we have
\begin{equation}
	M_2(\gamma_{i,2})=\frac{1}{p_2}\cdot \frac{1}{W_2-w_{i,2}}\,,
	\qquad w_{i,2}:=\frac{1}{p_2}(\gamma_{i,2}+p_1^2\,m_2(\gamma_{i,2}))\,,
	\label{eq: expectation in equipart 2}
\end{equation}
and note that $w_{i,2} \approx w_i:=\frac{1}{p_2}\left(\gamma_{i}+\frac{p_1^2}{\lVert p\rVert}\,m_{\rm sc}\left(\frac{\gamma_{i}}{\lVert p\rVert}\right)\right)$ by~\eqref{loc2} and since $\gamma_{i,2}\approx\gamma_i$  in the bulk by rigidity~\eqref{rig}.

Now we are ready to evaluate the rhs. of \eqref{eq: expectation in equipart 1}.
By elementary manipulations using \eqref{eq: expectation in equipart 2}, 
we can now write the rhs.~of \eqref{eq: expectation in equipart 1} as
 \begin{equation} \label{eq:exp easyman}
\frac{\langle p_2W_2\Im M_2(\gamma_{i,2})\rangle}{\langle \Im M_2(\gamma_{i,2})\rangle} = 
\gamma_{i,2} + \frac{p_1^2}{p_2} \frac{\Im  \big[ \langle (W_2 - w_{i,2})^{-1} \rangle^2\big] }{\Im \, \langle (W_2 - w_{i,2})^{-1} \rangle}\,. 
 \end{equation}
 
 Using~\eqref{loc2}, we obtain
\begin{equation} 	\label{eq: expectation in equipart 4}
	\langle  (W_2-w_{i,2})^{-1}\rangle \approx \langle (W_2-w_i)^{-1}\rangle \approx m_{\rm sc}(w_i) 
\end{equation} 
with very high $W_2$-probability.
Continuing with \eqref{eq:exp easyman} and using $\gamma_{i,2} \approx \gamma_i$, we thus find 
\begin{equation}  \label{eq:exp msc}
\frac{\langle p_2W_2\Im M_2(\gamma_{i,2}))\rangle}{\langle \Im M_2(\gamma_{i,2})\rangle} \approx \gamma_{i} + \frac{p_1^2}{p_2}
\frac{\Im \big[ m_{\rm sc}(w_i)^2\big]}{\Im m_{\rm sc}(w_i)}\,. 
\end{equation}

Next, we combine  \eqref{eq:mdesc} with $p_2 m_2(\gamma_{i}) \approx m_{\rm sc}(w_i)$ from 
\eqref{loclaw}, \eqref{eq: expectation in equipart 2}
and \eqref{eq: expectation in equipart 4} and find that 
\begin{equation} 	\label{eq: expectation in equipart 8}
	m_{\rm sc}(w_i)^2 \approx -\frac{p_2^2}{p_1^2+p_2^2}\left(1+\frac{1}{p_2}\gamma_{i} \, m_{\rm sc}(w_i)\right)\,.
\end{equation}
Hence, plugging \eqref{eq: expectation in equipart 8} into \eqref{eq:exp msc} we deduce
\begin{equation*}
\frac{\langle p_2W_2\Im M_2(\gamma_{i,2}))\rangle}{\langle \Im M_2(\gamma_{i,2})\rangle} \approx  \left(1 - \frac{p_1^2}{p_1^2+p_2^2}\right)\gamma_{i} = \frac{p_2^2}{\|\boldsymbol{p}\|^2}\gamma_i\,.
\end{equation*}
This completes the proof of \eqref{eq:expectation}.

\subsection{Computation of the variance \eqref{eq:variance}} 
As in the calculation of the expectation in Section~\ref{subsec:exp}, we first condition on $W_2$ and work in the probability space of $W_1$.
 So, we apply Theorem \ref{theo:newCLT} to the matrix $p_1W_1+p_2W_2$, where the second
  term is considered a fixed deterministic deformation. 
  Indeed, using the same notations as in Section~\ref{subsec:exp}, this gives that the lhs.~of \eqref{eq:variance} equals 
\begin{equation}
    p_2^2\, \Var_{\gamma_{i,2}}\left(W_2\right)= p_2^2\frac{1}{\langle \Im M_2 (\gamma_{i,2})\rangle^2} \left( \big\langle \big(\mathring{W}_2^{\gamma_{i,2}} \Im M_2(\gamma_{i,2})\big)^2\big\rangle - \frac{p_1^2}{2} \Re\left[ \frac{\big\langle \big(M_2(\gamma_{i,2})\big)^2 \mathring{W}_2^{\gamma_{i,2}} \big\rangle^2}{1 - p_1^2\big\langle \big(M_2(\gamma_{i,2})\big)^2 \big\rangle} \right]  \right)
    \label{eq:variace equipart explicit}
\end{equation}
up to an additive error of order $\mathcal{O}_\prec \big(N^{-\epsilon}\big)$, which will appear on the rhs.~of \eqref{eq:variance}. The factor $p_1^2$ in the second term of \eqref{eq:variace equipart explicit} is a natural rescaling caused by applying 
Theorem~\ref{theo:newCLT} to a deformation of $p_1W_1$ instead of 
a Wigner matrix $W_1$.  Further we express $M_2$ in terms of a Wigner resolvent $G:=(W_2-w)^{-1}$ and use 
 \emph{local laws} not only for a single resolvent  $\langle G\rangle$
but also their multi-resolvent versions for  $\langle G^2\rangle$ and $\langle GG^*\rangle$ (see \cite{multiG}).  
With a slight abuse of notation we shall henceforth drop the subscript `$2$' in $\gamma_{i,2}$ and $w_{i,2}$ and replace them 
by their deterministic values $\gamma_i$ and $w_i$, respectively,  at a negligible error of order $N^{-1}$
 exactly as in Section~\ref{subsec:exp}. Note that $\Im w_i\gtrsim 1$
for bulk indices $i$, 
so all resolvents  below are stable and all denominators are well separated away from zero; this is needed to justify 
the $\approx$ relations below.

The first term in \eqref{eq:variace equipart explicit} can be rewritten as   (here $G=G(w_i)$ and $m_{\rm sc}:=m_{\rm sc}(w_i)$ for brevity)
\begin{align}
    \left\langle\big(\mathring{W}_2^{\gamma_i} \Im M_2(\gamma_i)\big)^2\right\rangle &\approx \frac{1}{2} \Re\left(\left\vert \frac{w_i}{p_2}-\frac{\gamma_i}{p_1^2+p_2^2}\right\vert^2\langle GG^*\rangle-\left(\frac{w_i}{p_2}-\frac{\gamma_i}{p_1^2+p_2^2}\right)^2\langle G^2\rangle\right)\nonumber\\
    &\approx \frac{1}{2}\Re\left(\left\vert \frac{w_i}{p_2}-\frac{\gamma_i}{p_1^2+p_2^2}\right\vert^2
    \frac{\vert m_{\rm sc}\vert^2}{1-\vert m_{\rm sc}\vert^2}-\left(\frac{w_i}{p_2}
    -\frac{\gamma_i}{p_1^2+p_2^2}\right)^2\frac{m_{\rm sc}^2}{1-m_{sc}^2}\right) \nonumber\\
    &\approx\frac{p_1^4}{2p_2^2(p_1^2+p_2^2)^2}\Re\left(\frac{1}{1-\vert m_{\rm sc}\vert^2} - \frac{1}{1-m_{\rm sc}^2} \right),
    \label{eq:variance term1}
\end{align}
where in the last step we used \eqref{eq: expectation in equipart 8}. Similarly for the second term in \eqref{eq:variace equipart explicit},
we have
\begin{align}
     \frac{\left\langle\big(M_2(\gamma_i)\big)^2 \mathring{W}_2^{\gamma_i}\right\rangle^2}{1 - p_1^2\left\langle \big(M(\gamma_i)\big)^2 \right\rangle} &\approx \frac{ \Big[ \frac{1}{p_2}\left\langle \frac{1}{p_2}G
     +\left(\frac{w_i}{p_2}-\frac{\gamma_i}{p_1^2+p_2^2}\right)G^2\right\rangle\Big]^2}{1-\frac{p_1^2}{p_2^2}\langle G^2\rangle}   \approx \frac{m_{\rm sc}^2(p_2^2-(p_1^2+p_2^2)m_{\rm sc}^2)}{p_2^2(p_1^2+p_2^2)^2(1-m_{\rm sc}^2)}\,.    \label{eq:variance term2}
\end{align}
Plugging \eqref{eq:variance term1} and \eqref{eq:variance term2} into \eqref{eq:variace equipart explicit} we obtain
\begin{equation} \label{eq:varwithoutident}
    p_2^2 \, \Var_{\gamma_i}(W_2)
    \approx \frac{2p_2^2+p_2\gamma_i\Re m_{\rm sc}}{\left( \Im m_{\rm sc}\right)^2}\cdot \frac{p_1^2p_2^2}{2(p_1^2+p_2^2)^2}\,.
\end{equation}
Taking the imaginary part of \eqref{eq: expectation in equipart 8}, we find that $|m_{\rm sc}|^2 \approx \frac{p_2^2}{p_1^2 + p_2^2}$ and hence, using \eqref{eq:mdesc} again, we infer
\begin{equation} \label{eq:varident}
\frac{1}{p_1^2 + p_2^2} \frac{2p_2^2+p_2\gamma_i\Re m_{sc}}{\left( \Im m_{\rm sc}\right)^2} \approx \frac{\Re [|m_{\rm sc}|^2 - m_{\rm sc}^2]}{(\Im m_{\rm sc})^2} = 2\,. 
\end{equation}
Combining \eqref{eq:varwithoutident} and \eqref{eq:varident} with \eqref{eq:variace equipart explicit}, this completes the proof of \eqref{eq:variance}. This proves Lemma~\ref{lem:main}.

\section{Multi--resolvent local laws: Proof of Theorem~\ref{theo:ETH}} \label{sec:proofETH}

To study the eigenvectors of $H$ we analyse its resolvent $G(z):=(H-z)^{-1}$, with $z\in \C\setminus\R$. It is well known \cite{slowcorr, edgelocallaw} that $G(z)$ becomes approximately deterministic in the large $N$ limit.
 Its deterministic approximation (as a matrix) is given by $M(z)$, the unique solution of \eqref{eq:MDE},
 in the following \emph{averaged} and \emph{isotropic} sense:
\begin{equation}
\label{eq:singlegllaw}
|\langle (G(z)-M(z))B\rangle|\prec \frac{1}{N|\Im z|}, \qquad |\langle{\bm x}\,, (G(z)-M(z)) {\bm y}\rangle|\prec \frac{1}{\sqrt{N|\Im z|}} \,,
\end{equation}
uniformly in deterministic vectors $\lVert {\bm x}\rVert+\lVert{\bm y}\rVert\lesssim 1$ and deterministic matrices 
$\lVert B\rVert\lesssim 1$. To be precise, while the local laws \eqref{eq:singlegllaw} hold for 
$\Re z \in \mathbf{B}_\kappa$ and $\mathrm{dist}(\Re z, \mathrm{supp}(\rho)) \gtrsim 1$ for \emph{arbitrary} 
bounded self-adjoint deformations $D =D^*$ (see \cite[Theorem~2.2]{slowcorr}), the complementary regime 
requires the strengthened Assumption \ref{ass:Mbdd} on $D$ (see \cite{edgelocallaw}). Note that cusps for $\rho$ 
have been excluded in Assumption \ref{ass:Mbdd}, hence the complementary regime only consists of edges, 
which are covered in \cite[Theorem~2.6]{edgelocallaw}, under the requirement that $\Vert M \Vert$ is bounded 
-- which was also supposed in Assumption \ref{ass:Mbdd}. 

The isotropic bound $\langle{\bm x}, \Im G(z) {\bm x}\rangle\prec 1$ from \eqref{eq:singlegllaw} immediately
 gives an (almost) optimal bound on the delocalisation of eigenvectors: $|\langle {\bm u}_i,{\bm x}\rangle|\prec N^{-1/2}$ 
\cite{LeeSchnelli, LandonYau, slowcorr, edgelocallaw, BenigniLopatto2103.12013}. However, these estimates are not precise
 enough to conclude optimal bounds for eigenvector overlaps and generic matrices $A$ as in Theorem~\ref{theo:ETH}; 
 in fact by \eqref{eq:singlegllaw} we can only obtain the trivial bound $|\langle {\bm u}_i,A{\bm u}_j\rangle|\prec 1$.
Instead of the single resolvent local law \eqref{eq:singlegllaw}, we rely on the fact that (see \eqref{eq:overlap2G} below)
\begin{equation}\label{uAu}
N \big|\langle {\bm u}_i,A{\bm u}_j\rangle \big|^2\lesssim \langle \Im G(\gamma_i+\ii\eta)A \Im G(\gamma_j+\ii\eta)A^*\rangle\,,
\end{equation}
for $\eta\sim N^{-1+\epsilon}$, where $\epsilon>0$ is small but fixed, and $\gamma_i, \gamma_j \in \mathbf{B}_\kappa$ are in the bulk
and we estimate the rhs. of~\eqref{uAu}.  In particular, to prove Theorem~\ref{theo:ETH} we will use the \emph{multi-resolvent local laws} from Proposition~\ref{pro:mresllaw} below. 

Multi-resolvent local laws are natural generalisations of~\eqref{eq:singlegllaw} and they
 assert that longer products 
 \begin{equation}\label{GBG}
	G_1 B_1 G_2 \,  \cdots \, G_{k-1} B_{k-1} G_{k}
\end{equation}
 of resolvents  $G_i := G(z_i)$ and deterministic matrices\footnote{We will use the the notational convention, that the letter $B$
  denotes arbitrary (generic) matrices, while $A$ is reserved for \emph{regular} matrices, in the sense of Definition \ref{def:regobs} below. }
   $B_1, ... , B_{k-1}$
  also become approximately deterministic
both in average and isotropic sense in the large $N$ limit as long as $N|\Im z_i|\gg 1$.
The deterministic approximation to the  chain~\eqref{GBG}
 is denoted by
\begin{equation} \label{eq:Mdef}
	M(z_1, B_1, z_2, ... , z_{k-1}, B_{k-1}, z_{k}).
\end{equation}
It is not simply $M(z_1) B_1 M(z_2) B_2\ldots$, i.e. it cannot be obtained by 
mechanically replacing each $G$ with $M$ as~\eqref{eq:singlegllaw} might incorrectly suggest.
 Instead, \nc
it is defined recursively in the length $k$ of the chain as follows (see~\cite[Definition 4.1]{iid}): 
\begin{definition} \label{def:Mdef}
	Fix $k \in \N$ and let $z_1, ... , z_k \in \C \setminus \R$ be spectral parameters. As usual, the corresponding solutions to \eqref{eq:MDE} are denoted by $M(z_j)$, $j \in [k]$. Then, for deterministic matrices $B_1, ... , B_{k-1}$ we recursively define 
	\begin{align}
		M(z_1,B_1, ... B_{k-1} , z_{k}) = \big(\mathcal{B}_{1k}\big)^{-1}&\bigg[M(z_1) B_1 M(z_{2}, ...  , z_{k}) \label{eq:M_definitionapp}\\
		& +  \sum_{l = 2}^{k-1}  M(z_1) \langle M(z_1,  ... , z_l)  \rangle M(z_l, ... , z_{k}) \bigg]\,, \nonumber
	\end{align}
	where we introduced the shorthand notation
	\begin{equation} \label{eq:stabop}
		\mathcal{B}_{mn} \equiv \mathcal{B}(z_m,z_n)= 1 - M(z_m) \langle \cdot \rangle M(z_n)
	\end{equation}
	for the stability operator acting on the space of $N\times N$ matrices.
\end{definition}

It turns out that the size of $M(z_1, B_1, z_2,\ldots, z_k)$ in the relevant regime of small $\eta:=\min_j |\Im z_j|$
is roughly $\eta^{-k+1}$ in the worst case, with a matching error term in the corresponding local  law.
This blow-up in the small $\eta$ regime comes recursively  from the large norm of the inverse of the 
stability operator $\mathcal{B}_{1k}$ in~\eqref{eq:M_definitionapp}. However, for a special subspace
of observable matrices $B_i$, called \emph{regular} matrices, the size of $M(z_1, B_1, z_2,\ldots, z_k)$ is much smaller.
 For \emph{Wigner matrices}, i.e.~for $D=0$, the {regular observables} are simply 
 the \emph{traceless} matrices, i.e. observables $B$ such that $\langle B\rangle=0$. 
 In \cite{ETHpaper, thermalization, multiG, A2} it was shown that when the 
matrices $B_i$ are all traceless, then $M(z_1, B_1, z_2,\ldots, z_k)$ hence~\eqref{GBG} are
 smaller by  an $\eta^{k/2}$-factor
 than for general $B_i$;'s.

The situation for deformed Wigner matrices is more complicated, since the concept of \emph{regular observables} will be dependent on the precise location in the spectrum of $H$, i.e.~dependent on the energy. More precisely, we will require that the trace of $A$ tested against a deterministic energy dependent matrix has to vanish; this reflects the inhomogeneity introduced by $D$. Analogously to the Wigner case, 
in Proposition~\ref{pro:mresllaw} below we will show that  resolvent chains~\eqref{GBG} 
are much smaller when the deterministic matrices $B_i$ are regular.

\nc

Next, we give the definition of regular matrices in the chain~\eqref{GBG}. Using the notation $A$ for regular matrices,
we will consider chains of resolvents and deterministic matrices of the form
\begin{equation} \label{eq:av}
	\langle G_1 A_1 \, \cdots \, G_k A_k \rangle 
\end{equation}
in the averaged case, or
\begin{equation} \label{eq:iso}
	\big(G_1 A_1 \, \cdots \, A_k G_{k+1} \big)_{\boldsymbol{x}\boldsymbol{y}}
\end{equation}
in the isotropic case, with $G_i:=G(z_i)$ and $A_i$ being \emph{regular} matrices according to the following 
Definition~\ref{def:regobs} (cf. \cite[Definition 4.2]{iid}), which generalises the earlier Definition \ref{def:regobs1}.

\begin{definition}[Regular observables -- Two-point regularisation in chains]
\label{def:regobs}
Fix a parameter $\kappa > 0$ and let $\delta = \delta(\kappa, \Vert D\Vert) > 0$ be small enough (see the discussion below). Consider one of the two expressions \eqref{eq:av} or \eqref{eq:iso} for some fixed length $k \in \N$ and bounded matrices $\Vert A_i \Vert \lesssim 1$ and let  $z_1, ... , z_{k+1} \in \C \setminus \R$ be spectral parameters with $\Re z_j \in \mathbf{B}_\kappa$.
For any $j \in [k]$, we denote
\begin{equation} \label{eq:case regulation2}
	\mathbf{1}_\delta(z_j, z_{j+1}) := \phi_\delta(\Re z_j - \Re z_{j+1} ) \ \phi_\delta(\Im z_j) \ \phi_\delta(\Im z_{j+1})
\end{equation}
where $0 \le \phi_\delta \le 1$ is a smooth symmetric bump function on $\R$ satisfying $\phi_\delta(x) = 1$ for $|x|\le \delta/2$ and $\phi_\delta(x) = 0$ for $|x| \ge \delta$. Here and in the following, in case of \eqref{eq:av}, the indices in \eqref{eq:case regulation2} are understood cyclically modulo $k$. 
\begin{itemize}
	\item[(a)] For $j \in [k]$, denoting $\mathfrak{s}_j := - \mathrm{sgn}(\Im z_j \Im z_{j+1})$, we define the \emph{(two-point) regularisation} of $A_j$ from \eqref{eq:av} or \eqref{eq:iso} \emph{w.r.t.~the spectral parameters $(z_j, z_{j+1})$} as
	\begin{equation} \label{eq:circ def}
\mathring{A}_j^{{z_j,z_{j+1}}} := A_j - 	\mathbf{1}_\delta(z_j, z_{j+1})\frac{\langle M(\Re z_j+ \ii \Im z_j) A_j M(\Re z_{j+1} + \mathfrak{s}_j \ii \Im z_{j+1}) \rangle}{\langle M(\Re z_j + \ii \Im z_j) M(\Re z_{j+1} + \mathfrak{s}_j \ii \Im z_{j+1}) \rangle}  \,. 
	\end{equation}
	\item[(b)] Moreover, we call $A_j$ \emph{regular w.r.t.~$(z_j, z_{j+1})$} if and only if $\mathring{A}^{z_j,z_{j+1}}_j = A_j$. 
\end{itemize}
\end{definition}
As already indicated above, the two-point regularisation generalises Definition \ref{def:regobs1} in the sense that
\begin{equation} \label{eq:1ptreg2ptreg}
\mathring{A}^{e \pm \ii \eta, e \pm \ii \eta} \longrightarrow \mathring{A}^e\,, \quad \text{and} \quad \mathring{A}^{e \pm \ii \eta, e \mp \ii \eta} \longrightarrow \mathring{A}^e \,, \quad \text{as} \quad \eta \downarrow 0\,,
\end{equation}
with a linear speed of convergence,
for $e \in \mathbf{B}_\kappa$ and any bounded deterministic $A \in \C^{N \times N}$, where we used that,  by taking the imaginary part of \eqref{eq:MDE}, $M(z)M(z)^* = \Im M(z)/(\langle \Im M(z)\rangle + \Im z)$. 

Moreover, we point out, that the above Definition \ref{def:regobs} of the regularisation is identical to 
\cite[Defs.~3.1 and 4.2]{iid} when dropping the summand with $\mathfrak{s} \tau = -1$ in Equation~(3.7) of \cite{iid}. 
In particular, for spectral parameters $z_j, z_{j+1}$ satisfying $\mathbf{1}_\delta(z_j,z_{j+1}) > 0$ (for some $\delta >0 $ small enough), it holds that the denominator in \eqref{eq:circ def} is bounded away from zero, which shows that the linear map $A \mapsto \mathring{A}$ is bounded. Additionally, we have the following Lipschitz property (see \cite[Lemma~3.3]{iid}):
\begin{equation}
\label{eq:lipprop}
\mathring{A}^{z_1,z_2}=\mathring{A}^{w_1,w_2}+\mathcal{O}\big(|z_1-w_1|+|z_2-w_2|\big) I,
\end{equation}
for any $z_1,z_2,w_1,w_2\in \C\setminus\R$ such that $\Im z_i\Im w_i>0$. 
It is important that the error in~\eqref{eq:lipprop} is a constant times the identity matrix, 
indicated by $\mathcal{O}(\cdot) I$.

Next, we give bounds on the size of $M(z_1, A_1, ... A_{k-1}, z_{k})$, the deterministic approximation to
the chain  $G_1 A_1 \,  \cdots \, A_{k-1}G_{k}$ introduced in Definition~\ref{def:Mdef}; the proof of this lemma is presented in Appendix~\ref{app:techlem}.

 \begin{lemma}
 \label{lem:Mbound}
 	Fix $\kappa > 0$. Let $k \in [4]$ and $z_1, ... , z_{k+1} \in \C \setminus \R$ be spectral parameters with $\Re z_j \in \mathbf{B}_\kappa$.
	Set $\eta:=\min_j |\Im z_j|$. 
 	Then, for bounded \emph{regular} deterministic matrices $A_1, ... , A_{k}$  (according to Definition~\ref{def:regobs}), we have the bounds
	\begin{align}
		\Vert M(z_1, A_1, ... , A_{k}, z_{k+1}) \Vert &\lesssim\begin{cases}
			\frac{1}{\eta^{\lfloor k/2 \rfloor}} \hspace{11.5mm} &\text{if} \ \eta \le 1 \\
			\frac{1}{\eta^{k+1}} \qquad &\text{if} \ \eta > 1
		\end{cases} \,, \label{eq:Mboundnorm}\\[2mm]
		\vert \langle  M(z_1, A_1, ... , A_{k-1}, z_{k})A_k  \rangle \vert & \lesssim \begin{cases}
			\frac{1}{\eta^{\lfloor k/2\rfloor-1 }}\vee 1 \quad &\text{if} \ \eta \le 1 \\
			\frac{1}{\eta^{k}} \quad &\text{if} \ \eta > 1
		\end{cases}\,. \label{eq:Mboundtrace} 
	\end{align}
\end{lemma}
For the presentation of Proposition \ref{pro:mresllaw}, the main technical result underlying the proof of Theorem~\ref{theo:ETH}, we would only need \eqref{eq:Mboundnorm} and \eqref{eq:Mboundtrace} for $k \in [2]$ from the previous lemma. However, the remaining bounds covered by Lemma \ref{lem:Mbound} will be instrumental in several parts of our proofs (see Section \ref{sec:llaw} and Appendix \ref{app:techlem}).

\begin{proposition}
\label{pro:mresllaw}
Fix $\epsilon>0$, $\kappa > 0$, $k\in [2]$, and consider $z_1,\dots,z_{k+1} \in \C \setminus \R$ with $\Re z_j \in \mathbf{B}_\kappa$. Consider regular matrices $A_1,\dots,A_k$ with $\lVert A_i\rVert\le 1$, deterministic vectors ${\bm x}, {\bm y}$ with $\lVert {\bm x}\rVert +\lVert {\bm y} \rVert\lesssim 1$, and set $G_i:=G(z_i)$. Then, uniformly in 
$\eta:=\min_j |\Im z_j|\ge N^{-1+\epsilon}$, we have the \emph{averaged local law}
\begin{subequations}
	\begin{equation}
\label{eq:avellaw}
\big|\langle \big(G_1A_1\dots G_k-M(z_1, A_1, ... ,  z_{k})\big)A_k \rangle\big|\prec \begin{cases}
	\frac{N^{k/2-1}}{\sqrt{N \eta}} \quad &\text{if} \ \eta \le 1 \\
	\frac{1}{N \eta^{k+1}} \quad &\text{if} \ \eta > 1
\end{cases}
	\end{equation}
and the \emph{isotropic local law}
\begin{equation}
\label{eq:isollaw}
\big|\langle {\bm x},\big(G_1A_1\dots G_{k+1}-M(z_1, A_1, ... , z_{k+1})\big) {\bm y} \rangle\big|\prec
 \begin{cases}
	\frac{N^{(k-1)/2}}{\sqrt{N \eta^2}} \quad &\text{if} \ \eta \le 1 \\
	\frac{1}{\sqrt{N} \eta^{k+2}} \quad &\text{if} \ \eta > 1
\end{cases}\,.
\end{equation}
\end{subequations} 
\end{proposition}
In Section \ref{sec:llaw}, we will carry out the proof of Proposition \ref{pro:mresllaw} in the much more involved $\eta \le 1$ regime. For $\eta > 1$, the bound simply follows by induction on the number of resolvents in a chain by invoking the trivial estimate $\Vert M(z) \Vert \lesssim 1/|\Im z|$. The detailed argument has been carried out in \cite[Appendix~B]{multiG} for the case of Wigner matrices. 
Having Proposition \ref{pro:mresllaw} at hand, we can now prove Theorem \ref{theo:ETH}.

\begin{proof}[Proof of Theorem~\ref{theo:ETH}]
By \eqref{eq:avellaw} and \eqref{eq:Mboundtrace} for $k=2$ it follows that
\begin{equation} \label{eq:2Gbddbulk}
|\langle G_1A_1G_2A_2 \rangle|\prec 1\,,
\end{equation}
for arbitrary regular matrices $A_1 = \mathring{A}_1^{z_1, z_2}$ and $A_2 = \mathring{A}_2^{z_2, z_1}$. Now, using that (see \cite[Lemma~3.6]{iid} for an analogous statement; see also \eqref{eq:1ptreg2ptreg} and \eqref{eq:lipprop})
\begin{equation*} 
		\mathring{A}^{\gamma_i} = \mathring{A}^{{\gamma_i \pm \ii \eta, \gamma_j \pm 2\ii \eta}} 
		+ \mathcal{O}\big(|\gamma_i - \gamma_j|+ \eta\big) I
		= \mathring{A}^{{\gamma_i \pm \ii \eta, \gamma_j \mp2\ii \eta}} + \mathcal{O}\big(|\gamma_i - \gamma_j|+ \eta\big)I\,,
\end{equation*}
and analogously for $(\mathring{A}^*)^{\gamma_i}$, we obtain (cf.~\cite[Sec. 3.3.1]{iid})
\begin{equation*}
 \langle \Im G(\gamma_i + \ii \eta) \mathring{A}^{\gamma_i} \Im G(\gamma_j + 2\ii \eta)  (\mathring{A}^*)^{\gamma_i} \rangle \prec 1\,. 
\end{equation*}
Moreover, by spectral decomposition, together with the rigidity of eigenvalues (see e.g. \cite{firstcorr, slowcorr}) it follows that (cf.~\cite[Lemma~3.5]{iid})
\begin{equation} \label{eq:overlap2G}
N | \langle  {\bm u}_i,  \mathring{A}^{\gamma_i}  {\bm u}_j \rangle|^2 \prec (N \eta)^2\langle \Im G(\gamma_i + \ii \eta) \mathring{A}^{\gamma_i} \Im G(\gamma_j + 2\ii \eta)  (\mathring{A}^*)^{\gamma_i} \rangle \prec (N \eta)^2\,. 
\end{equation}
Choosing $\eta=N^{-1+\xi/2}$ for some arbitrary small $\xi>0$, we conclude the desired. 
\end{proof}

\section{Dyson Brownian motion: Proof of Theorem~\ref{theo:newCLT}}
\label{sec:proofETHfluct}

The main observation we used to prove Theorem~\ref{theo:ETH} in Section~\ref{sec:proofETH} is the relation \eqref{uAu}, i.e. we related the eigenvector overlaps with a trace of the product of two resolvents and two deterministic matrices. For Theorem~\ref{theo:ETH} we only needed an upper bound on the size of the eigenvector overlaps, however to prove Theorem~\ref{theo:newCLT} we need to identify their size. For this purpose the main input is the relation
\begin{equation}
\label{eq:relres}
\frac{1}{N^{2\epsilon}} \sum_{|i-i_0|\le N^\epsilon \atop |j-j_0|\le N^\epsilon}N| \langle  {\bm u}_i,  \mathring{A}^{\gamma_i}  {\bm u}_j \rangle|^2 \sim \langle \Im G(\gamma_{i_0} + \ii \eta) \mathring{A}^{\gamma_{i_0}} \Im G(\gamma_{j_0} + 2\ii \eta)  (\mathring{A}^*)^{\gamma_{i_0}}\rangle,
\end{equation}
with $\eta=N^{-1+\epsilon}$, for some small fixed $\epsilon>0$, and $i_0,j_0$ being some fixed bulk indices. The relation \eqref{eq:relres} is clearly not enough to identify the fluctuations of the individual eigenvector overlaps,  but it gives a hint on the expression of the variance of these overlaps. More precisely, to identify the fluctuations of $N | \langle  {\bm u}_i,  \mathring{A}^{\gamma_i}  {\bm u}_i\rangle|^2$ we will rely on a Dyson Brownian motion analysis which will reveal that
\begin{equation}
\label{eq:reltores}
N \E [| \langle  {\bm u}_i,  \mathring{A}^{\gamma_i}  {\bm u}_i \rangle|^2] \approx \frac{1}{\langle \Im M(\gamma_i) \rangle^2} \E\langle \Im G(\gamma_i + \ii \eta) \mathring{A}^{\gamma_i} \Im G(\gamma_i + 2\ii \eta)  (\mathring{A}^*)^{\gamma_i}\rangle,
\end{equation}
and a similar relation holds for higher moments as well. Finally, the rhs. of \eqref{eq:reltores} is computed using a multi--resolvent local law (see e.g. \eqref{eq:avellaw} for $k=2$), and after some algebraic manipulation (see \eqref{eq:longcomp}--\eqref{eq:explain} below) this results in $\mathrm{Var}_{\gamma_i}(A)$ as defined in \eqref{eq:VarAgammai}.

Given the optimal a priori bound \eqref{eq:ETH}, the proof of Theorem~\ref{theo:newCLT} is very similar to the analysis of the \emph{Stochastic Eigenstate Equation (SEE)} in \cite[Sections 3-4]{normalfluc} and \cite[Section 4]{A2}. Even if very similar to those papers, to make the presentation clearer, here we write out the main steps of the proof and explain the differences, but we do not write the details; we defer the interested reader to \cite{normalfluc}. We also remark that the proof in \cite{normalfluc, A2} heavily relies on the analysis of SEE developed in \cite{BourgadeYau1312.1301} and extend in \cite{BourgadeYauYin1807.01559, MarcinekYau2005.08425}.

Similarly to \cite{normalfluc, A2} we only consider the real case, the complex case is completely analogous and so omitted. We prove Theorem~\ref{theo:newCLT} dynamically, i.e. we consider the flow
\begin{equation}
\label{eq:matDBM}
\dif W_t=\frac{\dif \widetilde{B}_t}{\sqrt{N}}\,, \qquad W_0=W\,,
\end{equation}
with $\widetilde{B}_t$ a real symmetric matrix valued Brownian motion (see e.g. \cite[Definition 2.1]{BourgadeYau1312.1301}). Note that $W_t$ has a Gaussian component of size $\sqrt{t}$, i.e.
\[
W_t\stackrel{\dif}{=}W_0+\sqrt{t}U\,,
\]
with $U$ being a GOE matrix independent of $W_0$. Denoting by $\lambda_i(t)$ the eigenvalues of $W_t$ (labeled in increasing order) and by ${\bm u}_i(t)$ the corresponding orthonormal eigenvectors, we will prove Theorem~\ref{theo:newCLT} for the eigenvectors ${\bm u}_i(T)$, with $T=N^{-1+\omega}$, for some small fixed $\omega>0$. Since $T$ is very small, the Gaussian component added in the flow \eqref{eq:matDBM} can easily be removed by a standard Green function comparison (GFT) argument as in \cite[Appendix B]{A2}.

By \cite{BourgadeYau1312.1301}, it is known that the eigenvalues $\lambda_i(t)$ and the eigenvectors ${\bm u}_i(t)$ are the unique strong solution of the following system of stochastic differential equations (SDEs):
\begin{align}
    \label{eq:evaluflow}
    \dif \lambda_i(t)&=\frac{\dif B_{ii}(t)}{\sqrt{N}}+\frac{1}{N}\sum_{j\ne i} \frac{1}{\lambda_i(t)-\lambda_j(t)} \dif t \\\label{eq:evectorflow}
    \dif {\bm u}_i(t)&=\frac{1}{\sqrt{N}}\sum_{j\ne i} \frac{\dif B_{ij}(t)}{\lambda_i(t)-\lambda_j(t)}{\bm u}_j(t)-\frac{1}{2N}\sum_{j\ne i} \frac{{\bm u}_i}{(\lambda_i(t)-\lambda_j(t))^2}\dif t\,, 
\end{align}
where the matrix $B(t)=(B_{ij}(t))_{i,j=1}^N$ is a standard real symmetric Brownian motion (see e.g. \cite[Definition 2.1]{BourgadeYau1312.1301}).

Even if in Theorem~\ref{theo:newCLT} we want to prove a CLT only for diagonal overlaps $\langle {\bm u}_i, A{\bm u}_i\rangle$, by \eqref{eq:evectorflow}, it follows that there is no closed equation for such quantities. For this reason, following \cite[Section 2.3]{BourgadeYauYin1807.01559}, we study the evolution of the \emph{perfect matching observable} (see \eqref{eq:deff} below) along the flow \eqref{eq:evectorflow}.

\subsection{Perfect matching observable and proof of Theorem \ref{theo:newCLT}}

We introduce the notation
\begin{equation}
\label{eq:defpij}
p_{ij}=p_{ij}(t)=\langle {\bm u}_i,A {\bm u}_j\rangle-\delta_{ij}C_0\,, 
\end{equation}
with $A$ being a fixed real symmetric deterministic matrix $A$ and $C_0$ being a fixed constant independent of $i$. Note that compared to \cite{normalfluc, A2} in \eqref{eq:defpij} we define the diagonal $p_{ii}$ without subtracting their expectation (see \eqref{eq:ETH} above), but rather a generic constant $C_0$ which we will choose later (see \eqref{eq:choicec0c1} below). The reason behind this choice is that in the current setting, unlike in  the Wigner case \cite{normalfluc, A2}, the expectation of $p_{ii}$ is now $i$--dependent, hence the flow \eqref{eq:1dequa} below would not be satisfied if we had defined \eqref{eq:deff} with the centred $p_{ii}$'s.

To study moments of the $p_{ij}$'s we use the \emph{particle representation} introduced in  \cite{BourgadeYau1312.1301} and further developed in \cite{BourgadeYauYin1807.01559, MarcinekYau2005.08425}.. A particle configuration, corresponding to a certain monomials of $p_{ij}$'s, can be encoded by a function $\bm{\eta}:[N]\to\N_0$. The image $\eta_j=\bm{\eta}(j)$ denotes the number of particle at the site $j$, and $\sum_j\eta_j=n$ denotes the total number of particles. Additionally, given a particle configuration $\bm{\eta}$, by $\bm{\eta}^{ij}$, with $i\ne j$, we denote a new particle configuration in which a particle at the site $i$ moved to a new site $j$, if there is no particle in $i$ then $\bm{\eta}^{ij}={\bm \eta}$. We denote the set of such configuration by $\Omega^n$.

Fix a configuration $\bm{\eta}$, then we define the \emph{perfect matching observable} (see \cite[Section 2.3]{BourgadeYauYin1807.01559}):
\begin{equation}
    \label{eq:deff}
    f_{{\bm \lambda},t,C_0,C_1}({\bm \eta}):= \frac{N^{n/2}}{ [2 C_1]^{n/2}} \frac{1}{(n-1)!!}\frac{1}{\mathcal{M}({\bm \eta}) }\E\left[\sum_{G\in\mathcal{G}_{\bm \eta}} P(G)\Bigg| 
    {\bm \lambda}\right] \,, \quad \mathcal{M}({\bm \eta}):=\prod_{i=1}^N (2\eta_i-1)!!\,,
\end{equation}
with $n$ being the total number of particles in the configuration \({\bm \eta}\).  The sum in \eqref{eq:deff} is taken over $\mathcal{G}_{\bm \eta}$, which denotes the set of perfect matchings on the complete graph with vertex set
\[
\mathcal{V}_{\bm \eta}:=\{(i,a): 1\le i\le n, 1\le a\le 2\eta_i\}\,. 
\]
We also introduced the short--hand notation
\begin{equation}
    \label{eq:pg}
    P(G):=\prod_{e\in\mathcal{E}(G)}p(e), \qquad p(e):=p_{i_1i_2}\,,
\end{equation}
where  $e=\{(i_1,a_1),(i_2,a_2)\}\in \mathcal{V}_{\bm \eta}^2$, and $\mathcal{E}(G)$ denotes the edges of $G$. Note that in \eqref{eq:deff} we took the conditional  expectation with respect to the entire trajectories of the eigenvalues, \({\bm \lambda} =\{\bm \lambda(t)\}_{t\in [0,T]}\) for some  fixed  $0<T\ll 1$. We also remark that the definition \eqref{eq:deff} differs slightly from \cite[Eq. (3.9)]{normalfluc} and \cite[Eq. (4.6)]{A2}, since we now do not normalise by $\langle (A-\langle A\rangle)^2\rangle$ but using a different constant $C_1$ which we will choose later in the proof (see \eqref{eq:choicec0c1} below); this is a consequence of the fact that the diagonal overlaps $p_{ii}$ are not correctly centred and normalised. Note that we did not incorporate the factor $2$ in \eqref{eq:deff} into the constant $C_1$, since $C_1$ will be chosen has a normalisation constant to compensate the size of the matrix $A$, whilst the factor $2$ represents the fact that diagonal overlaps, after the proper centering and normalisation depending on $A$ and $i$, would be centred Gaussian random variable of variance two. Furthermore, we consider eigenvalues paths $\{\bm \lambda(t)\}_{t\in [0,T]}$ which lie in the event 
\begin{equation}
\label{def:Omega}
    \widetilde{\Omega}=\widetilde{\Omega}_\xi,:= \Big\{ \sup_{0\le t \le T} \max_{i\in [N]} \eta_{\mathrm{f}}(\gamma_i(t))^{-1}| \lambda_i(t)-\gamma_i(t)| \le N^\xi\Big\}
\end{equation}
for any $\xi>0$, where $\eta_{\mathrm{f}}(\gamma_i(t))$ is the local fluctuation scale defined as in \cite[Definition 2.4]{cuspI}. In most instances we will use this rigidity estimate in the bulk regime when $\eta_{\mathrm{f}}(\gamma_i(t))\sim N^{-1}$; at the edges $\eta_{\mathrm{f}}(\gamma_i(t))\sim N^{-2/3}$. We recall that here $\gamma_i(t)$ denote the quantiles of $\rho_t$ defined as in \eqref{eq:quant}. The fact that the event $ \widetilde{\Omega}$ holds with very high probability follows by \cite[Corollary 2.9]{edgelocallaw}.

By~\cite[Theorem 2.6]{BourgadeYauYin1807.01559} it follows that $f_{{\bm \lambda},t}$ is a solution of the parabolic discrete partial differential equation (PDE):
\begin{align}
\label{eq:1dequa}
    \partial_t f_{{\bm \lambda},t}&=\mathcal{B}(t)f_{{\bm \lambda},t}\,, \\\label{eq:1dkernel}
    \mathcal{B}(t)f_{{\bm \lambda},t}&=\sum_{i\ne j} c_{ij}(t) 2\eta_i(1+2\eta_j)\big(f_{{\bm \lambda},t}({\bm \eta}^{ij})-f_{{\bm \lambda},t}({\bm \eta})\big)\,. 
\end{align}
where 
\begin{equation}\label{eq:defc}
    c_{ij}(t):= \frac{1}{N(\lambda_i(t) -  \lambda_j(t))^2}\,. 
\end{equation}
In the remainder of this section we may often omit $\bm{\lambda}$ from the notation since the paths of the eigenvalues are fixed within this proof.

The main result of this section is the following Proposition \ref{pro:flucque}, which will readily prove Theorem~\ref{theo:newCLT}. For this purpose we define a version of $f_t(\bm{\eta})$ with centred and rescaled $p_{ii}$:
\begin{equation}
\label{eq:defq}
q_{{\bm \lambda},t}({\bm \eta}):= \left(\prod_{i=1}^N\frac{1}{ \mathrm{Var}_{\gamma_i}(A)^{\eta_i/2}}\right) \frac{N^{n/2}}{2^{n/2} (n-1)!!}\frac{1}{\mathcal{M}({\bm \eta}) }\E\left[\sum_{G\in\mathcal{G}_{\bm \eta}} Q(G)\Bigg| 
    {\bm \lambda}\right]
\end{equation}
with $\mathring{A}^{\gamma_i}$ denoting the regular component of $A$ defined as in \eqref{eq:defreg1}:
\[
 \mathring{A}^{\gamma_i}:=A-\frac{\langle A\Im M(\gamma_i)\rangle }{\langle \Im M(\gamma_i)\rangle},
 \]
and
\begin{equation}
\label{eq:qg}
    Q(G):=\prod_{e\in\mathcal{E}(G)}q(e), \qquad q(e):=\langle {\bm u}_{i_1},\mathring{A}^{\gamma_{i_1}} {\bm u}_{i_2}\rangle.
\end{equation}
Note that the definition in \eqref{eq:qg} is not asymmetric for $i_1\ne i_2$, since in this case $\langle {\bm u}_{i_1},\mathring{A}^{\gamma_{i_1}} {\bm u}_{i_2}\rangle=\langle {\bm u}_{i_1},\mathring{A}^{\gamma_{i_2}} {\bm u}_{i_2}\rangle$. 

We now comment on the main difference between $q_t$ and $f_t$ from \eqref{eq:defq} and \eqref{eq:deff}, respectively. First of all we notice the $q(e)$'s in \eqref{eq:qg} are slightly different compared with the $p(e)$'s from \eqref{eq:defpij}. In particular, we choose the $q(e)$'s in such a way that the diagonal overlaps have very small expectation (i.e. much smaller than their fluctuations size). The price to pay for this choice is that the centering is $i$--dependent, hence $q_t$ is not a solution of an equation of the form \eqref{eq:1dequa}--\eqref{eq:1dkernel}. We also remark that later within the proof, $C_0$ from \eqref{eq:defpij} will be chosen as 
\[
C_0 =  \frac{\langle A\Im M(\gamma_{i_0})\rangle}{\langle \Im M(\gamma_{i_0})\rangle}
\]
 for some fixed $i_0$ such that $\gamma_{i_0}\in\mathbf{B}_\kappa$ is in the bulk (recall \eqref{eq:bulk}). The idea behind this choice is that the analysis of the flow \eqref{eq:1dequa}--\eqref{eq:1dkernel} will be completely local, we can thus fix a base point $i_0$ and ensure that the corresponding  overlap is exactly centred, then the nearby overlaps for indices $|i-i_0|\le K$, for some $N$--dependent $K>0$, will not be exactly centred, but their expectation will be very small compared to the size of their fluctuations:
\[
\frac{\langle A\Im M(\gamma_{i_0})\rangle}{\langle \Im M(\gamma_{i_0})\rangle}-\frac{\langle A\Im M(\gamma_i)\rangle}{\langle \Im M(\gamma_i)\rangle}=\mathcal{O}\left(\frac{K}{N}\right)\,. 
\]
A consequence of this choice is also that the normalisation for $q_t$ and $f_t$ is different: for $q_t$ we chose a normalisation that is $i$'s dependent, whilst for $f_t$ the normalisation $C_1$ is $i$--independent and later, consistently with the choice of $C_0$, it will be chosen as
\[
C_1=\mathrm{Var}_{\gamma_{i_0}}(A)^{n/2},
\]
which is exactly the normalisation that makes $f_t(\bm{\eta})=1$ when $\bm{\eta}$ is such that $\eta_{i_0}=n$ and zero otherwise.

\begin{proposition}\label{pro:flucque}
    For any $n\in\mathbf{N}$ there exists $c(n)>0$ such that for any $\epsilon>0$, and for any $T\ge N^{-1+\epsilon}$ it holds
    \begin{equation}\label{eq:mainbthissec}
        \sup_{{\bm \eta}}\left|q_T({\bm\eta})- \bm1(n\,\, \mathrm{even})\right|\lesssim N^{-c(n)}\,,
    \end{equation}
    with very high probability.  The supremum is taken over configurations ${\bm \eta}$ supported on bulk indices and the implicit constant in~\eqref{eq:mainbthissec} depends on $n$ and $\epsilon$.
\end{proposition}

\begin{proof}[Proof of Theorem~\ref{theo:newCLT}]

Fix $n\in\N$, an index $i$ such that $\gamma_i \in \mathbf{B}_\kappa$ is in the bulk, and choose a configuration $\bm{\eta}$ such that $\eta_i=n$ and $\eta_j=0$ for any $j\ne i$. Then by Proposition~\ref{pro:flucque}, we conclude that
\[
\E\left[\sqrt{\frac{N}{2\mathrm{Var}_{\gamma_i}(A)}}\langle {\bm u}_i(T),\mathring{A}^{\gamma_i}{\bm u_i}(T)\rangle\right]^n=\bm1(n\,\, \mathrm{even})(n-1)!!+\mathcal{O}\left( N^{-c(n)}\right),
\]
with $T=N^{-1+\epsilon}$, for some very small fixed $\epsilon>0$, and $c(n)>0$. Here $\mathring{A}^{\gamma_i}$ is defined in \eqref{eq:defreg1} and $\mathrm{Var}_{\gamma_i}(A)$ is defined in \eqref{eq:VarAgammai}. Then, by a standard GFT argument (see e.g. \cite[Appendix B]{A2}), we se that
\[
\E\left[\sqrt{\frac{N}{2\mathrm{Var}_{\gamma_i}(A)}}\langle {\bm u}_i(T),\mathring{A}^{\gamma_i}{\bm u_i}(T)\rangle\right]^n=\E\left[\sqrt{\frac{N}{2\mathrm{Var}_{\gamma_i}(A)}}\langle {\bm u}_i(0),\mathring{A}^{\gamma_i}{\bm u_i}(0)\rangle\right]^n+\mathcal{O}\left( N^{-c(n)}\right).
\]
This shows that the Gaussian component added by the dynamics \eqref{eq:matDBM} can be removed at the price of a negligible error implying \eqref{eq:standgauss}.

The lower bound on the variance~\eqref{eq:Varlowerbound} is an explicit calculation relying on the 
the definition of $M$ from~\eqref{eq:MDE}. In particular, 
 we use that
 \begin{itemize}
 \item[(i)] $A$ and hence $\mathring{A}^{\gamma_i}$ are self-adjoint; 
 \item[(ii)] $\Im M(\gamma_i) \ge g$ for some 
 $g = g(\kappa, \Vert D \Vert) > 0$ since we are in the bulk;
\item[(iii)]   $\langle \mathring{A}^{\gamma_i} \Im M(\gamma_i)\rangle = 0$ by definition of the
regularisation;
\item[(iv)] $[\Re M(\gamma_i), \Im M(\gamma_i)] = 0$ from~\eqref{eq:MDE}.
\end{itemize}
 Then, after writing $\Var_{\gamma_i}(A)$ as a sum of squares and abbreviating $\Im M = \Im M(\gamma_i)$, we find
	\begin{equation*} 
\Var_{\gamma_i}(A) \ge \frac{\big\langle \big(\sqrt{\Im M} \big[ A  - \frac{\langle A (\Im M)^2\rangle}{ \langle (\Im M)^2\rangle}\big]\sqrt{\Im M} \big)^2 \big\rangle}{\langle (\Im M)^2\rangle} \ge g^2 \frac{\big\langle  \big[ A  - \frac{\langle A (\Im M)^2\rangle}{ \langle (\Im M)^2\rangle}\big]^2 \big\rangle}{\langle (\Im M)^2\rangle} \ge \frac{g^2}{\langle (\Im M)^2\rangle}  \langle (A - \langle A \rangle)^2 \rangle\,,
	\end{equation*}
	where in the last step we used the trivial variational principle $\langle (A - \langle A \rangle)^2 \rangle = \inf_{t \in \R} \langle ( A - t )^2 \rangle$. 
This completes the proof of Theorem~\ref{theo:newCLT}.
\end{proof}

\subsection{DBM analysis}

Similarly to  \cite[Section 4.1]{normalfluc} and \cite[Section 4.2]{A2} we introduce an equivalent particle representation to encode moments of the $p_{ij}$'s. In particular, here, and previously in \cite{normalfluc, A2}, we relied on the particle representation \eqref{eq:deflambda}--\eqref{xeta} below since our arguments heavily builds on \cite{MarcinekYau2005.08425}, which use this latter representation.

Consider a particle configuration ${\bm \eta}\in\Omega^n$, for some fixed $n\in\N$, i.e. ${\bm \eta}$ is such that $\sum_j\eta_j=n$. We now define the new configuration space
\begin{equation}
\label{eq:deflambda}
    \Lambda^n:= \{ {\bm x}\in [N]^{2n} \, : \,\mbox{\(n_i({\bm x})\) is even for every \(i\in [N]\)} \big\},
\end{equation}
where
\begin{equation}
    n_i({\bm x}):=|\{a\in [2n]:x_a=i\}|
\end{equation}
for all $i\in \N$. 

By the correspondence 
\begin{equation}\label{xeta}
    {\bm \eta} \leftrightarrow {\bm x}\qquad \eta_i=\frac{n_i( {\bm x})}{2}\,. 
\end{equation}
it is easy to see that these two representations are basically equivalent. The only difference is that \({\bm x}\) uniquely determines \({\bm \eta}\), but \({\bm \eta}\) determines only the coordinates of \({\bm x}\) as a multi-set and not its ordering.

From now on, given a function $f$ defined on $\Omega^n$, we will always consider functions $g$ on $\Lambda^n\subset[N]^{2n}$ defined by
\[
f({\bm \eta})= f(\phi({\bm x}))= g({\bm x}),
\]
with \(\phi\colon\Lambda^n\to \Omega^n\), \(\phi({\bm x})={\bm \eta}\)  being the projection from the \({\bm x}\)-configuration space to the \({\bm \eta}\)-configuration space using~\eqref{xeta}. We thus defined the observable
\begin{equation}\label{eq:defg}
    g_t({\bm x})=g_{{\bm \lambda},t}({\bm x}):= f_{{\bm \lambda},t}( \phi({\bm x}))\,,
\end{equation}
with \( f_{{\bm \lambda},t}\) from~\eqref{eq:deff}. Note that $g_t({\bm x})$ is equivariant under permutation of the 
arguments, i.e.~it depends on \({\bm x}\) only as a multi--set. Similarly we define
\begin{equation}
\label{eq:defr}
    r_t({\bm x})=r_{{\bm \lambda},t}({\bm x}):= q_{{\bm \lambda},t}( \phi({\bm x}))\,. 
\end{equation}
We remark that $g_t$ and $r_t$ are the counterpart of $f_t$ and $q_t$, respectively, in the $\bm{x}$-configuration space.

We can thus now write the flow~\eqref{eq:1dequa}--\eqref{eq:1dkernel} in the $\bm{x}$--configuration space:
\begin{align}\label{eq:g1deq}
    \partial_t g_t({\bm x})&=\mathcal{L}(t)g_t({\bm x}) \\\label{eq:g1dker}
    \mathcal{L}(t):=\sum_{j\ne i}\mathcal{L}_{ij}(t), \quad \mathcal{L}_{ij}(t)g({\bm x}):&= c_{ij}(t) \frac{n_j({\bm x})+1}{n_i({\bm x})-1}\sum_{a\ne b\in[2 n]}\big(g({\bm x}_{ab}^{ij})-g({\bm x})\big),
\end{align}
where
\begin{equation}
    \label{eq:jumpop}
    {\bm x}_{ab}^{ij}:={\bm x}+\delta_{x_a i}\delta_{x_b i} (j-i) ({\bm e}_a+{\bm e}_b),
\end{equation} 
with ${\bm e}_a\in \R^{2n}$ denoting the standard unit vector, i.e.~${\bm e}_a(b)=\delta_{ab}$. We remark that this flow is  map  on functions  defined on \(\Lambda^n\subset [N]^{2n}\) which preserves equivariance.

For the following analysis it is convenient to define the scalar product and the natural measure on $\Lambda^n$:
\begin{equation}\label{eq:scalpro}
    \langle f, g\rangle_{\Lambda^n}=\langle f, g\rangle _{\Lambda^n, \pi}:=\sum_{{\bm x}\in \Lambda^n}\pi({\bm x})\bar f({\bm x})g({\bm x}),
    \qquad \pi({\bm x}):=\prod_{i=1}^N ((n_i({\bm x})-1)!!)^2,
\end{equation}
as well as the norm on $L^p(\Lambda^n)$:
\begin{equation}
    \lVert f\rVert_p=\lVert f\rVert_{L^p(\Lambda^n,\pi)}:=\left(\sum_{{\bm x}\in \Lambda^n}\pi({\bm x})|f({\bm x})|^p\right)^{1/p}.
\end{equation}

The operator $\mathcal{L}=\mathcal{L}(t)$  is symmetric with respect to the measure $\pi$
and it is a negative in $L^2(\Lambda^n)$, with associated Dirichlet form (see~\cite[Appendix A.2]{MarcinekThesis}):
\[
D(g)=\langle g, (-\mathcal{L}) g\rangle_{\Lambda^n} = \frac{1}{2}  \sum_{{\bm x}\in \Lambda^n}\pi({\bm x})
\sum_{i\ne j} c_{ij}(t) \frac{n_j({\bm x})+1}{n_i({\bm x})-1}
\sum_{a\ne b\in[2 n]}\big|g({\bm x}_{ab}^{ij})-g({\bm x})\big|^2.
\]
Finally, by $\mathcal{U}(s,t)$ we denote the semigroup associated to $\mathcal{L}$, i.e.\ for any $0\le s\le t$ it holds
\begin{equation}
\label{eq:defsemigroup}
\partial_t\mathcal{U}(s,t)=\mathcal{L}(t)\mathcal{U}(s,t), \quad \mathcal{U}(s,s)=I.
\end{equation}

\subsection{Short range approximation}

As a consequence of the singularity of the coefficients $c_{ij}(t)$ in \eqref{eq:g1dker}, the main contribution to the flow \eqref{eq:g1deq} comes from nearby eigenvalues, hence its analysis will be completely local. For this purpose we define the sets
\begin{equation}
\label{eq:defintJ}
    \mathcal{J}=\mathcal{J}_\kappa:=\{ i\in [N]:\, \gamma_i(0)\in \mathbf{B}_{\kappa}\}, 
\end{equation}
which correspond to indices with quantiles $\gamma_i(0)$ (recall \eqref{eq:quant}) in the bulk. 

Fix a point ${\bm y}\in \mathcal{J}^{2n}$, and an $N$-dependent parameter $K$ such that $1\ll K\ll \sqrt{N}$. We remark that ${\bm y}\in \mathcal{J}^{2n}$ will be fixed for the rest of the analysis. Next, we define the  \emph{averaging operator} as a simple multiplication operator by a ``smooth'' cut-off function:
\begin{equation}
    \Av(K,{\bm y})h({\bm x}):=\Av({\bm x};K,{\bm y})h({\bm x}), \qquad \Av({\bm x}; K, {\bm y}):=\frac{1}{K}\sum_{j=K}^{2K-1} \bm1(\lVert{\bm x}-{\bm y}\rVert_1<j),
\end{equation}
with \(\lVert{\bm x}-{\bm y}\rVert_1:=\sum_{a=1}^{2n} |x_a-y_a|\). For notational simplicity we may often omit $K,\bm{y}$ from the notation since they are fixed throughout the proof:
\begin{equation}
\label{eq:xnotation}
\Av(\bm{x})=\Av({\bm x};K,{\bm y})h({\bm x}), \qquad \Av h(\bm{x})=\Av((\bm{x}))h((\bm{x})).
\end{equation}
Additionally, fix an  integer \(\ell\) with \(1\ll\ell\ll K\),
and define the short range coefficients
\begin{equation}\label{eq:ccutoff}
    c_{ij}^{\mathcal{S}}(t):=\begin{cases}
        c_{ij}(t) &\mathrm{if}\,\, i,j\in \mathcal{J} \,\, \mathrm{and}\,\, |i-j|\le \ell \\
        0 & \mathrm{otherwise},
    \end{cases}
\end{equation}
where \(c_{ij}(t)\) is defined in~\eqref{eq:defc}. The parameter $\ell$ is the length of the short range interaction.

We now define a short--range approximation of $r_t$, with $r_t$ defined in \eqref{eq:defr}. Note that in the definition of the short--range flow \eqref{g-1} below there is a slight notational difference compared to \cite[Section 4.2]{normalfluc} and \cite[Section 4.2.1]{A2}: we now choose an initial condition $h_0$ which depends on $r_0$ rather than $g_0$. This minor difference is caused by the fact that in \cite{normalfluc, A2} the observable $g_t$ was already centred and rescaled, while in the current case the centred and rescaled version of $g_t$ is given by $r_t$, hence the definition in \eqref{g-1} is still conceptually the same as the one in \cite{normalfluc, A2}  (see also the paragraph above \eqref{eq:fundrelicon} for a more detailed explanation). We point out that we make this choice to ensure that the infinite norm of the short range approximation is always bounded by $N^\xi$ (see below \eqref{g-2}). The short range approximation $h_t=h_t({\bm x})$  is defined as the unique solution of the parabolic equation
\begin{equation}
\label{g-1}
    \begin{split}
        \partial_t h_t({\bm x}; \ell, K,{\bm y})&=\mathcal{S}(t) h_t({\bm x}; \ell, K,{\bm y})\\
        h_0({\bm x};\ell, K,{\bm y})=h_0({\bm x};K,{\bm y}):&=\Av({\bm x}; K,{\bm y})(r_0({\bm x})-\bm1(n \,\, \mathrm{even})), 
    \end{split}
\end{equation}
where
\begin{equation}\label{g-2}
    \mathcal{S}(t):=\sum_{j\ne i}\mathcal{S}_{ij}(t), \quad \mathcal{S}_{ij}(t)h({\bm x}):=c_{ij}^{\mathcal{S}}(t)\frac{n_j({\bm x})+1}{n_i({\bm x})-1}\sum_{a\ne b\in [2n]}\big(h({\bm x}_{ab}^{ij})-h({\bm x})\big).
\end{equation}
In the remainder of this section we may often omit $K$, ${\bm y}$ and $\ell$ from the notation, since they are fixed for the rest of the proof. We conclude this section defining the transition semigroup $\mathcal{U}_{\mathcal{S}}(s,t)=\mathcal{U}_{\mathcal{S}}(s,t;\ell)$ associated to the short range generator $\mathcal{S}(t)$. Note that $\lVert h_t\rVert_\infty\le N^\xi$, for any $t\ge 0$ and any small $\xi>0$, since $\mathcal{U}_{\mathcal{S}}(s,t)$ is a contraction and $\lVert h_0\rVert_\infty\le N^\xi$ by \eqref{eq:ETH}, as a consequence of $h_t(\bm{x})$ being supported on $\bm{x}\in\mathcal{J}^{2n}$.

\subsection{$L^2$--estimates}

To prove the $L^\infty$--bound in Proposition~\ref{pro:flucque}, we first prove an $L^2$--bound in Proposition~\ref{prop:mainimprov} below and then use an ultracontractivity argument for the parabolic PDE \eqref{eq:g1deq} (see \cite[Section 4.4]{normalfluc}) to get an $L^\infty$--bound. To get an $L^2$--bound we will analyse $h_t$, the short--range version of the observable $g_t$ from \eqref{eq:defg}, and then we will show that $h_t$ and $g_t$ are actually close to each other using the following finite speed of propagation (see \cite[Proposition 4.2, Lemmas 4.3--4.4]{normalfluc}):
\begin{lemma}
\label{lem:shortlongapprox}
    Let \(0\le s_1\le s_2\le s_1+\ell N^{-1}\), and \(f\) be a function on \(\Lambda^n\), 
    then for any \({\bm x}\in \Lambda^n\) supported on \(\mathcal{J}\) it holds
    \begin{equation}\label{eq:shortlong}
        \Big| (\mathcal{U}(s_1,s_2)-\mathcal{U}_{\mathcal{S}}(s_1,s_2;\ell) ) f({\bm x}) \Big|\lesssim N^{1+n\xi}\frac{s_2-s_1}{\ell} \| f\|_\infty\,,
    \end{equation}
    for any small \(\xi>0\). The implicit constant in~\eqref{eq:mainbthissec} depends on \(n\), \(\epsilon\), \(\delta\).
\end{lemma}

To estimate several terms in the analysis of \eqref{g-1} we will rely on the multi--resolvent local laws from Proposition~\ref{pro:mresllaw} (in combination with the extensions in Lemma~\ref{lem:addllaw} in Lemma \ref{lem:2Gllfaraway}). For this purpose,  for a small \(\omega>2\xi>0\), we define the very high probability event (see Lemmas~\ref{lem:addllaw}--\ref{lem:2Gllfaraway})
\begin{equation}  
    \label{eq:hatomega}
    \begin{split}
    \small
           & \widehat{\Omega}=\widehat{\Omega}_{\omega, \xi}:= \\
       &\bigcap_{\substack{e_i\in\mathbf{B}_\kappa, \atop |\Im z_i|\ge N^{-1+\omega}}}\Bigg[\bigcap_{k=2}^n \left\{\sup_{0\le t \le T}\left|\langle G_t(z_1)\mathring{A}_1\dots G_t(z_k)\mathring{A}_k\rangle-\bm1(k=2)\langle M(z_1, \mathring{A}_1, z_2) \mathring{A}_2\rangle\right|\le \frac{N^{\xi+k/2-1}}{\sqrt{N\eta}}\right\} \\
        &\qquad\quad\cap \left\{\sup_{0\le t \le T}\big|\langle G_t(z_1)\mathring{A}_1\rangle\big|\le \frac{N^\xi}{N\sqrt{|\Im z_1|}}\right\}\Bigg] \\
        &\bigcap_{z_1,z_2: e_1\in\mathbf{B}_\kappa,\atop |e_1-e_2|\ge c_1, |\Im z_i|\ge N^{-1+\omega}} \left\{\sup_{0\le t \le T}\big|\langle (G_t(z_1)B_1G_t(z_2)B_2\rangle\big|\le N^\xi\right\}\,,
    \end{split}
\end{equation}
where $\mathring{A}_1,\dots,\mathring{A}_k$ are regular matrices defined as in Definition~\ref{def:regobs} (here we used the short--hand notation $\mathring{A}_i=\mathring{A}_i^{z_i,z_{i+1}}$), 
\begin{equation}
\label{eq:defm12}
\langle M(z_1, \mathring{A}_1, z_2) \mathring{A}_2\rangle=\langle  M(z_1)\mathring{A}_1 M(z_2)\mathring{A}_2\rangle +\frac{\langle M(z_1)\mathring{A}_1M(z_2)\rangle\langle M(z_2)\mathring{A}_2M(z_1)\rangle}{1-\langle M(z_1)M(z_2)\rangle}\,,
\end{equation}
\(\eta:=\min\{|\Im z_i| : i\in[k]\}\), $c_1>0$ is a fixed small constant, and $B_1,B_2$ are norm bounded deterministic matrices. We remark that for $|e_1-e_2|\ge c_1$ we have the norm bound $\lVert M(z_1,B_1, z_2)\rVert\lesssim 1$, with $M(z_1,B_1, z_2)$ being defined in \eqref{eq:Mdef}. Then, by standard arguments (see e.g. \cite[Eq. (4.30)]{A2}), we conclude the bound (recall that $\widetilde{\Omega}_\xi$ from \eqref{def:Omega} denotes the rigidity event)
\begin{equation}
\label{eq:apriori}
    \max_{i,j\in \mathcal{J}}|\langle{\bm u}_i(t), A {\bm u}_j(t)\rangle |\le \frac{N^{\omega}}{\sqrt{N}} \qquad \,\,
    \mbox{on \(\;\widehat{\Omega}_{\omega,\xi}\)}\cap\widetilde{\Omega}_\xi\,,
\end{equation}
simultaneously for all $i,j\in \mathcal{J}$ and $0\le t\le T$. Additionally, using the notation $\rho_{i,t}:=|\Im \langle M_t(z_i)\rangle|$, on $\widehat{\Omega}_{\omega,\xi}\cap\widetilde{\Omega}_\xi$ it also holds that 
\begin{equation}
\label{eq:overweak}
|\langle{\bm u}_i(t), A {\bm u}_j(t)\rangle|\le N^\omega \sqrt{\frac{\langle \Im G(\gamma_i(t)+\ii\eta)A\Im G(\gamma_j(t)+\ii\eta)A  \rangle}{N\rho_{i,t}\rho_{j,t}}}\lesssim \frac{N^\omega}{N^{1/4}}\,,
\end{equation}
when one among $i$ and $j$ is in the bulk and $|i-j|\ge c N$, for some small constant  $c$ depending on $c_1$ from \eqref{eq:hatomega}. Here we used that for the index in the bulk, say $i$, we have $\rho_{i,t}\sim 1$ and for the other index $\rho_{j,t}\gtrsim N^{-1/2}$ as a consequence of $\eta\gg N^{-1}$. We point out that this non optimal bound $N^{-1/4}$, instead of the optimal $N^{-1/2}$, follows from the fact that the bound from Lemma~\ref{lem:2Gllfaraway} is not optimal when one of the two spectral parameters in close to an edge; this is exactly the same situation as in \cite[Eq. (4.31)]{A2} where we get an analogous non optimal bound for overlaps of eigenvectors that are not in the bulk.

We are now ready to prove the main technical proposition of this section. Note the additional term $KN^{-1/2}$ in the error 
$\mathcal{E}$ in \eqref{eq:basimpr} compared to \cite[Proposition 4.2]{normalfluc} and \cite[Proposition 4.4]{A2}; this is a consequence of the fact the $p_{ii}$'s in \eqref{eq:defpij} are not correctly centred. We stress that the base point $\bm{y}$ in Proposition~\ref{prop:mainimprov} is fixed throughout the remainder of this section.

\begin{proposition}\label{prop:mainimprov}
    For any parameters satisfying \(N^{-1}\ll \eta\ll T_1\ll \ell N^{-1}\ll K N^{-1}\ll N^{-1/2}\), and any small \(\epsilon, \xi>0\) it holds
    \begin{equation}\label{eq:l2b}
        \lVert h_{T_1}(\cdot; \ell, K, {\bm y})\rVert_2\lesssim K^{n/2}\mathcal{E},
    \end{equation}
    with 
    \begin{equation}\label{eq:basimpr}
        \mathcal{E}:= N^{n\xi}\left(\frac{N^\epsilon\ell}{K}+\frac{NT_1}{\ell}+\frac{N\eta}{\ell}+\frac{N^\epsilon}{\sqrt{N\eta}}+\frac{1}{\sqrt{K}}+\frac{K}{\sqrt{N}}\right),
    \end{equation}
    uniformly in particle configurations $\bm{y}$ such that $y_a=i_0$, for any $a\in [2n]$ and $i_0\in \mathcal{J}$, and eigenvalue trajectory \({\bm \lambda}\) 
    in the high probability event \(\widetilde{\Omega}_\xi \cap \widehat{\Omega}_{\omega,\xi}\).
\end{proposition}
\begin{proof}

The proof of this proposition is very similar to the one of \cite[Proposition 4.2]{normalfluc} and \cite[Proposition~4.4]{A2}, we thus only explain the main differences here. In the following by the star over $\sum$ we denote that the summation runs over two $n$-tuples of fully distinct indices. The key idea in this proof is that in order to rely on the multi--resolvent local laws \eqref{eq:hatomega} we replace the operator $\mathcal{S}(t)$ in \eqref{g-1} with the new operator
\begin{equation}\label{eq:defAgen}
        \mathcal{A}(t):=\sum_{{\bm i}, {\bm j}\in [N]^n}^*\mathcal{A}_{ {\bm i}{\bm j}}(t), \quad \mathcal{A}_{{\bm i}{\bm j}}(t)h({\bm x}):=\frac{1}{\eta}\left(\prod_{r=1}^n a_{i_r,j_r}^\mathcal{S}(t)\right)\sum_{{\bm a}, {\bm b}\in [2n]^n}^*(h({\bm x}_{{\bm a}{\bm b}}^{{\bm i}{\bm j}})-h({\bm x})),
    \end{equation}
    where
    \begin{equation}
        a_{ij}=a_{ij}(t):=\frac{\eta}{N((\lambda_i(t)-\lambda_j(t))^2+\eta^2)},
    \end{equation}
    and \(a_{ij}^\mathcal{S}\) are their short range version defined as in~\eqref{eq:ccutoff} with $c_{ij}(t)$ replaced with $a_{ij}(t)$, and
    \begin{equation}
        \label{eq:lotsofjumpsop}
        {\bm x}_{{\bm a}{\bm b}}^{{\bm i}{\bm j}}:={\bm x}+\left(\prod_{r=1}^n \delta_{x_{a_r}i_r}\delta_{x_{b_r}i_r}\right)\sum_{r=1}^n (j_r-i_r) ({\bm e}_{a_r}+{\bm e}_{b_r}).
    \end{equation}
  The main idea behind this replacement is that infinitesimally $\mathcal{S}(t)$ averages only in one direction at a time, whilst $\mathcal{A}(t)$ averages in all direction simultaneously. This is expressed by the fact that  ${\bm x}^{ij}_{ab}$ from~\eqref{eq:jumpop} changes two entries of ${\bm x}$ per time, instead  ${\bm x}_{{\bm a}{\bm b}}^{{\bm i}{\bm j}}$ 
    changes all the coordinates of ${\bm x}$ at the same time, i.e. let ${\bm i}:=(i_1,\dots, i_n), {\bm j}:=(j_1,\dots, j_n)\in [N]^n$, with $\{i_1,\dots,i_n\}\cap\{j_1,\dots, j_n\}=\emptyset$, then ${\bm x}_{{\bm a}{\bm b}}^{{\bm i}{\bm j}}\ne {\bm x}$ if and only if for all $r\in [n]$ it holds that $x_{a_r}=x_{b_r}=i_r$.  Technically, the replacement of $\mathcal{S}(t)$ by $\mathcal{A}(t)$ can be performed at the level of Dirichlet forms.
    \begin{lemma}[Lemma 4.6 of \cite{normalfluc}]
        \label{lem:replacement}
        Let \(\mathcal{S}(t)\), \(\mathcal{A}(t)\) be defined in~\eqref{g-2} and~\eqref{eq:defAgen}, respectively,
        and let $\mu$ denote the uniform measure on $\Lambda^n$ for which  \(\mathcal{A}(t)\) is reversible.
        Then there exists a constant \(C(n)>0\) such that
        \begin{equation}\label{eq:fundbound}
            \langle h, \mathcal{S}(t) h\rangle_{\Lambda^n, \pi}\le C(n) \langle h, \mathcal{A}(t) h\rangle_{\Lambda^n,\mu}\le 0,
        \end{equation}
        for any \(h\in L^2(\Lambda^n)\), on the very high probability event \(\widetilde{\Omega}_\xi \cap \widehat{\Omega}_{\omega,\xi}\).
    \end{lemma}

We start noticing the fact that by \eqref{g-1} it follows
    \begin{equation}\label{eq:l2der}
        \partial_t \lVert h_t\rVert_2^2=2\langle h_t, \mathcal{S}(t) h_t\rangle_{\Lambda^n}.
    \end{equation}
    Then, combining this with \eqref{eq:fundbound}, and using that \({\bm x}_{{\bm a}{\bm b}}^{{\bm i}{\bm j}}={\bm x}\) unless \({\bm x}_{a_r}={\bm x}_{b_r}=i_r\) for all \(r\in [n]\), we conclude that
    \begin{equation}\label{eq:boundderl2}
        \begin{split}
            \partial_t \lVert h_t\rVert_2^2&\le C(n) \langle h_t,\mathcal{A}(t) h_t\rangle_{\Lambda^n,\mu} \\
            &=\frac{C(n) }{2\eta}\sum_{{\bm x}\in \Lambda^n}\sum_{{\bm i},{\bm j}\in [N]^n}^* \left(\prod_{r=1}^n a_{i_r j_r}^\mathcal{S}(t)\right)\sum_{{\bm a}, {\bm b}\in [2n]^n}^*\overline{h_t}({\bm x})\big(h_t({\bm x}_{{\bm a}{\bm b}}^{{\bm i} {\bm j}})-h_t({\bm x})\big)\left(\prod_{r=1}^n \delta_{x_{a_r}i_r}\delta_{x_{b_r}i_r}\right).
        \end{split}
    \end{equation}
       Then, proceeding as in the proof of \cite[Proposition 4.5]{normalfluc} (see also \cite[Eq. (4.40)]{A2}), we conclude that
    \begin{equation}
        \label{eq:finb}
        \partial_t\lVert h_t\rVert_2^2\le -\frac{C_1(n)}{2\eta}\langle h_t\rangle_2^2+\frac{C_3(n)}{\eta}\mathcal{E}^2K^n,
    \end{equation}
    which implies \(\lVert h_{T_1}\rVert_2^2\le C(n) \mathcal{E}^2 K^n\), by a simple Gronwall inequality, using that \(T_1\gg \eta\).
    
    We point out that to go from \eqref{eq:boundderl2} to \eqref{eq:finb} the proof is completely analogous to \cite[Proof of Proposition 4.5]{normalfluc}, with the only exception being the proof of  \cite[Eqs. (4.41), (4.43)]{normalfluc}. We thus now explain how to obtain the analog of \cite[Eqs. (4.41), (4.43)]{normalfluc} in the current case as well. The fact that we now have the bound \eqref{eq:overweak} rather than the stronger bound $N^{-1/3}$ as in \cite[Eq. (4.31)]{A2} does not cause any difference in the final estimate. We thus focus on the main new difficulty in the current analysis, i.e. that in \eqref{g-1} we choose the initial condition depending on $r_0$ rather than $g_0$. We recall that the difference between $r_0$ and $g_0$ is that $r_0$ is defined in such a way all the eigenvector overlaps are precisely centred and normalised in an $i$--dependent way, whilst for $g_0$ we can choose the $i$--independent constant $C_0,C_1$ so that only the overlap corresponding to a certain base point $i_0$ is exactly centred and normalised, whilst the nearby overlaps are centred and normalised only modulo a  negligible error $K/N$ (see also the paragraph below \eqref{eq:qg} for a detailed explanation). This additional difficulty requires that to prove the analog of \cite[Eqs. (4.41), (4.43)]{normalfluc} we need to estimate the error produced by this mismatch.

Using that the function \(f({\bm x})\equiv \bm1(n \,\, \mathrm{even})\) is in the kernel of \(\mathcal{L}(t)\), for any fixed \({\bm x}\in\Gamma\), and for any fixed \({\bm i}\), \({\bm a}\), \({\bm b}\), we conclude (recall the notation from \eqref{eq:xnotation})
\begin{equation}\label{eq:fundrelicon}
\begin{split}
&h_t({\bm x}_{{\bm a}{\bm b}}^{{\bm i} {\bm j}})\\
&=\mathcal{U}_\mathcal{S}(0,t)\big((\Av r_0)({\bm x}_{{\bm a}{\bm b}}^{{\bm i}{\bm j}})-(\Av\bm1(n \,\, \mathrm{even}))({\bm x}_{{\bm a}{\bm b}}^{{\bm i}{\bm j}})\big) \\
&=\Av({\bm x}_{{\bm a}{\bm b}}^{{\bm i}{\bm j}})\big(\mathcal{U}_\mathcal{S}(0,t)r_0({\bm x}_{{\bm a}{\bm b}}^{{\bm i}{\bm j}})-\bm1(n\,\, \mathrm{even})\big)+\mathcal{O}\left(\frac{N^{\epsilon+n\xi} \ell}{K}\right) \\
&=\left(\Av({\bm x})+\mathcal{O}\left(\frac{\ell}{K}\right)\right)\left(\mathcal{U}(0,t)r_0({\bm x}_{{\bm a}{\bm b}}^{{\bm i}{\bm j}})-\bm1(n\,\, \mathrm{even})+\mathcal{O}\left(\frac{N^{1+n\xi}t}{\ell}\right)\right)+\mathcal{O}\left(\frac{N^{\epsilon+n\xi} \ell}{K}\right) \\
&=\Av({\bm x})\big(g_t({\bm x}_{{\bm a}{\bm b}}^{{\bm i}{\bm j}})-\bm1(n\,\, \mathrm{even})\big)+\mathcal{O}\left(\frac{N^{\epsilon+n\xi} \ell}{K}+\frac{N^{1+n\xi}t}{\ell}+\frac{N^\xi K}{\sqrt{N}}\right),
\end{split}
\end{equation}
where in the definition of $g_t$ from \eqref{eq:deff}, \eqref{eq:defg} we chose
\begin{equation}
\label{eq:choicec0c1}
C_0:=\frac{\langle A\Im M(\gamma_{i_0})\rangle }{\langle \Im M(\gamma_{i_0})\rangle}, \qquad C_1:=\mathrm{Var}_{\gamma_{i_0}}(A),
\end{equation}
and the error terms are uniform in \({\bm x}\in\Gamma\). Here $i_0$ is the index defined below \eqref{eq:basimpr}. The first three inequalities are completely analogous to \cite[Eq. (4.41)]{normalfluc}. We now explain how to obtain the last inequality at the price of the additional negligible error $KN^{-1/2}$. Recall the definition of $r_t$ from \eqref{eq:defq}, \eqref{eq:defr}, then we now show that for any $\bm{x}$ supported in the bulk it holds
\begin{equation}
\label{eq:rgclose}
\lVert r_0(\bm{x})-g_0(\bm{x})\rVert_\infty\lesssim \frac{N^\xi K}{\sqrt{N}},
\end{equation}
for $C_0,C_1$ chosen as in \eqref{eq:choicec0c1}. Using \eqref{eq:rgclose}, together with $\mathcal{U}(0,t)g_0=g_t$, this proves the last equality in \eqref{eq:fundrelicon}. The main input in the proof of  \eqref{eq:rgclose} is the following approximation result
\begin{equation}
\label{eq:appr}
\mathring{A}^{\gamma_{i_0}}-\mathring{A}^{\gamma_{i_r}}=\mathcal{O}(|\gamma_{i_0}-\gamma_{i_r}|)I=\mathcal{O}(KN^{-1})I.
\end{equation}
We now explain the proof of \eqref{eq:rgclose}; for simplicity we present the proof only in the case $n=2$.  To prove \eqref{eq:rgclose} we see that
\[
2|\langle {\bm u}_i,\mathring{A}^{\gamma_{i_0}}  {\bm u}_j\rangle|^2+\langle {\bm u}_i,\mathring{A}^{\gamma_{i_0}}  {\bm u}_i\rangle\langle {\bm u}_j,\mathring{A}^{\gamma_{i_0}}  {\bm u}_j\rangle=2|\langle {\bm u}_i,\mathring{A}^{\gamma_i}  {\bm u}_j\rangle|^2+\langle {\bm u}_i,\mathring{A}^{\gamma_i}  {\bm u}_i\rangle\langle {\bm u}_j,\mathring{A}^{\gamma_j}  {\bm u}_j\rangle+\mathcal{O}\left(\frac{1}{N}\cdot\frac{N^\xi K}{\sqrt{N}}\right),
\]
where we used \eqref{eq:appr} to replace the "wrong" $\mathring{A}^{\gamma_{i_0}}$ with the "correct" $\mathring{A}^{\gamma_i}$ together with the a priori a bound $\langle {\bm u}_i,\mathring{A}^{\gamma_i}  {\bm u}_j \rangle \le N^\xi N^{-1/2}$. Then multiplying this relation by $N$ we obtain \eqref{eq:rgclose}. Additionally, since $r_0$ and $g_0$ contain a different rescaling in terms of $\mathrm{Var}_{\gamma_{i_0}}(A)$ and $\mathrm{Var}_{\gamma_i}(A)$, we also used that by similar computations
\[
\mathrm{Var}_{\gamma_{i_0}}(A)=\mathrm{Var}_{\gamma_i}(A)+\mathcal{O}\left(\frac{N^\xi K}{\sqrt{N}}\right).
\]
In particular, we used this approximation to compensate the mismatch that only the diagonal overlaps corresponding to the index $i_0$ are properly centred and normalised in the definition of $g_0$, whilst for nearby indices we use this approximation to replace the approximate centering $C_0$ and normalisation $C_1$ from \eqref{eq:choicec0c1} with the correct one, which is the one in the definition of $r_0$.

Then proceeding as in the proof of \cite[Eq. (4.41)]{A2} we conclude the analog of \cite[Eq. (4.43)]{normalfluc}:
\begin{equation}
        \begin{split}
            \label{eq:3}
            &\sum_{{\bm j}}^*\left(\prod_{r=1}^n a_{i_r j_r}^\mathcal{S}(t)\right)\big(g_t({\bm x}_{{\bm a}{\bm b}}^{{\bm i} {\bm j}})-\bm1(n\,\, \mathrm{even})\big) \\
            &= \sum_{{\bm j}}\left(\prod_{r=1}^n a_{i_r j_r}(t)\right)\left(\frac{N^{n/2}}{\mathrm{Var}_{\gamma_{i_0}}(A)^{n/2} 2^{n/2}(n-1)!!}\sum_{G\in \mathcal{G}_{{\bm \eta}^{{\bm j}}}}P(G)-\bm1(n\,\, \mathrm{even})\right) \\
            &\quad+\mathcal{O}\left(\frac{N^{n\xi}}{N\eta}+\frac{N^{1+n\xi}\eta}{\ell}+\frac{N^\xi K}{\sqrt{N}}\right).
        \end{split}
    \end{equation}
    Given \eqref{eq:3}, the remaining part of the proof is completely analogous to \cite[Eqs. (4.44)--(4.51)]{normalfluc} except for the slightly different computation
    \begin{equation}
\begin{split}
\label{eq:longcomp}
&\frac{\langle \Im G(\lambda_{i_{r_1}}+\ii \eta)\mathring{A}^{\gamma_{i_0}}\Im G (\lambda_{i_{r_2}}+\ii \eta)\mathring{A}^{\gamma_{i_0}}\rangle}{\mathrm{Var}_{\gamma_{i_0}}(A)} \\
&\qquad\qquad=-\frac{1}{4\mathrm{Var}_{\gamma_{i_0}}(A)}\sum_{\sigma,\tau\in \{+,-\}}\langle  G(\lambda_{i_{r_1}}+\sigma\ii \eta)\mathring{A}^{\gamma_{i_0}} G (\lambda_{i_{r_2}}+\tau\ii \eta)\mathring{A}^{\gamma_{i_0}}\rangle\\
&\qquad \qquad= -\frac{1}{4\mathrm{Var}_{\gamma_{i_0}}(A)}\sum_{\sigma,\tau\in \{+,-\}}\left\langle G(\lambda_{i_{r_1}}+\ii \sigma \eta)\mathring{A}^{\gamma_{i_{r_1}}+\ii \sigma\eta,\gamma_{i_{r_2}}+\ii\tau\eta} G (\lambda_{i_{r_2}}+\ii \tau \eta)\mathring{A}^{\gamma_{i_{r_2}}+\ii \tau\eta, \gamma_{i_{r_1}}+\ii \sigma\eta}\right\rangle \\
&\qquad\qquad\quad+\mathcal{O}\left(\frac{K}{N^2\eta^{3/2}}+\frac{K^2}{N^2\eta} +\frac{1}{N\sqrt{\eta}}\nc\right)\\
&\qquad\qquad=\langle \Im M(\gamma_{i_0})\rangle^2+\mathcal{O}\left(\frac{K}{N^2\eta^{3/2}}+\frac{K^2}{N^2\eta}+\frac{1}{\sqrt{N\eta}}\right),
\end{split}
\end{equation}
which replaces  \cite[Eqs. (4.47)]{normalfluc}. Here $\mathring{A}^{\gamma_{i_{r_1}}\pm \ii \eta,\gamma_{i_{r_2}}\pm \ii\eta}$ is defined as in Definition~\ref{def:regobs}. We also point out that in the second equality we used the approximation (see \eqref{eq:lipprop})
\begin{equation}
\label{eq:appr2reg}
\mathring{A}^{\gamma_{i_0}}=\mathring{A}^{\gamma_{i_{r_1}}\pm\ii\eta,\gamma_{i_{r_2}}\pm \ii \eta}
+\mathcal{O}\big(|\gamma_{i_{r_1}}-\gamma_{i_0}|
+|\gamma_{i_{r_2}}-\gamma_{i_0}|+\eta\big)I
=\mathring{A}^{\gamma_{i_{r_1}},\gamma_{i_{r_2}}}+\mathcal{O}\left(\frac{K}{N}+\eta\right)I,
\end{equation}
together with (here we present the estimate only for one representative term, the other being analogous)
\begin{equation}
\begin{split}
\langle  G(\lambda_{i_{r_1}}+\ii \eta) G (\lambda_{i_{r_2}}+\ii \eta)\mathring{A}^{\gamma_{i_0}}\rangle&=\langle  G(\lambda_{i_{r_1}}+\ii \eta)G (\lambda_{i_{r_2}}+\ii \eta)\mathring{A}^{\gamma_{i_{r_2}}+\ii \eta,\gamma_{i_{r_2}}+\ii\eta}\rangle+\mathcal{O}\left(1+\frac{K}{N\eta}\right) \\
&=\langle M(\gamma_{i_{r_2}}+\ii  \eta, \mathring{A}^{\gamma_{i_{r_2}}+\ii \eta, \gamma_{i_{r_1}}+\ii\eta}, \gamma_{i_{r_1}}+\ii \eta) \rangle+\mathcal{O}\left( 1+\frac{1}{N\eta^{3/2}}+\frac{K}{N\eta}\right)\\
& =\mathcal{O}\left(1+\frac{1}{N\eta^{3/2}}+\frac{K}{N\eta}\right),
\end{split}
\end{equation}
which follows by \eqref{eq:hatomega} and Lemma~\ref{lem:Mbound} to estimate the deterministic term, together
 with the integral representation from \cite[Lemma~5.1]{iid} (see also \eqref{eq:integr identity} later),
  to bound the error terms arising from the replacement \eqref{eq:appr2reg}. Additionally, in the third equality we used the local for two resolvents from \eqref{eq:hatomega}:
\begin{equation}
\begin{split}
\label{eq:explain}
&-\frac{1}{4}\sum_{\sigma,\tau\in \{+,-\}}\left\langle G(\lambda_{i_{r_1}}+\ii \sigma \eta)\mathring{A}^{\gamma_{i_{r_1}}+\ii \sigma\eta,\gamma_{i_{r_2}}+\ii\tau\eta} G (\lambda_{i_{r_2}}+\ii \tau \eta)\mathring{A}^{\gamma_{i_{r_2}}+\ii \tau\eta, \gamma_{i_{r_1}}+\ii \sigma\eta}\right\rangle \\
&\qquad\quad =-\frac{1}{4}\sum_{\sigma,\tau\in \{+,-\}}\left\langle M(\gamma_{i_{r_1}}+\ii \sigma \eta, \mathring{A}^{\gamma_{i_{r_1}}+\ii \sigma\eta,\gamma_{i_{r_2}}+\ii\tau\eta}, \gamma_{i_{r_2}}+\ii \tau\eta)\mathring{A}^{\gamma_{i_{r_2}}+\ii \tau\eta, \gamma_{i_{r_1}}+\ii \sigma\eta}\right\rangle+\mathcal{O}\left(\frac{1}{\sqrt{N\eta}}\right) \\
&\qquad\quad = \langle \Im M(\gamma_{i_0})\rangle^2\mathrm{Var}_{\gamma_{i_0}}(A)+\mathcal{O}\left(\frac{1}{\sqrt{N\eta}}+\frac{K}{N}\right),
\end{split}
\end{equation}
with the deterministic term defined in \eqref{eq:defm12}, and in the second equality we used the approximation \eqref{eq:appr2reg} again. We remark that here we presented this slightly different (compared to \cite{normalfluc, A2}) computation only for chain of length $k=2$, the computation for longer chains is completely analogous and so omitted. This concludes the proof of \eqref{eq:l2b}.

\end{proof}

We conclude this section with the proof of Proposition~\ref{pro:flucque}.

\begin{proof}[Proof of Proposition~\ref{pro:flucque}]

Combining the $L^2$--bound on $h_t$ from Proposition~\ref{prop:mainimprov}  and the finite speed of propagation estimates in Lemma~\ref{lem:shortlongapprox} we can enhance this $L^2$--bound to an $L^\infty$--bound completely analogously to the proof of \cite[Proposition 3.2 ]{normalfluc} presented in \cite[Section 4.4]{normalfluc}.

\end{proof}

\section{Proof of Proposition~\ref{pro:mresllaw}}
\label{sec:llaw}
Our strategy for proving Proposition \ref{pro:mresllaw} (in the much more involved $\eta \le 1$ regime)
is to derive a system of {\it master inequalities}  (Proposition~\ref{prop:master}) 
for the errors in the local laws by cumulant expansion, then use an iterative scheme
to gradually improve their estimates. The cumulant expansion naturally introduces longer resolvent chains, potentially leading
to an uncontrollable hierarchy, so our
master inequalities are complemented by a set of {\it reduction inequalities} (Lemma~\ref{lem:reduction})
to estimate longer chain in terms of shorter ones.
We have used a similar strategy  in~\cite{multiG, A2} for Wigner matrices, but now, analogously to \cite{iid}, dealing with non-Hermitian i.i.d.~matrices, many new error terms due to several adjustments of the $z$-dependent two-point regularisations need to be handled. By the strong analogy to \cite{iid}, our proof of the master inequalities formulated in Proposition \ref{prop:master} and given in Section \ref{subsec:proofmaster} will be rather short and focus on the main differences between \cite{iid} and the current setup.

 As the basic control quantities, analogously to \cite{multiG, iid}, in the sequel of the proof, we  introduce the normalised differences
\begin{align} \label{eq:Psi avk}
	\Psi_k^{\rm av}(\boldsymbol{z}_k, \boldsymbol{A}_k) &:= N \eta^{k/2} |\langle  G_1 A_1 \cdots G_{k} A_k  -  M(z_1, A_1, ... , z_{k}) A_k  \rangle |\,, \\
	\label{eq:Psi isok}
	\Psi_k^{\rm iso}(\boldsymbol{z}_{k+1}, \boldsymbol{A}_{k}, \boldsymbol{x}, \boldsymbol{y}) &:= \sqrt{N \eta^{k+1}} \left\vert  \big(  G_1 A_1 \cdots A_{k} G_{k+1} - M(z_1, A_1, ... , A_{k}, z_{k+1}) \big)_{\boldsymbol{x} \boldsymbol{y}} \right\vert
\end{align}
for $k \in \N$, where we used the short hand notations
\begin{equation*} 
	G_i:= G(z_i)\,, \quad \eta := \min_i |\Im z_i|\,, \quad \boldsymbol{z}_k :=(z_1, ... , z_{k})\,, \quad \boldsymbol{A}_k:=(A_1, ... , A_{k})\,.
\end{equation*}
The deterministic matrices $\Vert A_i \Vert \le 1$, $i \in [k]$, are assumed to be \emph{regular} (i.e., $A_i  = \mathring{A}^{z_i, z_{i+1}}$,
 see Definition~\ref{def:regobs}) and the deterministic counterparts used in \eqref{eq:Psi avk} and \eqref{eq:Psi isok} are given recursively in Definition \ref{def:Mdef}. For convenience, we extend the above definitions to $k=0$ by
\begin{equation*} 
	\Psi_0^{\rm av}(z):= N \eta |\langle G(z) - M(z)\rangle |\,, \quad \Psi_0^{\rm iso}(z, \boldsymbol{x}, \boldsymbol{y}) := \sqrt{N \eta} \big| \big(G(z) - M(z)\big)_{\boldsymbol{x} \boldsymbol{y}} \big|
\end{equation*}
and observe that
\begin{equation} \label{eq:single G}
	\Psi_0^{\rm av} + \Psi_0^{\rm iso} \prec 1
\end{equation}
is the usual single-resolvent local law from \eqref{eq:singlegllaw}, where here and in the following the arguments of $\Psi_k^{\rm av/iso}$ shall occasionally be omitted.
We remark that the index $k$ counts the number of regular matrices in the sense of Definition~\ref{def:regobs}. 

Throughout the entire argument, let $\epsilon > 0$ and $\kappa > 0$ be \emph{arbitrary} but fixed, and let 
\begin{equation} \label{eq:Omega}
	\mathbf{D}^{(\epsilon, \kappa)}:= \big\{  z \in \C :  \Re z \in \mathbf{B}_{\kappa}\,, \ N^{100} \ge |\Im z| \ge N^{-1+\epsilon}\big\}
\end{equation}
be the \emph{spectral domain}, where the $\kappa$-bulk $\mathbf{B}_\kappa$ has been introduced in \eqref{eq:bulk}.  Strictly speaking, we would need to define an entire (finite) \emph{family} of slightly enlarged spectral domains along which the above mentioned iterative scheme for proving Proposition \ref{pro:mresllaw} is conducted. Since this has been carried out in detail in \cite{iid} (see, in particular, \cite[Figure 2]{iid}), we will neglect this technicality and henceforth assume all bounds on $\Psi_k^{\rm av/iso}$ to be \emph{uniform} on $\mathbf{D}^{(\epsilon, \kappa)}$ in the following sense. 

\begin{definition}[Uniform bounds in the spectral domain] \label{def:epsi unif}  Let $\epsilon > 0$ and $\kappa > 0$ as above and let $k \in \N$. 
	We say that the bounds
	\begin{equation} \label{eq:epsi unif}
		\begin{split}
			\big\vert \langle G(z_1) B_1 \ \cdots \ G(z_k) B_k - M(z_1, B_1, ... , z_k) B_k \rangle \big\vert &\prec \mathcal{E}^{\rm av}\,, \\[2mm]
			\left\vert  \big( G(z_1) B_1 \ \cdots \ B_k G(z_{k+1}) - M(z_1, B_1, ... , B_k, z_{k+1})  \big)_{\boldsymbol{x} \boldsymbol{y}}  \right\vert &\prec \mathcal{E}^{\rm iso}
		\end{split}
	\end{equation}
	hold \emph{$(\epsilon, \kappa)$-uniformly} (or simply \emph{uniformly}) for some deterministic control parameters $\mathcal{E}^{\rm av/iso} = \mathcal{E}^{\rm av/iso}(N, \eta)$, depending only on $N$ and $ \eta:= \min_i |\Im z_i|$, if the implicit constant in \eqref{eq:epsi unif} are uniform in bounded deterministic matrices $\Vert B_j \Vert \le 1$, deterministic vectors $\Vert \boldsymbol{x} \Vert , \Vert \boldsymbol{y} \Vert \le 1$, and \emph{admissible} spectral parameters $z_j\in \mathbf{D}^{(\epsilon, \kappa)}$ satisfying $ 1 \ge \eta:= \min_j |\Im z_j|$. 
	
		Moreover, we may allow for additional restrictions on the deterministic matrices. For example, we may talk about uniformity under the additional assumption that some (or all) of the matrices are \emph{regular} (in the sense of Definition \ref{def:regobs}). 
\end{definition}

Note that \eqref{eq:epsi unif} is stated for a fixed choice of spectral parameters $z_j$ in the left hand side, but it is in fact equivalent to an apparently stronger statement, when the same bound holds with a supremum over the spectral parameters (with the same constraints). While one implication is trivial, the other direction follows from \eqref{eq:epsi unif} by a standard \emph{grid argument} (see, e.g., the discussion after \cite[Definition~3.1]{multiG}).

We can now formulate Proposition \ref{pro:mresllaw}, in the language of our basic control quantities $\Psi_k^{\rm av/iso}$. 
\begin{lemma}[Estimates on $\Psi^{\rm av/iso}_1$ and $\Psi^{\rm av/iso}_2$] \label{lem:multiGll}
	For any $\epsilon > 0$ and $\kappa > 0$ we have
	\begin{equation}\label{lem62}
		\Psi_1^{\rm av} + \Psi_1^{\rm iso} \prec 1 \qquad \text{and} \qquad 	\Psi_2^{\rm av}  + 	\Psi_2^{\rm iso}\prec \sqrt{N \eta}
	\end{equation}
	$(\epsilon, \kappa)$-uniformly in regular matrices.  
\end{lemma}
\begin{proof}[Proof of Proposition \ref{pro:mresllaw}] 
The $\eta\ge 1$ case was already explained right after Proposition \ref{pro:mresllaw}.  The more critical $\eta\le 1$ case 	immediately follows from Lemma \ref{lem:multiGll}. 
\end{proof}

\subsection{Master inequalities and reduction lemma: Proof of Lemma \ref{lem:multiGll}} \label{subsec:masterandred}
We now state the relevant part of a non-linear infinite hierarchy of coupled master inequalities 
for $\Psi^{\rm av}_k$ and $\Psi^{\rm iso}_k$. In fact, for our purposes, it is sufficient to have only the inequalities for $k \in [2]$. 
Slightly simplified versions of this master inequalities will be used in Appendix \ref{app:techlem} for general $k \in \N$.
 The proof of Proposition \ref{prop:master} is given in 
Section~\ref{subsec:proofmaster}
\begin{proposition}[Master inequalities, see Proposition 4.9 in \cite{iid}]  \label{prop:master} Assume that  for some deterministic
control parameters $ \psi_j^{\rm av/iso}$ we have that 
	\begin{equation} \label{eq:apriori Psi}
		\Psi_j^{\rm av/iso} \prec \psi_j^{\rm av/iso}\,, \quad j \in [4]\,,
	\end{equation}
holds	uniformly in regular matrices. Then we have
	\begin{subequations}
		\begin{align}
			\label{eq:masterineq Psi 1 av}
			\Psi_1^{\rm av} &\prec 1 + \frac{\psi_1^{\rm av}}{N \eta} + \frac{\psi_1^{\rm iso} +  (\psi_2^{\rm av})^{1/2}  }{(N \eta)^{1/2}} + \frac{(\psi_2^{\rm iso})^{1/2}}{(N \eta)^{1/4}}	\,, \\
			\label{eq:masterineq Psi 1 iso}
			\Psi_1^{\rm iso} &\prec 1 + \frac{\psi_1^{\rm iso} + \psi_1^{\rm av} }{(N \eta)^{1/2}} + \frac{(\psi_2^{\rm iso})^{1/2}}{(N \eta)^{1/4}} \,, \\
			\label{eq:masterineq Psi 2 av}
			\Psi_2^{\rm av} &\prec 1 + \frac{(\psi_1^{\rm av})^{2} + {(\psi_1^{\rm iso})^2} + \psi_2^{\rm av}}{N \eta}+ \frac{\psi_2^{\rm iso} + (\psi_4^{\rm av})^{1/2}  }{(N \eta)^{1/2}}  + \frac{(\psi_3^{\rm iso})^{1/2} + (\psi_4^{\rm iso})^{1/2}}{(N \eta)^{1/4}}    \,, \\
			\label{eq:masterineq Psi 2 iso}
			\Psi_2^{\rm iso} &\prec 1 + \psi_1^{\rm iso} + \frac{ \psi_1^{\rm av} \psi_1^{\rm iso} + (\psi_1^{\rm iso})^2 }{N \eta}+ \frac{\psi_2^{\rm iso} + (\psi_1^{\rm iso} \psi_3^{\rm iso})^{1/2}}{(N \eta)^{1/2}}  +   \frac{(\psi_3^{\rm iso})^{1/2} + \hspace{-2pt}(\psi_4^{\rm iso})^{1/2}}{(N \eta)^{1/4}}\,,
		\end{align}
	\end{subequations}
	again uniformly in regular matrices. 
\end{proposition}
As shown in the above proposition, resolvent chains of length $k=1,2$ are estimated by resolvent chains up to length $2k$. In order to avoid the indicated infinite hierarchy of master inequalities with higher and higher $k$ indices, we will need the following \emph{reduction lemma}.

\begin{lemma}[Reduction inequalities, see Lemma 4.10 in \cite{iid}] \label{lem:reduction} 
	As in~\eqref{eq:apriori Psi}, assume that 
	$\Psi_j^{\rm av/iso} \prec \psi_j^{\rm av/iso}$ holds for $1 \le j \le 4$ uniformly in regular matrices. Then we have
	\begin{equation} \label{eq:reduction av}
		\Psi_4^{\rm av} \prec (N\eta)^2 + (\psi_2^{\rm av})^2\,,
	\end{equation}
	uniformly in regular matrices, and
	\begin{equation} \label{eq:reduction iso}
		\begin{split}
			\Psi_3^{\rm iso} &\prec N \eta \left( 1+ \frac{\psi_2^{\rm iso}}{\sqrt{N\eta}} \right) \left( 1 + \frac{\psi_2^{\rm av}}{N\eta} \right)^{1/2}\,,\\
			\Psi_4^{\rm iso} &\prec (N \eta)^{3/2}\left( 1+ \frac{\psi_2^{\rm iso}}{\sqrt{N\eta}} \right) \left( 1 + \frac{\psi_2^{\rm av}}{N\eta} \right)\,
		\end{split}
	\end{equation}
	again uniformly in regular matrices. 
\end{lemma}
\begin{proof}
This is completely analogous to \cite[Lemma~4.10]{iid} and hence omitted. The principal idea is to write out the
 lhs.~of \eqref{eq:reduction av} and \eqref{eq:reduction iso} by spectral decomposition and tacitly employ a Schwarz inequality. 
 This leaves us with shortened chains, where certain resolvents $G$ are replaced with absolute values $|G|$, 
 which can be handled by means of a suitable integral representation \cite[Lemma~6.1]{iid}
\end{proof}
Now the estimates~\eqref{lem62} follow by combining Proposition \ref{prop:master} and Lemma \ref{lem:reduction}
 in an iterative scheme, which has been carried out in detail in \cite[Section~4.3]{iid}. This completes the proof of Lemma~\ref{lem:multiGll}.
 
\subsection{Proof of the master inequalities in Proposition \ref{prop:master}} \label{subsec:proofmaster} The proof of
 Proposition~\ref{prop:master} is very similar to the proof of the master inequalities in \cite[Proposition~4.9]{iid}. Therefore, 
we shall only elaborate on \eqref{eq:masterineq Psi 1 av} as a showcase 
in some detail and briefly discuss \eqref{eq:masterineq Psi 1 iso}--\eqref{eq:masterineq Psi 2 iso} afterwards.

First, we notice that \cite[Lemma~5.2]{iid} also holds for deformed Wigner matrices (see Lemma \ref{lem:underlined} below). In order to formulate it, recall the definition of the \emph{second order renormalisation}, denoted by underline, from \cite[Equation (5.3)]{iid}. For a function $f(W)$ of the Wigner matrix $W$, we define 
\begin{equation} \label{eq:underline}
	\underline{W f(W)} := W f(W) - \widetilde{\mathbb{E}} \big[   \widetilde{W} (\partial_{\widetilde{W}}f)(W) \big]\,,
\end{equation}
where $\partial_{\widetilde{W}}$ denotes the directional derivative in the direction of $\widetilde{W}$, which is a GUE matrix that is independent of $W$. The expectation is taken w.r.t.~the matrix $\widetilde{W}$. Note that, if $W$ itself a GUE matrix, then $\E \underline{Wf(W)} = 0$, while for $W$ with general single entry distributions, this expectation is independent of the first two moments of $W$. In other words, the underline renormalises the product $W f(W)$ to second order. 

We note that $\widetilde{\E} \widetilde{W} R \widetilde{W} = \langle R \rangle$ and furthermore, that the directional derivative of the resolvent is given by	$\partial_{\widetilde{W}} G = -G \widetilde{W} G$.
For example, in the special case $f(W) = (W + D -z)^{-1} = G$, we thus have
\begin{equation*}
	\underline{WG} = WG + \langle G \rangle G
\end{equation*}
by definition of the underline in \eqref{eq:underline}.

\begin{lemma} \label{lem:underlined} Under the assumption~\eqref{eq:apriori Psi}, 
for any regular matrix $A = \mathring{A}$ we have that 
\begin{equation}
	\langle (G-M)\mathring{A}\rangle = -\langle\underline{WG\mathring{A}'}\rangle+\mathcal{O}_{\prec}\left(\mathcal{E}_1^{\rm av}\right)\,,
	\label{eq:Psi1av_fullunderline}
\end{equation}
for some other regular matrix $A'=\mathring{A}'$, which linearly depends on $A$ (see \eqref{eq:A'} for an explicit formula). 
Here $G=G(z)$ and $\eta:=|\Im z|$.  For the error term we used the shorthand notation
\begin{equation}
	\mathcal{E}_1^{\rm av} :=\frac{1}{N\eta^{1/2}}\left(1+\frac{\psi_1^{\rm av}}{N\eta}\right)\,. 
\end{equation}
\end{lemma}
By simple complex conjugation of \eqref{eq:Psi1av_fullunderline}, we may henceforth assume that $z = e + \ii \eta$ with $\eta > 0$. The representation \eqref{eq:Psi1av_fullunderline} will be verified later. Now, using \eqref{eq:Psi1av_fullunderline} we compute the even moments of $\langle (G-M)\mathring{A}\rangle$ as
\begin{equation} \label{eq:G-M=underline}
    \E\left\vert\langle(G-M)A\rangle\right\vert^{2p}=\left\vert -\E \langle \underline{WG}A'\rangle\langle (G-M)A\rangle^{p-1}\langle (G-M)^*A^*\rangle^p\right\vert+\mathcal{O}_\prec \left(\left(\mathcal{E}_1^{\rm av}\right)^{2p}\right)
\end{equation}
and then apply a so-called \emph{cumulant expansion} to the first summand.
More precisely, we write out the averaged traces and employ an integration by parts (see, e.g., \cite[Eq.~(4.14)]{ETHpaper})
\begin{equation} \label{eq:cumex}
\E w_{ab} f(W) = \E |w_{ab}|^2 \E \partial_{w_{ba}}f(W) + ...  \quad \text{with} \quad \E |w_{ab}|^2 = \frac{1}{N}\,,
\end{equation}
indicating higher derivatives and an explicit error term, which can be made arbitrarily small, depending on the number of involved derivatives (see, e.g., \cite[Proposition 3.2]{slowcorr}). We note that, if $W$ were a GUE matrix, the relation \eqref{eq:cumex} would be exact without higher derivatives, which shall be discussed below. 

Considering the explicitly written \emph{Gaussian term} in \eqref{eq:cumex} for the main term in \eqref{eq:G-M=underline}, we find that it is bounded from above by (a $p$-dependent constant times)
\begin{equation}
\E\left[\frac{\vert \langle GGA'GA\rangle\vert + \vert\langle G^*GA'G^*A^*\rangle\vert}{N^2}\vert \langle (G-M)A\rangle\vert^{2p-2} \right]\,. 
     \label{eq:psi1av GC}
\end{equation}
The main technical tool to estimate 
\eqref{eq:psi1av GC} is the following contour integral representation for the square of resolvent (see \cite[Lemma~5.1]{iid}). This is given by
\begin{equation}
G(z)^2=\frac{1}{2\pi i}\int\limits_{\Gamma}\frac{G(\zeta)}{(\zeta-z)^2}\dif \zeta\,,
\label{eq:integr identity}
\end{equation}
where the contour $\Gamma=\Gamma(z)$ is the boundary of a finite disjoint union of half bands $J\times [\ii \tilde{\eta},\ii \infty)$
 which are parametrised counter-clockwise. Here $J$ is a finite disjoint union\footnote{Note that, for our concrete setting \eqref{eq:integr identity}, one closed interval (i.e.~half of Figure \ref{fig:decomp}) would be sufficient. However, we formulated it more generally here in order to ease the relevant modifications for the (omitted) proofs of \eqref{eq:masterineq Psi 1 iso}--\eqref{eq:masterineq Psi 2 iso}.} of closed intervals, which we take as $\mathbf{B}_{\kappa'}$ for a suitable 
 $ \kappa' \in (0, \kappa)$ -- to be chosen below -- and hence contains $e=\Re z$; the parameter $\Tilde{\eta}$ is 
 chosen to be smaller than, say, $\eta/2$. 
\begin{figure}[htbp]
	\centering
\includegraphics{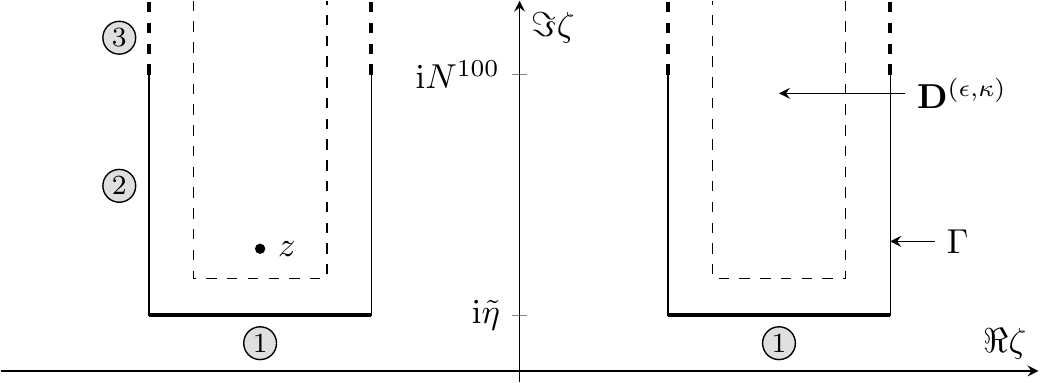}
	\caption{The contour $\Gamma$ is split into three parts (see \eqref{eq:contourdecomp}). Depicted is the situation, where the bulk $\mathbf{B}_{\kappa}$ consists of  two components. The boundary of the associated domain $\mathbf{D}^{(\epsilon, \kappa)}$ is indicated by the two U-shaped dashed lines. Modified version of \cite[Figure 4]{iid}.}
	\label{fig:decomp}
\end{figure}
Applying the integral identity \eqref{eq:integr identity} to the product $GG$ in the term $\langle GGA'GA\rangle$ yields that
\begin{equation}
\left\vert\langle GGA'GA\rangle\right\vert \lesssim \left\vert\int\limits_{\Gamma}\frac{\langle G(\zeta)A'GA\rangle}{(\zeta-z)^2}\dif\zeta\right\vert\,. 
\label{eq:psi1av GC1}
\end{equation}

Now, we split the contour $\Gamma$ in three parts, i.e.
\begin{equation} \label{eq:contourdecomp}
	\Gamma = \Gamma_1 + \Gamma_2 + \Gamma_3\,. 
\end{equation}
As depicted in Figure \ref{fig:decomp}, the first part, $\Gamma_1$, of the contour consists of the entire horizontal part of $\Gamma$. The second part, $\Gamma_2$, covers the vertical components up to $|\Im \zeta| \le N^{100}$. Finally, $\Gamma_3$ consists of the remaining part with $|\Im \zeta| > N^{100}$. 
The contribution coming from $\Gamma_3$ can be estimated with a trivial norm bound on $G$. In order to estimate the integral over $\Gamma_2$, we choose the parameter $\kappa'$ in the definition of $J = \mathbf{B}_{\kappa'}$ in such a way that the distance between $z$ and $\Gamma_2$ is greater than $\delta> 0$ from Definition \ref{def:regobs}. Hence, for $\zeta \in \Gamma_2$, every matrix is considered regular w.r.t.~$(z,\zeta)$. 

Therefore, after splitting the contour and estimating each contribution as just described, we find, with the aid of Lemma \ref{lem:Mbound}, 
\begin{equation*}
\left\vert\langle GGA'GA\rangle\right\vert \prec \left(1+\frac{\psi_2^{\rm av}}{N\eta}\right)+\int\limits_J\frac{\vert \langle G(x+\ii\Tilde{\eta})A'G(e+\ii \eta)A\rangle\vert}{(x-e)^2+\eta^2}\dif x \,.
\end{equation*}
 In this integral over $J= \mathbf{B}_{\kappa'}$, the horizontal part $\Gamma_1$, we decompose $A$ and $A'$ in accordance to the spectral parameters of the adjacent resolvents and use the following regularity property:
\begin{align*}
A&=\mathring{A}^{e+\ii \eta,e+\ii \eta}=\mathring{A}^{e+\ii \eta,x+\ii \tilde{\eta}}+\mathcal{O}(\vert x-e\vert+\eta) I \,, \\
A'&=\mathring{\left(A'\right)}^{e+\ii \eta,e+\ii \eta}=\mathring{\left(A'\right)}^{x+\ii \tilde{\eta},e+\ii \eta}+\mathcal{O}(\vert x-e\vert+\eta) I \,. 
\end{align*}
Now the integral over $J$ is represented as a sum of four integrals: one of them contains two regular matrices, the rest have at 
least one identity matrix with a small error factor.
 For the first one we use the same estimates as for the integral over the vertical part $\Gamma_2$. For the other terms, 
 thanks to the identity matrix, we use a resolvent identity, e.g.   for $G(x+i\Tilde{\eta})I G(e+i\eta)$ and
  note that the$\big(\vert x-e\vert+\eta\big)$-error improves the original $1/\eta$-blow up of $\int_J \frac{1}{(x-e)^2+\eta^2}\dif x$ to 
  an only $|\log \eta|$-divergent singularity, which is incorporated into `$\prec$'. 

For the term $\langle G^*GA'G^*A^*\rangle$ from \eqref{eq:psi1av GC}, we use a similar strategy: After an application of the Ward identity $G^* G = \Im G/\eta$, we decompose the deterministic matrices $A, A'$ according to the spectral parameters of their neighbouring resolvents in the product. This argument gives us that each of terms $\langle GGA'GA\rangle$ and $\langle G^*GA'G^*A^*\rangle$ is stochastically dominated by 
\begin{equation*}
    \frac{1}{\eta}\left(1+\frac{\psi_2^{\rm av}}{N\eta}\right).
\end{equation*}
Contributions stemming from higher order cumulants in \eqref{eq:psi1av GC} are estimated exactly in the same way as in \cite[Section~5.5]{iid}. The proof of~\eqref{eq:masterineq Psi 1 av} is
 concluded by applying Young inequalities to \eqref{eq:psi1av GC} (see \cite[Section~5.1]{iid}). 

Finally we show that \eqref{eq:Psi1av_fullunderline} holds:

\begin{proof}[Proof of Lemma \ref{lem:underlined}]
The proof of this representation is simpler than the proof of its analogue~\cite[Lemma~5.2]{iid} because
 in the current setting all terms in~\cite[Lemma~5.2]{iid} 
 containing the particular chiral symmetry matrix $E_{-}$ are absent. In the same way as in \cite{iid} we arrive at the identity
\begin{equation}
	\langle (G-M)A\rangle = -\langle \underline{WG} \mathcal{X}\left[A\right]M \rangle+\langle G-M\rangle \langle (G-M)\mathcal{X}[A]M \rangle\,,
	\label{eq:G1av underline1}
\end{equation}
where we introduced the bounded linear operator $\mathcal{X}[B] := \big(1 - \langle M \cdot M \rangle\big)^{-1}[B]$. Indeed, boundedness follows from the explicit formula
\begin{equation*}
\mathcal{X}[B] = B + \frac{\langle MBM\rangle}{1 - \langle M^2 \rangle}
\end{equation*}
by means of the lower bound
\begin{equation*}
\vert 1 - \langle M^2 \rangle \vert = \vert (1 - \langle M M^* \rangle) - 2\ii \langle M \Im M \rangle \vert   \ge 2 \langle \Im M \rangle^2 \gtrsim 1\,, 
\end{equation*}
obtained by taking the imaginary part of the MDE \eqref{eq:MDE}, in combination with $\Vert M \Vert \lesssim 1$.

 Next, completely analogously to \cite[Lemma~5.4]{iid}, we find the decomposition
\begin{equation} \label{eq:A'}
\mathcal{X}[A]M = \big(\mathcal{X}[A]M\big)^\circ + \mathcal{O}(\eta) I = \mathring{A}' + \mathcal{O}(\eta)I \,, 
\qquad A':=\mathcal{X}[A]M\,. 
\end{equation}
Plugging this into \eqref{eq:G1av underline1}, we thus infer
\begin{equation}
	\langle (G-M)A\rangle = -\langle \underline{WG}A'\rangle + \langle G-M\rangle \langle (G-M)\mathring{A}'\rangle + \left(-\langle \underline{WG}\rangle + \langle G-M\rangle^2\right) \mathcal{O}(\eta),
	\label{eq:G1av underline2}
\end{equation}
The second term in the rhs.~of \eqref{eq:G1av underline2} is obviously bounded by $\psi_1^{\rm av}/(N^2\eta^{3/2})$ and in the third term we use the usual local law \eqref{eq:singlegllaw} to estimate it by $\mathcal{O}_\prec\big(N^{-1}\big)$.
Combining these information gives~\eqref{eq:Psi1av_fullunderline}. 
\end{proof}

Notice that the above arguments leading to \eqref{eq:masterineq Psi 1 av} are completely identical to the ones 
required in the proof of the analogous master inequality in \cite[Proposition~4.9]{iid} with one minor
 but key modification: Every term involving the chiral symmetry matrix $E_-$ in \cite{iid} is simply absent and hence, 
with the notation of~\cite{iid},  all  sums over signs $\sum_{\sigma = \pm}\cdots$ collapse to a single summand with $E_+ \equiv I$. With this recipe, the proofs of \eqref{eq:masterineq Psi 1 iso}--\eqref{eq:masterineq Psi 2 iso} are completely analogous to the ones given in \cite[Sections 5.2--5.5]{iid} and hence omitted.

\appendix

\section{Additional technical Lemmas}
\label{app:techlem}
In this appendix, we prove several technical lemmas underlying the proofs our main results.

\subsection{Bounds on averaged multi-resolvent chains}
For the proof of Theorem \ref{theo:newCLT} we need to extend the key estimate 
$\big\vert \big\langle G(z_1) \mathring{A}_1 G(z_2) \mathring{A}_2 \big\rangle\big\vert \prec 1$ 
 for $\Re z_1, \Re z_2 \in \mathbf{B}_\kappa$ from \eqref{eq:2Gbddbulk}, which underlies the proof of the ETH in Theorem \ref{theo:ETH}, in two directions.
 First, we need to consider chains with an arbitrary number of resolvents in Lemma \ref{lem:addllaw},
but we will only need a quite weak suboptimal  bound which makes its proof quite direct and short. 
 Second, in Lemma \ref{lem:2Gllfaraway}, we no longer restrict $\Re z_1, \Re z_2 \in \mathbf{B}_\kappa$ to the bulk, but assume that $|\Re z_1 - \Re z_2| + |\Im z_1| + |\Im z_2| \ge \nu$ for some $N$-independent constant $\nu > 0$ and additionally allow for arbitrary (non-regular) matrices $A_1, A_2$. This second extension requires  Assumption~\ref{ass:Mbdd} on the deformation $D$, i.e.~the boundedness of $M(z)$ also for $\Re z \notin \mathbf{B}_\kappa$; this is a slightly stronger requirement than just the boundedness of $D$ 
 assumed in  Theorem \ref{theo:ETH}. 
 Both extensions are relevant for constructing the high probability event $\widehat{\Omega}$ in \eqref{eq:hatomega} for the DBM analysis. The proofs of Lemma~\ref{lem:addllaw} and 
 Lemma~\ref{lem:2Gllfaraway} are simple extensions and slight adjustments of 
 the arguments used in Proposition \ref{pro:mresllaw} and they will only be sketched.
\begin{lemma}
\label{lem:addllaw}
Fix $\epsilon>0$, $\kappa > 0$, $k\in \N$, and consider $z_1,\dots,z_{k} \in \C \setminus \R$ with $\Re z_j \in \mathbf{B}_\kappa$. Consider regular matrices $A_1,\dots,A_k$ with $\lVert A_i\rVert\le 1$, deterministic vectors ${\bm x}, {\bm y}$ with $\lVert {\bm x}\rVert +\lVert {\bm y} \rVert\lesssim 1$, and set $G_i:=G(z_i)$. Define
\begin{equation}
\mathcal{G}_k:=\widehat{G}_1A_1\dots A_{k-1}\widehat{G}_kA_k, \qquad \widehat{G}_j\in \{G_j,|G_j|\}\,. 
\end{equation}
Then, uniformly in $\eta:=\min_j |\Im z_j|\ge N^{-1+\epsilon}$, we have 
\begin{equation}
\label{eq:avellawlong}
\big|\langle \mathcal{G}_k \rangle\big|\prec \frac{N^{k/2-1}}{\sqrt{N\eta}}\,. 
\end{equation}
\end{lemma}
\begin{proof}

We only consider the case when $\mathcal{G}_k=G_1A_1\dots A_{k-1}G_kA_k$, 
the general case when some $G_j$ is replaced with $|G_j|$ is completely analogous and so omitted. 
To keep the notation short with a slight abuse of 
 notation we will often denote $(GA)^k=G_1A_1\dots A_{k-1}G_kA_k$.

 We split the proof into three steps.
In  Step (i)  we first prove the slightly weaker bound $\big| \langle \mathcal{G}_k \rangle\big|\prec N^{k/2-1}$ for 
any $k\ge 3$ and a similar bound in isotropic sense; then, using Step (i) as an input, we will prove the 
better estimate \eqref{eq:avellawlong} for $k=3,4$; finally we prove  \eqref{eq:avellawlong} for any $k\ge 3$.
\\[1mm]
\textbf{Step (i):} Similarly to the proof of the reduction inequalities in Lemma~\ref{lem:reduction} (see \cite[Lemma 4.10]{iid}) 
we readily see that for $k=2j$ (we omit the indices):
\begin{equation}
\label{eq:even}
\small
\langle (GA)^{2j}\rangle \lesssim N \begin{cases}
\langle |G|A(GA)^{j/2-1}|G|A(G^*A)^{j/2-1}\rangle^2 &\quad j\,\,\mathrm{even}, \\
\langle |G|A(GA)^{(j-1)/2}|G|A(G^*A)^{(j -1)/2}\rangle \langle |G|A(GA)^{(j-3)/2}|G|A(G^*A)^{(j-3)/2}\rangle &\quad j\,\,\mathrm{odd}, \\
\end{cases}
\end{equation}
and for $k=2j-1$:
\begin{equation}
\label{eq:odd}
\langle (GA)^{2j-1}\rangle \lesssim \langle|G|A(GA)^{j-2}|G|A(G^*A)^{j-2} \rangle^{1/2}\langle|G|A(GA)^{j-1}|G|A(G^*A)^{j-1} \rangle^{1/2}.
\end{equation}
We proceed by induction on the length of the chain. First we use \eqref{eq:even} for $j=2$, together with $\langle GAGA\rangle\prec 1$ 
from Proposition~\ref{pro:mresllaw}, to get the bound $\langle \mathcal{G}_4 \rangle\prec N$, and then 
use this bound as an input to obtain $\langle \mathcal{G}_3 \rangle\prec N^{1/2}$ using \eqref{eq:odd}. 
Then proceeding exactly in the same way we prove that if $\langle \mathcal{G}_l \rangle\prec N^{l/2-1}$ 
holds for any $l\le k$, then the same bound holds for chains of length $l=k+1$ and $k+2$ as well. Similarly,
 in the isotropic chains we prove $\langle{\bm x}, \mathcal{G}_k{\bm y} \rangle\prec N^{(k-1)/2}$; 
 this concludes Step (i).
\\[1mm]
\textbf{Step (ii):} Given the bounds $\langle \mathcal{G}_k \rangle\prec N^{k/2-1}$, 
$ \langle{\bm x}, \mathcal{G}_k{\bm y} \rangle\prec N^{(k-1)/2}$, the estimate in \eqref{eq:avellawlong} for $k=3,4$
 immediately follows by writing the equation for $\mathcal{G}_k$, performing cumulant expansion and
  using the corresponding bounds on $M(z_1, ... , z_k)$ from Lemma \ref{lem:Mbound}. 
  This was done in \cite[Proof of Proposition 3.5]{multiG}, hence we omit the details. 
\\[1mm]
\textbf{Step (iii):} The proof of \eqref{eq:avellawlong} for $k\ge 5$ proceeds by induction.
 We first show that it holds for $k=5,6$ and then we prove that if it holds for $k-2$ and $k-1$ 
 then it holds for $k$ and $k+1$ as well. 

By Step (ii), it follows that \eqref{eq:avellawlong} holds for $k=3,4$, and for $k=2$ we have $\langle GAGA\rangle \prec 1$. 
Then by \eqref{eq:even} we immediately conclude that the same bound is true for $k=6$, which together with \eqref{eq:odd} 
also imply the desired bound for $k=5$. The key point is that \eqref{eq:even} splits a longer $k$-chain ($k$ even) into a 
product  of shorter chains of length $k_1$, $k_2$ with $k_1+k_2=k$. As long as $k\ge 5$, at least one of the shorter chain
 has already length at least three, so we gain the  factor $(N\eta)^{-1/2}$. In fact chains of length $k=2$, from which we 
 do not gain any extra factor, $| \langle GAGA\rangle | \prec 1$, appear only once when we apply \eqref{eq:even} for $k=6$. 
 But in this case the other factor is a chain of length four with a gain of a $(N\eta)^{-1/2}$ factor. Similarly, \eqref{eq:odd} 
 splits the long $k$ chain ($k$ odd) into the square root of two chains of length $k-1$, and $k+1$, and for $k\ge 5$ we 
 have the $(N\eta)^{-1/2}$ factor from both. The induction step then readily follows by using again \eqref{eq:even}--\eqref{eq:odd} 
 as explained above; this concludes the proof. In fact, in most steps of the induction we gain more than one factor
  $(N\eta)^{-1/2}$; this would allow us to improve the bound \eqref{eq:avellawlong}, but for the purpose of the 
  present paper the suboptimal estimate \eqref{eq:avellawlong} is sufficient.
\end{proof}
We now turn to the second extension of \eqref{eq:2Gbddbulk} allowing for arbitrary spectral parameters, i.e.~not necessarily in the bulk, but separated by a safe distance $\nu > 0$. 
\begin{lemma} \label{lem:2Gllfaraway} Fix $\epsilon, \nu > 0$ and let the deformation $D \in \C^{N \times N}$ satisfy Assumption \ref{ass:Mbdd}. Let $z_1, z_2 \in \C \setminus \R$ be spectral parameters with $\Delta :=|\Re z_1 - \Re z_2| + |\Im z_1| + |\Im z_2| \ge \nu> 0$ and 
$B_1, B_2 \in \C^{N \times N}$ bounded deterministic matrices. 
Then, uniformly in $\eta := \min \big( |\Im z_1|, |\Im z_2|\big) \ge N^{-1+\epsilon}$, it holds that
	\begin{equation}
\big| \big\langle  G(z_1) B_1 G(z_2) B_2\big\rangle \big| \prec 1\,. 
	\end{equation}
\end{lemma}
\begin{proof} The proof is very similar to that of Proposition \ref{pro:mresllaw}, relying on a system of 
\emph{master inequalities} (Proposition \ref{prop:master}) complemented by the \emph{reduction inequalities} 
(Lemma \ref{lem:reduction}),  we just comment on the minor differences. 

Recall that the naive size, in averaged sense, of a chain
\begin{equation}\label{GBG1}
G_1B_1G_2B_2\ldots G_{k-1} B_{k-1} G_k
\end{equation}
 with $k$ resolvents
and arbitrary deterministic matrices  in between is of order $\eta^{-k+1}$; generically this is the size of the corresponding 
deterministic term in the usual multi-resolvent local law
(see \cite[Theorem 2.5]{multiG} with $a = 0$ for the case of Wigner matrices)
\begin{equation} \label{eq:Psikavisonaive}
	\begin{split}
|\langle  G_1 B_1 \cdots G_{k} B_k  -  M(z_1, B_1, ... , z_{k}) B_k  \rangle | \prec \frac{1}{N \eta^k} \\
 \left\vert  \big(  G_1 B_1 \cdots B_{k} G_{k+1} - M(z_1, B_1, ... , B_{k}, z_{k+1}) \big)_{\boldsymbol{x} \boldsymbol{y}} \right\vert \prec \frac{1}{\sqrt{N \eta} \, \eta^k}
	\end{split}
\end{equation}
with the customary short hand notations
\begin{equation*} 
	G_i:= G(z_i)\,, \quad \eta := \min_i |\Im z_i|\,, \quad \boldsymbol{z}_k :=(z_1, ... , z_{k})\,, \quad \boldsymbol{B}_k:=(B_1, ... , B_{k})\,.
\end{equation*}
In the following, we will consider every deterministic matrix $B_j$ together with its neighbouring resolvents, $G_jB_j G_{j+1}$, 
which we will call the \emph{unit} of $B_j$. 
Two units are called \emph{distinct} if they do not share a common resolvent
and we will count the number of such distinct units.\footnote{However, in the averaged case, one of the $k$ resolvents can be ``reused'' in this counting.} The main mechanism for the improvement over~\eqref{eq:Psikavisonaive} in Proposition \ref{pro:mresllaw} 
for \emph{regular} matrices was that  
for every distinct
unit $G_j B_j G_{j+1}$ in the initial resolvent chain with a \emph{regular} $B_j$, the naive size of  $M$ gets reduced by an $\eta$-factor, yielding the bound $\eta^{-\lfloor k/2\rfloor+1 }$ in \eqref{eq:Mboundtrace}
when all matrices are regular.\footnote{This improvement was also termed as the $\sqrt{\eta}$-rule, asserting that
 every regular matrix improves the $M$ bound and the error in the local law by a factor $\sqrt{\eta}$.
 This formulation is somewhat imprecise, the $M$ bound always involves integer $1/\eta$-powers;
 the correct counting is that each \emph{distinct} unit of regular matrices yields a factor $\eta$. However, the $\sqrt{\eta}$-rule applies
 to the error term.}  In the most relevant regime of small $\eta\sim N^{-1+\epsilon}$
this improvement in $M$ is (almost) matched by the corresponding 
 improvement in the error term, see  $N^{k/2-1}$ in~\eqref{eq:avellaw} (except that for odd $k$, the error is bigger 
 by an extra 
 $\eta^{-1/2}\sim N^{1/2}$). 
   
 The key point is that  if the spectral parameters $z_j$ and $z_{j+1}$ are "far away"
 in the sense that 
\begin{equation} \label{eq:Deltai}
\Delta_{j} := |\Re z_j- \Re z_{j+1}| + |\Im z_j| + |\Im z_{j+1}| \ge \nu > 0\,,
 \end{equation} 
 then \emph{any} matrix $B_j$ in the chain $\ldots G_j B_j G_{j+1}\ldots $
behaves as if it were regular. The reason is that the corresponding stability operator  $\mathcal{B}_{j,j+1}$ from 
\eqref{eq:stabop} (explicitly given in \eqref{eq:B-1explicit} and \eqref{eq:betalowerbound} below) has no singular direction,
its inverse  is bounded, i.e.
\begin{equation*} 
\Vert \mathcal{B}_{j,j+1}^{-1}[R] \Vert \lesssim \Vert R \Vert \quad \text{for all} \quad j \in [k]\,, \quad R \in \C^{N \times N}\,. 
\end{equation*}
For example, using the definition~\eqref{eq:M_definitionapp}, we have
 \begin{equation}\label{Mb}
  \| M(z_1, B, z_2)\|\lesssim 1
 \end{equation}
 whenever $\Delta_{12}\ge \nu$, hence $\mathcal{B}_{12}^{-1}$ is bounded. 
Therefore, when mimicking the proof of Proposition \ref{pro:mresllaw}, 
instead of counting regular matrices with distinct units,
we need to count the distinct units within the chain~\eqref{GBG1} for which the corresponding spectral parameters are far away;
their overall effects are the same -- modulo a minor  difference, that now the errors for odd $k$ do not get increased 
by $\eta^{-1/2}$ when compared to the $M$-bound (see later). 

To be  more precise, in our new setup 
we introduce the 
modified\footnote{Notice that for odd $k$, the $\eta$-power in the prefactor is slightly different from
those in~\eqref{eq:Psi avk} and \eqref{eq:Psi isok}.} normalised differences 
\begin{align} \label{eq:Psitilde avk}
	\widetilde{\Psi}_k^{\rm av}(\boldsymbol{z}_k, \boldsymbol{B}_k) &:= N \eta^{\lfloor k/2 \rfloor} |\langle  G_1 B_1 \cdots G_{k} B_k 
	 -  M(z_1, B_1, ... , z_{k}) B_k  \rangle |\,, \\
	\label{eq:Psitilde isok}
	\widetilde{\Psi}_k^{\rm iso}(\boldsymbol{z}_{k+1}, \boldsymbol{B}_{k}, \boldsymbol{x}, \boldsymbol{y}) &:= \sqrt{N \eta} \, \eta^{\lfloor k/2 \rfloor} \left\vert  \big(  G_1 B_1 \cdots B_{k} G_{k+1} - M(z_1, B_1, ... , B_{k}, z_{k+1}) \big)_{\boldsymbol{x} \boldsymbol{y}} \right\vert
\end{align}
for $k \in \N$ 	as a new set of basic control quantities (cf.~\eqref{eq:Psi avk} and \eqref{eq:Psi isok}). 
The deterministic counterparts $M$ used in \eqref{eq:Psitilde avk} and \eqref{eq:Psitilde isok} are again given recursively in
  Definition \ref{def:Mdef}. Contrary to \eqref{eq:Psi avk} and \eqref{eq:Psi isok}, the deterministic matrices 
  $\Vert B_j \Vert \le 1$, $j \in [k]$, are \emph{not} assumed to be regular.
   This ``lack of regularity'' is compensated by the requirement that consecutive spectral parameters $z_j, z_{j+1}$ of the unit
    of $B_j$ satisfy~\eqref{eq:Deltai}.
 Just as in Definition \ref{def:regobs}, in case of \eqref{eq:Psitilde avk}, the indices in \eqref{eq:Deltai} are
  understood cyclically modulo $k$. Chains satisfying \eqref{eq:Deltai} for all $j \in [k]$ are called \emph{good}. Hence, in a good
  chain one can potentially gain a factor $\eta$ from every unit $G_j B_j G_{j+1}$. Therefore, analogous to the regularity requirement for \emph{all} deterministic matrices in Definition \ref{def:epsi unif}, the normalised differences $\widetilde{\Psi}_k^{\rm av/iso}$ in \eqref{eq:Psitilde avk} and \eqref{eq:Psitilde isok} will only be used for good chains. 
 
 As already indicated above, the analogy between our new setup and the setup of Proposition \ref{pro:mresllaw} is not perfect due to the following reason: For $k=1$ the error bounds in~\eqref{eq:Psikavisonaive}
 improve by $\sqrt{\eta}$ for $B_1$ being a regular matrix, but for $\Delta_1 \ge \nu > 0$, the improvement is by a full power of $\eta$.\footnote{For the averaged case, this improvement is really artificial, since the $\Delta_1 \ge \nu$-requirement means that $\eta \gtrsim 1$. }
 This discrepancy causes slightly different $\eta$-powers for odd $k$ in all estimates (cf.~\eqref{eq:Psitilde avk} and \eqref{eq:Psitilde isok}).

We now claim  that for good chains, the requirement \eqref{eq:Deltai} for all $j\in [k]$ reduces the naive sizes
of the errors in the usual multi-resolvent local laws~\eqref{eq:Psikavisonaive} at least 
by a factor $\eta^{\lceil k/2 \rceil}$ for $k=1,2$.  
Previously, in the proof of Proposition \ref{pro:mresllaw}, these sizes got reduced by a factor $\sqrt{\eta}$ for every matrix $B_j$ which was regular in the sense of Definition \ref{def:regobs}. Now, compared to this \emph{regularity gain}, the main effect for our new, say, \emph{$\nu$-gain} for good chains is that for every $j \in [k]$, the inverse of the stability operator \eqref{eq:stabop} (explicitly given in \eqref{eq:B-1explicit} and \eqref{eq:betalowerbound} below) is bounded, i.e.
\begin{equation} \label{eq:stabopbounded}
\Vert \mathcal{B}_{j,j+1}^{-1}[R] \Vert \lesssim \Vert R \Vert \quad \text{for all} \quad j \in [k]\,, \quad R \in \C^{N \times N}\,. 
\end{equation}

Armed with \eqref{eq:stabopbounded}, completely analogously to Proposition \ref{pro:mresllaw}, one then starts a proof of the master inequalities (similar to those in Proposition \ref{prop:master}): First, one establishes suitable \emph{underlined lemmas} (cf.~Lemma~\ref{lem:underlined} and \cite[Lemmas 5.2, 5.6, 5.8, and 5.9]{iid}), where now no \emph{splitting} 
of observables into singular and regular parts (see, e.g., \eqref{eq:A'} and \cite[Equation (5.35)]{iid}) is necessary,
 since the bounded matrices $B_j$ are arbitrary. Afterwards, the proof proceeds by cumulant expansion (see \eqref{eq:psi1av GC}), where resolvent chains of length $k$ are estimated by resolvent chains of length up to $2k$. 
 This potentially infinite hierarchy is truncated by suitable reduction inequalities, 
as in Lemma~\ref{lem:reduction}.  
Along this procedure, 
we also create non-good chains, but  a direct inspection\footnote{We spare the reader from presenting 
the case by case checking for the new setup,
 but we point out that this is doable since Lemma~\ref{lem:2Gllfaraway}, as well as  Proposition \ref{pro:mresllaw}, 
concern chains of length at most $k\le 2$. Extending these local laws for general $k$ is possible, but it would require a more
systematic power-counting of good chains.} shows that 
there are always sufficiently many good chains left that
provide the necessary improvements,  exactly  as   in the proof of Proposition \ref{pro:mresllaw}.

Just to indicate this mechanism, consider,
 for example, the Gaussian term appearing in the cumulant expansion of \eqref{eq:Psitilde avk} for $k=2$, analogous to \eqref{eq:psi1av GC}.
 In this case we encounter  the following term with five resolvents that we immediately estimate in terms of chains with four resolvents: 
 \begin{equation} \label{eq:exampleterm}
\frac{\vert \langle G_2 B_2 G_1 B_1 G_2 G_1 B_1 G_2 B_2 \rangle\vert }{N^2} \prec \frac{\vert \langle G_2 B_2 G_1 B_1 G_2 B_1 G_2 B_2 \rangle\vert + \vert \langle G_2 B_2 G_1 B_1 G_1 B_1 G_2 B_2 \rangle\vert }{N^2}\,.
 \end{equation}
Here we used that $\Delta =|\Re z_1 - \Re z_2| + |\Im z_1| + |\Im z_2| \ge \nu> 0$  to reduce  $G_2 G_1$ 
to a single $G$ term. Strictly speaking,  the estimate~\eqref{eq:exampleterm} directly follows from the
resolvent identity $G_2G_1= (G_2-G_1)/(z_1-z_2)$
only when $|z_1-z_2| \sim |\Re z_1 - \Re z_2| + |\Im z_1-\Im z_2|\gtrsim \nu$; this latter
condition  follows from  $\Delta\ge \nu$
only if $\Im z_1\cdot \Im z_2 <0$. In the remaining case, when $z_1\approx z_2$ but both 
with a large imaginary part (since $\Delta\ge \nu$), we can use an appropriate contour integral representation
$$
     G(z_2)G(z_1) =\frac{1}{2\pi \ii} \int_\Gamma  \frac{G(\zeta)}{(\zeta-z_1)(\zeta-z_2)} {\rm d}\zeta
$$
similar to~\eqref{eq:integr identity} with a contour well separated from $z_1, z_2$.  
Hence  we obtain a four-resolvent chain $\langle G_2 B_2 G_1 B_1 G(\zeta) B_1 G_2 B_2 \rangle$
on the rhs.~of~\eqref{eq:exampleterm}, where 
the spectral parameter $\zeta$ of the
new $G(\zeta)$ resolvent is "far away" from the other spectral parameters, and it can be treated 
 as  $\langle G_2 B_2 G_1 B_1 G_j B_1 G_2 B_2 \rangle$, $j=1,2$.
In fact, in our concrete application of Lemma~\ref{lem:2Gllfaraway} we always know
that not only $\Delta\ge \nu$, but already $|\Re z_1 - \Re z_2|\ge \nu$, hence the  argument
with the resolvent identity is always sufficient.

Note that the two chains on 
the rhs.~of \eqref{eq:exampleterm} cannot  be directly cast in the form \eqref{eq:Psitilde avk}, since 
not every unit has well separated spectral parameters (e.g. we have $G_2B_1G_2$), hence these chains are not good. 
However, after application of a reduction inequality (see \eqref{eq:even}), we find that 
\begin{equation*}
\vert \langle G_2 B_2 G_1 B_1 G_2 B_1 G_2 B_2 \rangle\vert \prec N \big( \langle |G_2| B_2 |G_1| B_2^* \rangle \langle |G_1|B_1 |G_2|B_1^* \rangle \langle |G_2| B_1 |G_2| B_1^*\rangle  \langle |G_2| B_2 |G_2| B_2^*\rangle\big)^{1/2}
\end{equation*}
and analogously for the second summand in \eqref{eq:exampleterm}. Estimating the two shorter  non-good 
chains involving only $G_2$ by $1/\eta$ via a trivial Schwarz inequality, this yields that
\begin{equation*}
\frac{\vert \langle G_2 B_2 G_1 B_1 G_2 G_1 B_1 G_2 B_2 \rangle\vert }{N^2} \prec \frac{ \big( \langle |G_2| B_2 |G_1| B_2^* \rangle \langle |G_1|B_1 |G_2|B_1^* \rangle \big)^{1/2}}{N\eta}\,,
\end{equation*}
where the remaining shorter chains are good and can be estimated in terms of $\widetilde{\Psi}_2^{\rm av}$.
 A similar mechanism works for any other term. This completes the discussion of the discrepancies between the current setup and Proposition \ref{pro:mresllaw}.

 Notice that this argument always assumes that we have a single resolvent local 
law and that $M$'s are bounded.  At potential cusps in the scDos $\rho$ we do not have a single resolvent local law 
 (see the discussion below \eqref{eq:singlegllaw}) and  the estimates
  on $\Vert M(z)\Vert$ for $\Re z$ close to edges (and cusps) of $\rho$ may deteriorate for general deformation $D$. 
  However, these two phenomena are simply excluded by Assumption \ref{ass:Mbdd} 
  on the deformation $D$ (see also Remark \ref{rmk:Mbdd}). 
In particular, this assumption allows us to show  exactly
 same estimates on $M(z_1, ... , z_k)$ as given in Lemma~\ref{lem:Mbound}, 
 which serve as an input in the proof sketched above. 

To conclude, similarly to Proposition \ref{pro:mresllaw}, our method again shows that
\begin{equation*}
\big| \big\langle  \big(G(z_1) B_1 G(z_2)  - M(z_1, B_1, z_2)\big)B_2\big\rangle \big| \prec \frac{1}{(N \eta)^{1/2}}
\end{equation*}
which, together with the corresponding bound~\eqref{Mb}, 
 immediately yields the desired bound and completes the proof of Lemma~\ref{lem:2Gllfaraway}.
\end{proof}

\subsection{Proof of Lemma \ref{lem:Mbound}}
	The proof is completely analogous to the proof of Lemma 4.3 from \cite{iid}, hence we only show how the three main technical aspects of the latter should be adjusted to our setup of deformed Wigner matrices.  In general, the setup of~\cite{iid} is more
	complicated due to the chiral symmetry which involves summations over signs $\sigma=\pm$.
	As a rule of thumb, we can obtain the necessary $M$-formulas for our current case just by mechanically using
	the corresponding formulas in~\cite{iid} and drop the $\sigma=-1$ terms. 
	\\[1mm]
\underline{Recursive Relations:} The principal idea is to derive several different recursive relations for $M(z_1, ... , z_k)$ 
(which itself is defined by one of those in Definition \ref{def:Mdef})
 by a so-called \emph{meta argument} \cite{metaargument, iid}.  These alternative recursions can then be employed to prove Lemma \ref{lem:Mbound} iteratively in the number of spectral parameters. 
 These recursive relations are identical to those in Lemma D.1
 of the \href{https://arxiv.org/abs/2301.03549}{arXiv: 2301.03549} version of~\cite{iid} 
 when dropping the $\sigma = -1$ terms in Eqs. (D.1) and (D.2) therein
  and writing the $N \times N$ identity instead of $E_+$.  
\\[1mm]
\underline{Stability Operator:} The inverse of the stability operator \eqref{eq:stabop} can be expressed in the following explicit form
	\begin{equation} \label{eq:B-1explicit}
		\mathcal{B}_{12}^{-1}[R]=R+\frac{\langle R \rangle}{1-\langle M_1M_2\rangle}M_1 M_2 = R + \frac{1}{\beta_{12}}  \langle R \rangle M_1 M_2 \,,
	\end{equation}
where $\beta_{12} := 1 - \langle M_1 M_2 \rangle$ is the only non-trivial eigenvalue of $\mathcal{B}_{12}$. 
Completely analogously to \cite[Lemma~B.2~(b)]{iid}, it holds that 
\begin{equation} \label{eq:betalowerbound}
|\beta_{12}| \gtrsim\big( |\Re z_1 - \Re z_2| + |\Im z_1| + |\Im z_2|\big) \wedge 1\,,
\end{equation}
 which, in combination with \eqref{eq:B-1explicit}, in particular  implies \eqref{eq:Mboundnorm} for $k=1$ and \eqref{eq:Mboundtrace} for $k=2$.
	\\[1mm]
	\underline{Longer Chains:}	In order to prove \eqref{eq:Mboundnorm} for $k=3$, similarly to \cite{iid} we verify 
	at first \eqref{eq:Mboundnorm} for $k=2$ in the case when exactly one observable is regular. 
	For this purpose we again use the recursive relation of the form Eq.~(D.12)
	of the \href{https://arxiv.org/abs/2301.03549}{arXiv: 2301.03549} version of~\cite{iid}:
	\begin{equation*}
		M(z_1,A_1,z_2,A_2,z_3)
		=M(z_1,\mathcal{X}_{12}[A_1]M_2A_2,z_3)+M(z_1,\mathcal{X}_{12}[A_1]M_2,z_3)\langle M(z_2,A_2,z_3)\rangle\,,
	\end{equation*}
where we denoted the linear operator $\mathcal{X}_{mn}$ as
\begin{equation} \label{eq:Xmn}
\mathcal{X}_{mn}[R] := \big( 1 - \langle M_m \cdot M_n\rangle \big)^{-1}[R] \qquad \text{for} \quad R \in \C^{N \times N}\,. 
\end{equation}
	Now, similarly to the arguments around Eq.~(D.13)
	of the \href{https://arxiv.org/abs/2301.03549}{arXiv: 2301.03549} version of~\cite{iid}, 
	we observe a balancing cancellation in the last term, which comes from the continuity with 
	respect to one of spectral parameters of the regular part of a deterministic matrix when another 
	spectral parameter is fixed (see \eqref{eq:lipprop}). \qed

\end{document}